%% file: arXiv.tex
\documentclass[a4paper,UKenglish,cleveref, autoref]{lipics-v2019}
\title{On Computing the Hamiltonian Index of Graphs} 
\pdfoutput=1
\nolinenumbers

\usepackage[absolute]{textpos}

\author{Geevarghese Philip}{Chennai Mathematical Institute, Chennai, India and
  UMI ReLaX}{gphilip@cmi.ac.in}{https://orcid.org/0000-0003-0717-7303}{%
  The European Research Council (ERC) under the European Union's Horizon 2020
  research and innovation programme (grant agreement No 819416), and the
  Norwegian Research Council via grants MULTIVAL and CLASSIS.\\
}

\author{Rani M. R.}{National Institute of Technology, Calicut, India}{rani\_p150067cs@nitc.ac.in}{https://orcid.org/0000-0002-4918-8150}{}

\author{Subashini R.}{National Institute of Technology, Calicut, India}{suba@nitc.ac.in}{https://orcid.org/0000-0002-9724-3484}{}

\authorrunning{G. Philip, Rani M. R, and Subashini R.}

\Copyright{Geevarghese Philip, Rani M. R and Subashini R.}



\input{preamble}

\begin{document}

\maketitle
\begin{textblock}{20}(15.3, 13.5)
\includegraphics[scale=0.20]{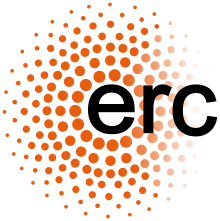}
\end{textblock}
\begin{textblock}{20}(15.3, 12.7)
\includegraphics[scale=0.05]{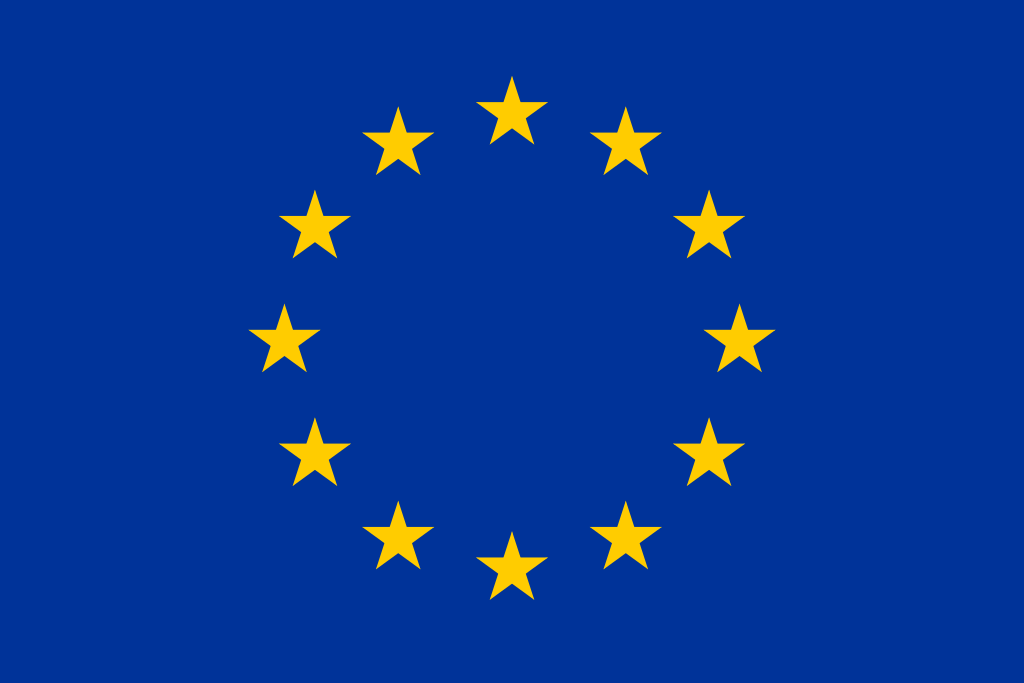}
\end{textblock}
\input{abstract}
\pagebreak
\input{introduction}
\input{preliminaries}
\input{ess}

\input{hamIndex}
\input{conclusion}
\newpage
\bibliographystyle{plainurl}
\bibliography{arXiv}\label{references}
\appendix
\input{des}

\end{document}

%% file: preamble.tex
\usepackage{xspace}

\usepackage[T1]{fontenc}
\usepackage[OT1]{eulervm}
\usepackage{mathtools}
\usepackage{amssymb}
\usepackage{cite}

\usepackage{todonotes}

\newaliascnt{observation}{theorem}
\newtheorem{observation}[observation]{Observation}
\aliascntresetthe{observation}

\theoremstyle{plain}
\newtheorem*{theorem*}{Theorem}
\newaliascnt{fact}{theorem}

\aliascntresetthe{fact}

\newcommand{\name}[1]{\textsc{#1}}

\newcommand{\upcite}[1]{\textup{\textrm{\cite{#1}}}}

\newcommand{\Oh}[1]{\ensuremath{\mathcal{O}(#1)}}
\newcommand{\OhStar}[1]{\ensuremath{\mathcal{O}^{\star}(#1)}\xspace}

\newcommand{\yes}{\textbf{yes}\xspace}
\newcommand{\no}{\textbf{no}\xspace}

\newcommand{\HI}{\name{Hamiltonian Index}\xspace}
\newcommand{\HIs}{\name{HI}\xspace} 
\newcommand{\ST}{\name{Steiner Tree}\xspace}
\newcommand{\SES}{\name{Spanning Eulerian Subgraph}\xspace}
\newcommand{\SESs}{\name{SES}\xspace} 
\newcommand{\ESS}{\name{Eulerian Steiner Subgraph}\xspace}
\newcommand{\ESSs}{\name{ESS}\xspace} 
\newcommand{\DES}{\name{Dominating Eulerian Subgraph}\xspace}
\newcommand{\DESs}{\name{DES}\xspace} 
\newcommand{\EHC}{\name{Edge Hamiltonian Cycle}\xspace}
\newcommand{\EHCs}{\name{EHC}\xspace} 
\newcommand{\EHP}{\name{Edge Hamiltonian Path}\xspace}
\newcommand{\EHPs}{\name{EHP}\xspace} 

\newcommand{\FPT}{\ensuremath{\mathsf{FPT}}\xspace}

\newcommand{\NPH}{\ensuremath{\mathsf{NP}}-hard\xspace}
\newcommand{\NPC}{\ensuremath{\mathsf{NP}}-complete\xspace}

\newcommand{\ILG}[2]{\ensuremath{L^{#1}(#2)}\xspace} 

\newcommand{\vstar}{\ensuremath{v^{\star}}\xspace}

\newcommand{\TT}{\ensuremath{\mathcal{T}}\xspace}

\newcommand{\LL}{\ensuremath{\mathcal{L}}\xspace}

\newcommand{\CC}{\ensuremath{\mathcal{C}}\xspace}

\newcommand{\PP}{\ensuremath{\mathcal{P}}\xspace}

\newcommand{\calA}{\ensuremath{\mathcal{A}}\xspace}
\newcommand{\calB}{\ensuremath{\mathcal{B}}\xspace}

\newcommand{\defparproblem}[4]{
  \vspace{1mm}
\noindent\fbox{
  \begin{minipage}{0.96\textwidth}
  \begin{tabular*}{\textwidth}{@{\extracolsep{\fill}}lr} \textsc{#1}  & {\bf{Parameter:}} #3
\\ \end{tabular*}
  {\bf{Input:}} #2  \\
  {\bf{Question:}} #4
  \end{minipage}
  }
  \vspace{1mm}
}

\addto\extrasUKenglish{}


%% file: abstract.tex
\begin{abstract}
  For an integer \(r \geq 0\) the \(\mathit{r}\)\emph{-th iterated line graph}
  \ILG{r}{G} of a graph \(G\) is defined by: (i) \(\ILG{0}{G} = G\) and (ii)
  \(\ILG{r}{G} = L(\ILG{(r- 1)}{G})\) for \(r > 0\), where \(L(G)\) denotes the
  line graph of \(G\). The \emph{Hamiltonian Index} \(h(G)\) of \(G\) is the
  smallest \(r\) such that \ILG{r}{G} has a Hamiltonian cycle [Chartrand, 1968].
  Checking if \(h(G) = k\) is \NPH for any fixed integer \(k \geq 0\) even for
  subcubic graphs \(G\) [Ryj\'{a}\v{c}ek et al., 2011]. We study the
  parameterized complexity of this problem with the parameter treewidth,
  \(tw(G)\), and show that we can find \(h(G)\) in time\footnote{The \OhStar{}
    notation hides polynomial factors in input size.}
  \(\OhStar{(1 + 2^{(\omega + 3)})^{tw(G)}}\) where \(\omega\) is the matrix
  multiplication exponent. Prior work on computing \(h(G)\) includes various
  \(\OhStar{2^{\Oh{tw(G)}}}\)-time algorithms for checking if \(h(G) = 0\)
  holds; i.e., whether \(G\) has a Hamiltonian Cycle~[Cygan et al., FOCS 2011;
  Bodlaender et al., Inform. Comput., 2015; Fomin et al., JACM 2016]; an
  \(\OhStar{tw(G)^{\Oh{tw(G)}}}\)-time algorithm for checking if \(h(G) = 1\)
  holds; i.e., whether \(L(G)\) has a Hamiltonian Cycle~[Lampis et al., Discrete
  Appl. Math., 2017]; and, most recently, an
  \(\OhStar{(1 + 2^{(\omega + 3)})^{tw(G)}}\)-time algorithm for checking if
  \(h(G) = 1\) holds~[Misra et al., CSR 2019]. Our algorithm for computing
  \(h(G)\) generalizes these results.
  
  The \NPH \ESS problem takes as input a graph \(G\) and a specified subset
  \(K\) of \emph{terminal} vertices of \(G\) and asks if \(G\) has an
  Eulerian\footnote{That is: connected, and with all vertices of even degree.}
  subgraph \(H\) containing all the terminals. A key ingredient of our
  algorithm for finding \(h(G)\) is an algorithm which solves \ESS in
  \(\OhStar{(1 + 2^{(\omega + 3)})^{tw(G)}}\) time. To the best of our knowledge
  this is the first \FPT algorithm for \ESS. Prior work on the special case of
  finding a \emph{spanning} Eulerian subgraph (i.e., with \(K = V(G)\)) includes
  a polynomial-time algorithm for series-parallel graphs [Richey et al., 1985]
  and an \(\OhStar{2^{\Oh{\sqrt{n}}}}\)-time algorithm for planar graphs on
  \(n\) vertices~[Sau and Thilikos, 2010]. Our algorithm for \ESS generalizes
  both these results.
\end{abstract}


%% file: introduction.tex
\section{Introduction}\label{intro}
All graphs in this article are finite and undirected, and are without self-loops
or multiple edges unless explicitly stated. We use \(\mathbb{N}\) to denote the
set of non-negative integers, and \(V(G), E(G)\), respectively, to denote the
vertex and edge sets of graph \(G\). A graph is \emph{Eulerian} if it has a
closed Eulerian trail, and \emph{Hamiltonian} if it has a Hamiltonian
cycle\footnote{See \autoref{sec:preliminaries} for definitions.}. The vertex set
of the \emph{line graph} of a graph \(G\)---denoted \(L(G)\)---is the edge set
\(E(G)\) of \(G\), and two vertices \(e,f\) are adjacent in \(L(G)\) if and only
if the edges \(e\) and \(f\) share a vertex in \(G\). A graph \(H\) is said to
be a line graph if there exists a graph \(G\) such that \(L(G) = H\). Line
graphs are an extremely well studied class of graphs; we recall a few well-known
properties (See,
e.g.:~\cite{balakrishnan2000textbook,west2001introduction,catlin1990hamilton,Prisner1996,chartrand1964phdthesis}).
The line graph operation is \emph{almost} injective: if \(H\) is a line graph
then there is a unique graph \(G\) such that \(L(G) = H\), \emph{except} when
\(H\) is the triangle \(C_{3}\), in which case \(G\) can either be \(C_{3}\) or
the star \(K_{1,3}\) with three leaves. A graph \(G\) is connected if and only
if its line graph \(L(G)\) is connected\footnote{We deem the empty graph---with
  no vertices---to be connected.}. Let \(P_{\ell}\) (respectively, \(C_{\ell}\))
denote the path (resp. cycle) with \(\ell\) edges. Then for any \(\ell \geq 1\)
we have \(L(P_{\ell}) = P_{\ell - 1}\), and for any \(\ell \geq 3\) we have
\(L(C_{\ell}) = C_{\ell}\). More generally, for any connected graph \(G\) which
is \emph{not} a path we have that \(L(G)\) is connected and has \emph{at least
  as many edges} as \(G\) itself. This implies that starting with a non-empty
connected graph \(G\) which is not a path and repeatedly applying the line graph
operation will never lead to the empty graph. More precisely: Let \(r\) be a
non-negative integer. The \(\mathit{r}\)\emph{-th iterated line graph}
\ILG{r}{G} of \(G\) is defined by: (i) \(\ILG{0}{G} = G\), and (ii)
\(\ILG{r}{G} = L(\ILG{(r- 1)}{G})\) for \(r > 0\). If \(G=P_{\ell}\) for a
non-negative integer \(\ell\) then \(\ILG{\ell}{G}\) is \(K_{1}\), the graph
with one vertex and no edges, and \(\ILG{r}{G}\) is the empty graph for all
\(r>\ell\). If \(G\) is a connected graph which is \emph{not} a path then
\(\ILG{r}{G}\) is nonempty for \emph{all}
\(r \geq 0\)~\cite{catlin1990hamilton}.

It was noticed early on that the operation of taking line graphs has interesting
effects on the properties of the (line) graph being Eulerian or Hamiltonian. For
instance, Chartrand~\cite{chartrand1964phdthesis} observed that: (i) if \(G\) is
Eulerian, then \(L(G)\) is Eulerian; (ii) if \(G\) is Eulerian, then \(L(G)\) is
\emph{Hamiltonian}; and (iii) if \(G\) is Hamiltonian, then \(L(G)\) is
Hamiltonian, and that the converse does not (always) hold in each case. Another
example, again due to Chartrand~\cite{chartrand1968hamiltonian}: If \(G\) is a
connected graph which is \emph{not} a path then \emph{exactly one} of the
following holds: (i) \(G\) is Eulerian, (ii) \(G\) is \emph{not} Eulerian, but
\(L(G)\) is Eulerian, (iii) neither \(G\) nor \(L(G)\) is Eulerian, but
\(\ILG{2}{G}\) is Eulerian, and (iv) there is \emph{no} integer \(r \geq 0\)
such that \(\ILG{r}{G}\) is Eulerian. For a third example we look at two
characterizations of graphs \(G\) whose line graphs \(L(G)\) are Hamiltonian. An
\emph{edge Hamiltonian path} of a graph \(G\) is any permutation \(\Pi\) of the
edge set \(E(G)\) of \(G\) such that every pair of consecutive edges in \(\Pi\)
has a vertex in common, and an \emph{edge Hamiltonian cycle} of \(G\) is an edge
Hamiltonian path of \(G\) in which the first and last edges also have a vertex
in common.
\begin{theorem}\label{fac:line_graph_hamiltonian_characterization}
  The following are equivalent for a graph \(G\):
  \begin{itemize}
  \item Its line graph \(L(G)\) is Hamiltonian 
  \item \(G\) has an edge Hamiltonian cycle~\upcite{chartrand1968hamiltonian}
  \item \(G\) contains a closed trail \(T\) such that every edge in \(G\) has at
    least one end-point in \(T\)~\upcite{harary1965eulerian}
  \end{itemize}
\end{theorem}
Given these results a natural question would be: what are the graphs \(G\) such
that \(\ILG{r}{G}\) is Hamiltonian for \emph{some} integer \(r \geq 0\)?
Chartrand found the---perhaps surprising---answer: \emph{all} of them except for
the obvious discards.
\begin{theorem}~\textup{\textrm{\cite{chartrand1968hamiltonian}}}\label{fac:all_graphs_have_finite_hamiltonian_index}
  If \(G\) is a connected graph on \(n\) vertices which is not a path, then
  \(\ILG{r}{G}\) is Hamiltonian for all integers \(r \geq (n - 3)\).
\end{theorem}
This led Chartrand to define the \emph{Hamiltonian Index} \(h(G)\) of a
connected graph \(G\) which is not a path, to be the \emph{smallest}
non-negative integer \(r\) such that \(\ILG{r}{G}\) is
Hamiltonian~\cite{chartrand1968hamiltonian}. The Hamiltonian Index of graphs has
since received a lot of attention from graph theorists, and a number of
interesting results, especially on upper and lower bounds, are known. An early
result by Chartrand and Wall~\cite{chartrandWall1973hamiltonian}, for instance,
states that if the minimum degree of a graph \(G\) is at least three then
\(h(G) \leq 2\) holds. See the references for a number of other interesting
graph-theoretic results on the Hamiltonian Index~\cite{chartrand1968hamiltonian,
  chartrandWall1973hamiltonian, catlin1990hamilton, gould1981line,
  lai1988hamiltonian, xiong2001circuits, xiongLiu2002hamiltonian,
  xiong2004hamiltonian}.

We now move on to the algorithmic question of \emph{computing} \(h(G)\), which
is the main focus of this work. Checking if \(h(G) = 0\) holds is the same as
checking if graph \(G\) is Hamiltonian. This is long known to be \NPC, even when
the input graph is planar and has maximum degree at most
\(3\)~\cite{gareyJohnsonTarjan1976planar}. Checking if \(h(G) = 1\) holds is the
same as checking if (i) \(G\) is \emph{not} Hamiltonian, and (ii) the line graph
\(L(G)\) is Hamiltonian. Bertossi~\cite{bertossi1981edge} showed that the latter
problem is \NPC, and Ryj\'{a}\v{c}ek et al. proved that this holds even if graph
\(G\) has maximum degree at most
\(3\)~\cite{ryjacekWoegingerXiong2011hamiltonian}. Xiong and
Liu~\cite{xiongLiu2002hamiltonian} described a polynomial-time procedure which
took a graph \(G\) with \(h(G) \geq 4\) as input, and output a graph \(G'\) such
that \(h(G) = h(G') + 1\) holds. They conjectured that given an input graph
\(G\) \emph{with the guarantee} that \(h(G) \geq 2\) holds, it should be
possible to compute \(h(G)\) in polynomial time, since by their procedure it
suffices to (eventually) check whether the index is \(2\) or \(3\).
Ryj\'{a}\v{c}ek et al. \emph{disproved} this
conjecture~\cite{ryjacekWoegingerXiong2011hamiltonian}; they showed that
checking whether \(h(G) = t\) is \NPC for \emph{any} fixed integer \(t \geq 0\),
even when the input graph \(G\) has maximum degree at most \(3\).

\subparagraph{Our problems and results.} In this work we take up the
\emph{parameterized complexity analysis} of the problem of computing the
Hamiltonian Index. Briefly put, an instance of a \emph{parameterized problem} is
a pair \((x, k)\) where \(x\) is an instance of a classical problem and \(k\) is
a (usually numerical) \emph{parameter} which captures some aspect of \(x\). A
primary goal is to find a \emph{fixed-parameter tractable} (or FPT) algorithm
for the problem, one which solves the instance in time
\(\Oh{f(k) \cdot |x|^{c}}\) where \(f()\) is a function of the parameter \(k\)
alone, and \(c\) is a constant independent of \(x\) and \(k\); this running time
is abbreviated as \(\OhStar{f(k)}\). The design of FPT algorithms is a vibrant
field of research; we refer the interested reader to standard
textbooks~\cite{cygan2015parameterized,downeyFellows2013fundamentals}.

Since checking whether \(h(G) = t\) is \NPC for \emph{any} fixed \(t \geq 0\),
the value \(h(G)\) is not a sensible parameter for this problem. Indeed, if
computing \(h(G)\) were fixed-parameter tractable with \(h(G)\) as the parameter
then we could, for instance, check whether any graph \(G\) is Hamiltonian
(\(h(G) = 0\)) in \emph{polynomial} time, which in turn would imply
\(\mathsf{P} = \mathsf{NP}\). A similar comment applies to the maximum (or
average) degree of the input graph, since the problem is \NPC already for graphs
of maximum degree \(3\). We choose the \emph{treewidth}\footnote{See the next
  section for the definition of tree decompositions and treewidth.} of the input
graph \(G\) as our parameter. This is motivated by prior related work as well,
as we describe below. Thus the main problem which we take up in this work is

\defparproblem{\HI (\HIs)}%
{A connected undirected graph \(G=(V,E)\) which is not a path, a tree
  decomposition \(\mathcal{T}=(T,\{X_{t}\}_{t \in V(T)})\) of \(G\) of width
  \(tw\), and \(r \in \mathbb{N}\).}%
{\(tw\)}%
{Is \(h(G) \leq r\)?}

Our main result is that this problem is fixed-parameter tractable. \(\omega\)
denotes the matrix multiplication exponent; it is known that \(\omega < 2.3727\)
holds~\cite{williams2012multiplying}.
\begin{theorem}\label{thm:hamIndex_is_FPT}
  There is an algorithm which solves an instance \((G,\mathcal{T},tw,r)\) of \HI
  in \(\OhStar{(1 + 2^{(\omega + 3)})^{tw}}\) time.
\end{theorem}
From this and \autoref{fac:all_graphs_have_finite_hamiltonian_index} we get
\begin{corollary}\label{cor:hamIndex_is_FPT}
  There is an algorithm which takes as input a graph \(G\) and a tree
  decomposition \(\mathcal{T}\) of width \(tw\) of \(G\) as input, and outputs
  the Hamiltonian Index \(h(G)\) of \(G\) in
  \(\OhStar{(1 + 2^{(\omega + 3)})^{tw}}\) time.
\end{corollary}
We now describe a key ingredient of our solution which we believe to be
interesting in its own right. The input to a \emph{Steiner subgraph problem}
consists of a graph \(G\) and a specified set \(K\) of \emph{terminal vertices}
of \(G\), and the objective is to find a subgraph of \(G\) which (i) contains
all the terminals, and (ii) satisfies some other specified set of constraints,
usually including connectivity constraints on the set \(K\). The archetypal
example is the \ST problem where the goal is to find a \emph{connected} subgraph
of \(G\) of the smallest size (number of edges) which contains all the
terminals. Note that such a subgraph---if it exists---will be a tree, called a
\emph{Steiner tree} for the terminal set \(K\). The non-terminal vertices in a
Steiner tree, which are included for providing connectivity at small cost for
the terminals, are called \emph{Steiner vertices}. \ST and a number of its
variants have been the subject of extensive research from graph-theoretical,
algorithmic, theoretical, and applied points of
view~\cite{promel2012steiner,du2013advances,cheng2013steiner,Bezensek:2014:SPD:2589752.2589770,hwang1995steiner,hauptmann2013compendium}.

A key part of our algorithm for computing \(h(G)\) consists of solving:

\defparproblem{\ESS (\ESSs)}%
{An undirected graph \(G=(V,E)\), a set of ``terminal'' vertices
  \(K\subseteq V\), and a tree decomposition
  \(\mathcal{T}=(T,\{X_{t}\}_{t \in V(T)})\) of \(G\), of width \(tw\).}%
{\(tw\)}%
{Does there exist an Eulerian subgraph \(G'=(V',E')\) of \(G\) such that
  \(K\subseteq{}V'\)?}

We call such a subgraph \(G'\) an \emph{Eulerian Steiner subgraph} of \(G\)
\emph{for the terminal set \(K\).}
\begin{theorem}\label{thm:ESS_is_FPT}
  There is an algorithm which solves an instance \((G,K,\mathcal{T},tw)\) of
  \ESS in \(\OhStar{(1 + 2^{(\omega + 3)})^{tw}}\) time.
\end{theorem}

\subparagraph{Related work.} The parameterized complexity of computing \(h(G)\)
\emph{per se} has not, to the best of our knowledge, been previously explored.
The two special cases of checking if \(h(G) \in \{0,1\}\) \emph{have} been
studied with the treewidth \(tw\) of the input graph \(G\) as the parameter; we
now summarize the main existing results. Checking whether \(h(G) = 0\)
holds---that is, whether \(G\) is Hamiltonian---was long known to be solvable in
\(\OhStar{tw^{\Oh{tw}}}\) time. This was suspected to be the best possible till,
in a breakthrough result in 2011, Cygan et
al.~\cite{cyganNederlofPilipczukPilipczukRooijWojtaszczyk2011solving} showed
that this can be done in randomized \(\OhStar{4^{tw}}\) time. More recently,
Bodlaender et al.~\cite{bodlaenderCyganKratschNederlof2015deterministic} and
Fomin et al.~\cite{fominLokshtanovPanolanSaurabh2016efficient} showed,
independently and using different techniques, that this can be done in
\emph{deterministic} \(\OhStar{2^{\Oh{tw}}}\) time.

Recall that a \emph{vertex cover} of graph \(G\) is any subset
\(S \subseteq V(G)\) such that every edge in \(E(G)\) has at least one of its
two endpoints in the set \(S\). A subgraph \(G'\) of a graph \(G\) is said to be
a \emph{dominating Eulerian subgraph} of \(G\) if (i) \(G'\) is Eulerian, and
(ii) \(V(G')\) contains a vertex cover of \(G\). Note that---in conformance with
the literature (e.g.\cite{lesniak1977spanning, veldman1983existence,
  fleischner2001some, li2005eulerian, lampis2017parameterized}) on this
subject---the word ``dominating'' here denotes the existence of a \emph{vertex
  cover}, and \emph{not} of a \emph{dominating set}. The input to the \DES
(\DESs) problem consists of a graph \(G\) and a tree decomposition \TT of \(G\)
of width \(tw\), and the question is whether \(G\) has a dominating Eulerian
subgraph; the parameter is \(tw\). The input to the \EHP (\EHPs) (respectively,
\EHC (\EHCs)) problem consists of a graph \(G\) and a tree decomposition \TT of
\(G\) of width \(tw\), and the question is whether \(G\) has an edge Hamiltonian
path (resp. cycle); the parameter is \(tw\). Observe that a closed trail in
graph \(G\) is an Eulerian subgraph of \(G\). So
\autoref{fac:line_graph_hamiltonian_characterization} tells us that \EHCs is
equivalent to \DESs.

The parameterized complexity of checking whether \(h(G) = 1\) holds was first
taken up by Lampis et al. in
2014~\cite{lampisMakinoMitsouUno2014hamiltonicity,lampis2017parameterized},
albeit indirectly: they addressed \EHCs and \EHPs. They showed that \EHPs is
\FPT if and only if \EHCs is \FPT, and that these problems (and hence \DESs) can
be solved in \(\OhStar{tw^{\Oh{tw}}}\) time. Very recently Misra et
al.~\cite{MisraPanolanSaurabh2019} investigated an optimization variant of \EHP
which they called \textsc{Longest Edge-Linked Path} (\textsc{LELP}). An
\emph{edge-linked path} is a sequence of edges in which every consecutive pair
has a vertex in common. Given a graph \(G\), \(k \in \mathbb{N}\), and a tree
decomposition \TT of \(G\) of width \(tw\) as input the \textsc{LELP} problem
asks whether \(G\) has an edge-linked path of length at least \(k\). Note that
setting \(k = |E(G)|\) yields \EHPs as a special case. Misra et
al.~\cite{MisraPanolanSaurabh2019} gave an algorithm which solves \textsc{LELP}
(and hence, \EHPs, \EHCs and \DESs) in \(\OhStar{(1 + 2^{(\omega + 3)})^{tw}}\)
time. 
This gives the current best algorithm\footnote{See \autoref{sec:DES} for an
  algorithm of our own design which solves \DES (and hence \EHPs and \EHCs) in
  the same running time.} for checking if \(h(G) = 1\) holds.
\begin{theorem}\upcite{MisraPanolanSaurabh2019}\label{fac:EHP_EHC_DES_FPT}
  There is an algorithm which solves an instance \((G,\mathcal{T},tw)\) of \EHP
  (respectively, \EHC or \DES) in \(\OhStar{(1 + 2^{(\omega + 3)})^{tw}}\) time.
\end{theorem}
To the best of our knowledge, ours is the first \FPT algorithm for \ESS. A
subgraph \(H\) of a graph \(G\) is a \emph{spanning} subgraph of \(G\) if \(H\)
contains every vertex of \(G\). A graph \(G\) is \emph{supereulerian} if it has
a spanning subgraph \(H\) which is Eulerian. We could not find references to the
\ESS problem in the literature, but we did find quite a bit of existing work on
the special case---obtained by setting \(K = V(G)\)---of checking if an input
graph \(G\) is supereulerian~\cite{Catlin1992SupereulerianGA,lai2013update}.
Pulleyblank observed already in 1979 that this latter problem is \NPC even on
graphs of maximum degree at most \(3\)~\cite{pulleyblank1979note}. This implies
that \ESS is \NPC as well. Richey et al.~\cite{richeyParkerRardin1985finding}
showed in 1985 that the problem can be solved in polynomial time on
\emph{series-parallel} graphs. More recently, Sau and Thilikos showed in 2010
that the problem can be solved in \(\OhStar{2^{\Oh{\sqrt{n}}}}\) time on
\emph{planar graphs} with \(n\) vertices~\cite{sauThilikos2010subexponential}.
Now consider the following parameterization:

\defparproblem{\SES (\SESs)}%
{An undirected graph \(G=(V,E)\) and a tree decomposition
  \(\mathcal{T}=(T,\{X_{t}\}_{t \in V(T)})\) of \(G\), of width \(tw\).}%
{\(tw\)}%
{Does \(G\) have a spanning Eulerian subgraph?}

Setting \(K = V(G)\) in \autoref{thm:ESS_is_FPT} we get
\begin{corollary}\label{cor:SES_is_FPT}
  There is an algorithm which solves an instance \((G,\mathcal{T},tw)\) of \SES
  in \(\OhStar{(1 + 2^{(\omega + 3)})^{tw}}\) time.
\end{corollary}
It is known that series-parallel graphs have treewidth at most \(2\) and are
planar, and that planar graphs on \(n\) vertices have treewidth
\(\Oh{\sqrt{n}}\)~\cite{bodlaender1998Partial}. Further, given a planar graph
\(G\) of treewidth \(t\) we can, in polynomial time, output a tree decomposition
of \(G\) of width \(\Oh{t}\)~\cite{kammerTholey2016Approximate}. These facts
together with \autoref{cor:SES_is_FPT} subsume the results of Richey et al. and
Sau and Thilikos, respectively.

\subparagraph{Organization of the rest of the paper.} In
\autoref{sec:preliminaries} we collect together various definitions,
observations and preliminary results which we use in the rest of the work. We
prove \autoref{thm:ESS_is_FPT} in \autoref{sec:ESS} and our main result
\autoref{thm:hamIndex_is_FPT} in \autoref{sec:hamIndex}. We conclude in
\autoref{sec:conclusion}. An alternate proof of \autoref{fac:EHP_EHC_DES_FPT} is
in \autoref{sec:DES}.


%% file: preliminaries.tex
\section{Preliminaries}\label{sec:preliminaries}
We use \(deg_{G}(v)\) to denote the degree of vertex \(v\) in graph \(G\). The
\emph{union} of graphs \(G\) and \(H\), denoted \(G \cup H\), is the graph with
vertex set \(V(G) \cup V(H)\) and edge set \(E(G) \cup E(H)\). For \(X \subseteq
V(G)\) we use \(G[X]\) to denote the subgraph \((X, \{uv \in E(G)\;;\;\{u, v\}
\subseteq X\})\) of \(G\) \emph{induced} by \(X\). A \emph{walk} in a graph
\(G\) is a sequence $(v_{0}, e_{1}, v_{1}, e_{2}, \dotsc, e_{\ell}, v_{\ell})$
of vertices \(v_{i}\) and edges \(e_{j}\) of \(G\) such that for each $1 \leq i
\leq n$ the edge $e_{i}$ has endpoints $v_{i-1}$ and $v_{i}$. A \emph{trail} is
a walk with no repeated edge, and a \emph{path} is a trail with no repeated
vertex. A walk is \emph{closed} if its first and last vertices are the same. We
consider the walk with one vertex and no edges to be closed. A \emph{tour} is a
closed trail. A \emph{cycle} is a graph which consists of a path $(u_{1},u_{2},
\ldots u_{n})$ and the additional edge $\{u_{n},u_{1}\}$. Note that a cycle
contains no repeated vertex or edge. The \emph{length} of a
walk/trail/path/cycle is the number of edges present in it. A cycle
(respectively, path) on $\ell$ vertices is denoted $C_{\ell}$ ( respectively,
$P_{\ell}$). We say that a walk (or trail/path/tour/cycle) $T$ \emph{contains},
or \emph{passes through}, a vertex $v$ (respectively, an edge \(e\)), if $v$
(respectively, \(e\)) is present in the sequence $T$. A \emph{spanning} walk (or
trail/path/tour/cycle) is one which passes through all vertices in the graph. A
\emph{Hamiltonian path} (respectively, \emph{Hamiltonian cycle}) in a graph $G$
is any spanning path (respectively, cycle) in \(G\). An \emph{Eulerian tour} in
$G$ is any spanning tour which, in addition, contains every edge of \(G\). A
graph is said to be \emph{Hamiltonian} if it contains a Hamiltonian
\emph{cycle}, and \emph{Eulerian} if it has an Eulerian tour. A graph is
Eulerian if and only if it is connected and all its vertices have even
degrees~\cite[Theorem~1.2.26]{west2001introduction}.

A \emph{tree decomposition} of a graph \(G\) is a pair
\(\mathcal{T}=(T,\{X_{t}\}_{t\in{}V(T)})\) where \(T\) is a tree
and every vertex \(t\) of \(T\) is assigned a subset
\(X_{t}\subseteq{}V(G)\) of the vertex set of \(G\). Each such
\(X_{t}\) is called a \emph{bag}, and the structure satisfies the
following conditions:
\begin{enumerate}
  \item Every vertex of \(G\) is in at least one bag.
  \item For every edge \(uv\) in \(G\) there is at least one node
    \(t\in{}V(T)\) such that \(\{u,v\}\subseteq{}X_{t}\).
  \item For each vertex \(v\) in \(G\) the set
    \(\{t\in{}V(T)\;;\;v\in{}X_{t}\}\) of all nodes whose bags contain \(v\),
    form a \emph{connected subgraph} (i.e, a sub-tree) of \(T\).
\end{enumerate}
The \emph{width} of this tree decomposition is the maximum size of a bag, minus
one. The \emph{treewidth} of a graph \(G\), denoted \(tw(G)\), is the minimum
width of a tree decomposition of \(G\). A \emph{nice tree decomposition} of a
graph \(G\) is a tree decomposition \(\mathcal{T}=(T,\{X_{t}\}_{t\in{}V(T)})\)
with the following additional structure:
\begin{enumerate}
  \item The tree \(T\) is \emph{rooted} at a distinguished \emph{root} node \(r\in{}V(T)\).
  \item The bags associated with the root node \(r\) and with
    every leaf node are all empty.
  \item Every non-leaf node is one of four types:
    \begin{enumerate}
      \item An \emph{introduce vertex node}: This is a node \(t\in{}V(T)\)
        with exactly one child node \(t'\) such that
        \((X_{t}\setminus{}X_{t'})=\{v\}\) for some vertex
        \(v\in{}V(G)\); the vertex \(v\) is \emph{introduced} at
        node \(t\).  
      \item An \emph{introduce edge node}: This is a node
        \(t\in{}V(T)\) with exactly one child node \(t'\) such
        that \(X_{t}=X_{t'}\). Further, the node \(t\) is labelled
        with an edge \(uv\in{}E(G)\) such that
        \(\{u,v\}\subseteq{}X_{t}\); the edge \(uv\) is
        \emph{introduced} at node \(t\). Moreover, every edge in
        the graph \(G\) is introduced at \emph{exactly} one
        introduce edge node in the entire tree decomposition.
      \item A \emph{forget} node: This is a node \(t\in{}V(T)\)
        with exactly one child node \(t'\) such that
        \((X_{t'}\setminus{}X_{t})=\{v\}\) for some vertex
        \(v\in{}V(G)\); the vertex \(v\) is \emph{forgotten} at
        node \(t\).
        \item A \emph{join} node: This is a node \(t\in{}V(T)\)
        with exactly two child nodes \(t_{1},t_{2}\) such that
        \(X_{t}=X_{t_{1}}=X_{t_{2}}\). 
    \end{enumerate}
  \end{enumerate}
  For a node \(t \in V(T)\) of the nice tree decomposition \TT we define (i)
  \(T_{t}\) to be the subtree of \(T\) which is rooted at \(t\), (ii) \(V_{t}\)
  to be the union of all the bags associated with nodes in \(T_{t}\), (iii)
  \(E_{t}\) to be the set of all edges introduced in \(T_{t}\), and (iv)
  \(G_{t}=(V_{t},E_{t})\) to be \emph{the subgraph of \(G\) defined by
    \(T_{t}\)}. Note that, in general, \(G_{t}\) is \emph{not} the subgraph of
  \(G\) \emph{induced} by \(V_{t}\).
  \begin{definition}[Residual subgraph]\label{def:residual_subgraph}
    Let \(G_{1}\) be a subgraph of \(G\), let \(t\) be a node of \TT, and let
    \(Y_{t} = (V_{t} \setminus X_{t})\). We define the \emph{residual subgraph
      of } \(G_{1}\) \emph{with respect to} \(t\) to be the graph
    \(G_{1}^{t} = ((V(G_{1}) \setminus Y_{t}), (E(G_{1}) \setminus E_{t}))\)
    obtained by deleting from \(G_{1}\) (i) all edges of the graph \(G_{t}\) and
    (ii) all vertices of \(G_{t}\) \emph{except} those in bag \(X_{t}\). More
    generally, we say that a subgraph \(G'\) of \(G\) is \emph{a residual
      subgraph with respect to \(t\)} if (i) \(V(G') \cap Y_{t} = \emptyset\)
    and (ii) \(E(G') \cap E_{t} = \emptyset\).
  \end{definition}
  The next theorem lets us assume, without loss of generality, that
  any tree decomposition is a nice tree decomposition.
  \begin{theorem}\label{thm:small_nice_tree_decomposition_computation}
    \textup{\cite{Klo94}},
    \textup{\cite[Section~7.3.2]{cygan2015parameterized}},
    \textup{\cite[Proposition~2.2]{bodlaenderCyganKratschNederlof2015deterministic}}
    There is an algorithm which, given a graph \(G\) and a tree decomposition
    \(\mathcal{T}=(T,\{X_{t}\}_{t\in{}V(T)})\) of \(G\) of width \(w\), computes
    a \emph{nice} tree decomposition of \(G\) of width \(w\) and with
    \(\Oh{w\cdot{}|V(G)|}\) nodes, in time which is polynomial in
    \(|V(G)|+|V(T)|+w\).
\end{theorem}
Let \(X\) be a finite set. A \emph{partition} of \(X\) is any nonempty
collection of pairwise disjoint nonempty subsets of \(X\), whose union is \(X\).
We use \(\Pi(X)\) to denote the set of all partitions of \(X\). Each subset in a
partition is called a \emph{block} of the partition. For a partition \(P\) of
\(X\) and an element \(v \in X\) we use \(P(v)\) to denote the block of \(P\) to
which \(v\) belongs. We use \(P - v\) to denote the partition of
\(X \setminus \{v\}\) obtained by \emph{eliding} \(v\) from \(P\): if
\(P(v) = \{v\}\) then \(P - v\) is the partition obtained by deleting block
\(\{v\}\) from \(P\). Otherwise, \(P - v\) is the partition obtained by deleting
element \(v\) from its block in \(P\). Let \(P, Q\) be partitions of \(X\). We
say that \(Q\) is a \emph{refinement} of \(P\), denoted \(Q \sqsubseteq P\), if
every block of \(Q\) is a subset of some block of \(P\). Note that
\(P \sqsubseteq P\) holds for every partition \(P\) of \(X\). We use
\(P \sqcup Q\) to denote the unique partition \(R\) of \(X\)---called the
\emph{join} of \(P\) and \(Q\)---such that (i) \(P \sqsubseteq R\), (ii)
\(Q \sqsubseteq R\), and (iii) if both \(P \sqsubseteq R'\) and
\(Q \sqsubseteq R'\) hold for any partition \(R'\) of \(X\) then
\(R \sqsubseteq R'\) holds as well. For a graph \(G\), subset
\(X \subseteq V(G)\) of the vertex set of \(G\), and partition \(P\) of \(X\) we
say that \(P\) is \emph{the partition of \(X\) defined by \(G\)} if each block
of \(P\) consists of the set of all vertices of \(X\) which belong to a distinct
connected component of \(G\). In particular, if \(G\) is a connected graph then
the partition of \(X\) defined by \(G\) is \(P = \{X\}\).
\begin{theorem}\textup{\cite[Section~11.2.2]{cygan2015parameterized}},
  \textup{\cite{bodlaenderCyganKratschNederlof2015deterministic}}
  \label{fac:partition_join_is_union_of_graphs}
  Let \(P, Q\) be partitions of a finite set \(X\), and let \(G_{P}, G_{Q}\) be
  two graphs on vertex set \(X\) such that \(P\) is the partition of \(X\)
  defined by \(G_{P}\) and \(Q\) is the partition of \(X\) defined by \(G_{Q}\).
  Then \(P \sqcup Q\) is the partition of \(X\) defined by the graph
  \(G_{P} \cup G_{Q}\).
\end{theorem}
The effect of graph union on connectivity is correctly captured by the join
operation of partitions, even when restricted to arbitrary subsets of vertices.
\begin{lemma}\label{lem:partition_join_connectivity}
  Let \(P,Q\) be two partitions of a non-empty finite set \(X\) and let
  \(H_{P}, H_{Q}\) be two graphs such that (i) \(V(H_{P}) \cap V(H_{Q}) = X\),
  (ii) \(E(H_{P}) \cap E(H_{Q}) = \emptyset\), (iii) the vertex set of each
  component of \(H_{P}\) and of \(H_{Q}\) has a nonempty intersection with
  \(X\), and (iv) \(P\) is the partition of \(X\) defined by \(H_{P}\) and \(Q\)
  is the partition of \(X\) defined by \(H_{Q}\). Let
  \(H = H_{P} \cup H_{Q}\). 
  Then (i) \(P \sqcup Q\) is the partition of \(X\) defined by graph \(H\), and
  (ii) \(H\) is connected if and only if \(P \sqcup Q = \{X\}\).
\end{lemma}
\begin{proof}
  Let \(G_{P}, G_{Q}\) be two graphs on vertex set \(X\) such that \(P\) is the
  partition of \(X\) defined by \(G_{P}\) and \(Q\) is the partition of \(X\)
  defined by \(G_{Q}\), and let \(G = G_{P} \cup G_{Q}\). From
  \autoref{fac:partition_join_is_union_of_graphs} we know that \(P \sqcup Q\) is
  the partition of \(X\) defined by graph \(G\). To show that \(P \sqcup Q\) is
  the partition of \(X\) defined by graph \(H\), it is thus enough to show that
  two vertices from the set \(X\) are in the same component of graph \(H\) if
  and only if they are in the same component of graph \(G\).

  So let \(x_{1},x_{2}\) be two vertices in the set \(X\) such that there is a
  path \PP in \(H\) from \(x_{1}\) to \(x_{2}\). Since
  \(E(H) = E(H_{P}) \cup E(H_{Q}))\) and \(E(H_{P}) \cap E(H_{Q}) = \emptyset\),
  each edge in path \PP corresponds to an edge in the graph \(H_{P}\) or to an
  edge in the graph \(H_{Q}\). Call a maximal contiguous sequence of edges in
  \PP from any one of \(\{H_{P},H_{Q}\}\) a \emph{run}. Equivalently: In graph
  \(H\), give the colour red to each edge from the set \(E(H_{P})\) and the
  colour blue to each edge from the set \(E(H_{Q})\). A \emph{run} in path \PP
  then consists of a maximal contiguous set of edges with the same colour. Note
  that the edges of any one run belong to a \emph{single connected component} of
  one of the two graphs \(H_{P}, H_{Q}\). Let \(t\) be the \emph{number} of runs
  in path \PP. We prove by induction on \(t\) that there is a path from
  \(x_{1}\) to \(x_{2}\) in graph \(G\) as well.
    \begin{description}
    \item[Base case, \(t = 1\).] In this case the entire path \PP consists of
      edges from exactly one of the two graphs \(H_{P}\) and \(H_{Q}\). If \PP
      is made up exclusively of edges from \(H_{P}\) then \(x_{1}\) and
      \(x_{2}\) belong to the same connected component of \(H_{P}\), and hence
      to the same block of partition \(P\), and so there is a path from
      \(x_{1}\) to \(x_{2}\) in the graph \(G_{P}\). Similarly, if \PP is made
      up exclusively of edges from \(H_{Q}\) then there is a path from \(x_{1}\)
      to \(x_{2}\) in the graph \(G_{Q}\). In either case, this path survives
      intact in graph \(G\), and so there is a path from \(x_{1}\) to \(x_{2}\)
      in the graph \(G\).
    \item[Inductive step, \(t \geq 2\).] In this case path \PP has at least two
      runs. Let \(R_{1}\) be the first run in \PP, starting from vertex
      \(x_{1}\). Without loss of generality, suppose the edges of \(R_{1}\) are
      all from graph \(H_{P}\). Let \(x\) be the last vertex in \PP which is
      incident with an edge in \(R_{1}\). Then the edges of \(R_{1}\) form a
      path from \(x_{1}\) to \(x\) in graph \(H_{P}\). So the vertices \(x_{1}\)
      and \(x\) belong to the same connected component of graph \(H_{P}\), and
      hence to the same block of partition \(P\). It follows that vertices
      \(x_{1}\) and \(x\) belong to the same connected component of graph
      \(G_{P}\), and hence there is a path, say \(\PP_{1}\), from \(x_{1}\) to
      \(x\) in graph \(G_{P}\); this path survives intact in graph \(G\).

      Now since \(t \geq 2\) we have that (i) \(x\) is not the vertex \(x_{2}\),
      and (ii) the other edge in \PP which is incident on \(x\) is from graph
      \(H_{Q}\). Since \(V(H_{P}) \cap V(H_{Q}) = X\) we get that \(x \in X\)
      holds. By the inductive hypothesis applied to the sub-path of \PP from
      \(x\) to \(x_{2}\) we get that there is a path, say \(\PP_{2}\), from
      \(x\) to \(x_{2}\) in graph \(G\). The union of the paths \(\PP_{1}\) and
      \(\PP_{2}\) contains a path from \(x_{1}\) to \(x_{2}\) in the graph
      \(G\).
    \end{description}
    This completes the induction.
  
    Now we prove the reverse direction; the arguments are quite similar to those
    above. So let \(x_{1},x_{2}\) be two vertices in the set \(X\) such that
    there is a path \PP in \(G\) from \(x_{1}\) to \(x_{2}\). For each edge
    \(uv\) in graph \(G\), if \(uv\) is an edge in the set \(E(G_{P})\) then
    give the colour red to edge \(uv\). Give the colour blue to each remaining
    edge in \(G\); the blue edges are all in \(E(G_{Q})\). Define a \emph{run}
    in path \PP to consist of a maximal contiguous set of edges with the same
    colour. Note that the edges of any one run belong to a \emph{single
      connected component} of one of the two graphs \(G_{P}, G_{Q}\). Let \(t\)
    be the \emph{number} of runs in path \PP. We prove by induction on \(t\)
    that there is a path from \(x_{1}\) to \(x_{2}\) in graph \(H\) as well.
    \begin{description}
    \item[Base case, \(t = 1\).] In this case the entire path \PP consists of
      edges from exactly one connected component of one of the two graphs
      \(G_{P}, G_{Q}\). If \PP is made up exclusively of edges from \(G_{P}\)
      then---since the partitions of \(X\) defined by the two graphs \(G_{P}\)
      and \(H_{P}\) are identical---there is a path from \(x_{1}\) to \(x_{2}\)
      in the graph \(H_{P}\). If \PP is made up exclusively of edges from
      \(G_{Q}\) then---since the partitions of \(X\) defined by the two graphs
      \(G_{Q}\) and \(H_{Q}\) are identical---there is a path from \(x_{1}\) to
      \(x_{2}\) in the graph \(H_{Q}\). In either case, this path survives
      intact in graph \(H\), and so there is a path from \(x_{1}\) to \(x_{2}\)
      in the graph \(H\).
    \item[Inductive step, \(t \geq 2\).] In this case path \PP has at least two
      runs. Let \(R_{1}\) be the first run in \PP, starting from vertex
      \(x_{1}\). Without loss of generality, suppose the edges of \(R_{1}\) are
      all from graph \(G_{P}\). Let \(x\) be the last vertex in \PP which is
      incident with an edge in \(R_{1}\). Then the edges of \(R_{1}\) form a
      path from \(x_{1}\) to \(x\) in graph \(G_{P}\). So the vertices \(x_{1}\)
      and \(x\) belong to the same connected component of graph \(G_{P}\), and
      hence to the same block of partition \(P\). It follows that vertices
      \(x_{1}\) and \(x\) belong to the same connected component of graph
      \(H_{P}\), and hence there is a path, say \(\PP_{1}\), from \(x_{1}\) to
      \(x\) in graph \(H_{P}\); this path survives intact in graph \(H\).

      Now since \(t \geq 2\) we have that (i) \(x\) is not the vertex \(x_{2}\),
      and (ii) the other edge in \PP which is incident on \(x\) is from graph
      \(G_{Q}\). Since \(V(G_{Q}) = X\) we get that \(x \in X\) holds. By the
      inductive hypothesis applied to the sub-path of \PP from \(x\) to
      \(x_{2}\) we get that there is a path, say \(\PP_{2}\), from \(x\) to
      \(x_{2}\) in graph \(H\). The union of the paths \(\PP_{1}\) and
      \(\PP_{2}\) contains a path from \(x_{1}\) to \(x_{2}\) in the graph
      \(H\).
    \end{description}
    This completes the proof that \(P \sqcup Q\) is the partition of \(X\)
    defined by graph \(H\).

    Now we take up the second claim of the lemma. If graph \(H\) is connected
    then---since \(X \subseteq V(H)\) holds---the partition of \(X\) defined by
    \(H\) is \(\{X\}\). From the first part of the lemma we know that this
    partition is in fact \(P \sqcup Q\). Thus \(P \sqcup Q = \{X\}\) holds. In
    the reverse direction, suppose \(P \sqcup Q = \{X\}\) holds. Since---from
    the first part of the lemma---\(P \sqcup Q\) is the partition of \(X\)
    defined by graph \(H\) we get that there is a path in \(H\) between every
    pair of vertices in \(X\). Now since every component of \(H_{P}\) and of
    \(H_{Q}\) has a non-empty intersection with the set \(X\), we get that there
    is a path between \emph{any two} vertices of the graph
    \(H = H_{P} \cup H_{Q}\). Thus graph
    \(H\) is connected as well. This completes the proof of the lemma.\qedhere

\end{proof}
Let \(\calA \subseteq \Pi(X)\), \(\calB \subseteq \Pi(X)\) be collections of
partitions of set \(X\). We say that \calB is a \emph{representative subset} of
\calA if (i) \(\calB \subseteq \calA\), and (ii) for any two partitions
\(P \in \calA \) and \(R \in \Pi(X)\) with \(P \sqcup R =\{X\}\), there exists a
partition \(Q \in \calB\) such that \(Q \sqcup R=\{X\}\) holds. The next theorem
lets us keep a very small subset of possible partitions of each bag in order to
remember all relevant connectivity information, while doing dynamic programming
over the bags of a tree decomposition.
\begin{theorem} \textup{\cite[Theorem~11.11]{cygan2015parameterized}},
  \textup{\cite[Theorem~3.7]{bodlaenderCyganKratschNederlof2015deterministic}}
  \label{thm:computing_representative_subsets}
  There is an algorithm which, given a set of partitions
  \(\calA \subseteq \Pi(X)\) of a finite set \(X\) as input, runs in time
  \(\vert \calA \vert \cdot 2^{(\omega - 1)(\vert X \vert)} \cdot |X|^{\Oh{1}}\)
  and outputs a representative subset \(\calB \subseteq \calA\) of size at most
  \(2^{\vert X \vert -1}\).
\end{theorem}

%% file: ess.tex
\section{An \FPT Algorithm for \ESS}\label{sec:ESS}
In this section we prove \autoref{thm:ESS_is_FPT}: we describe an algorithm
which takes an instance \((G, K, \TT, tw)\) of \ESS as input and tells
in \(\OhStar{(1 + 2^{(\omega + 3)})^{tw}}\) time whether graph \(G\) has a
subgraph which is (i) Eulerian, and (ii) contains every vertex in the terminal
set \(K\). As a first step our algorithm applies
\autoref{thm:small_nice_tree_decomposition_computation} to \(\TT\) to
obtain a \emph{nice} tree decomposition in polynomial time. So we assume,
without loss of generality, that \(\TT\) is itself a nice tree
decomposition of width \(tw\). The rest of our algorithm for \ESS consists of
doing dynamic programming (DP) over the bags of this nice tree decomposition,
and is modelled after the algorithm of Bodlaender et
al.~\cite{bodlaenderCyganKratschNederlof2015deterministic} for \ST; see also the
exposition of this algorithm in the textbook of Cygan et
al.~\cite[Sections~7.3.3 and~11.2.2]{cygan2015parameterized}.

We make one further modification to the given nice tree decomposition \(\TT\):
we pick an arbitrary terminal \(\vstar \in K\) and add it to every bag of
\(\TT\); \emph{from now on we use \(\TT\) to refer to the resulting
  ``nearly-nice'' tree decomposition in which the bags at all the leaves and the
  root are equal to \(\{\vstar\}\)}. Note that \(\vstar\) is neither introduced
nor forgotten at any bag of \TT. This step increases the width of \(\TT\) by at
most $1$ and ensures that every bag of \(\TT\) contains at least one terminal
vertex.

Recall that \(G_{t}\) denotes the graph defined by the vertices \(V_{t}\) and
edges \(E_{t}\) of \(G\) which have been ``seen'' in the subtree \(T_{t}\) of
\TT rooted at a node \(t\). If the graph \(G\) has an Eulerian subgraph
$G'=(V',E')$ which contains all the terminals $K$ then it interacts with the
structures defined by node $t$ in the following way: The part of $G'$ contained
in $G_{t}$ is a collection \(\CC = \{C_{1},\dotsc,C_{\ell}\}\) of pairwise
vertex-disjoint connected subgraphs of $G_{t}$. This collection is never empty
because the bag $X_{t}$ contains at least one terminal vertex, \emph{viz.}
\(\vstar\). Indeed, since $G'$ is connected we get that \emph{each element}
\(C_{i}\) of \(\CC\) has a non-empty intersection with $X_{t}$. Further, every
terminal vertex in the set $K \cap V_{t}$ belongs to exactly one element of
\(\CC\).

\begin{definition}[Valid partitions, witness for validity]\label{def:valid_partitions_witnesses}
  For a bag \(X_{t}\) and subsets \(X \subseteq X_{t}\), \(O \subseteq X\), we
  say that a partition \(P=\{X^{1},X^{2},\ldots X^{p}\}\) of \(X\) is
  \emph{valid for the combination} \((t,X,O)\) if there exists a subgraph
  \(G'_{t} = (V'_{t},E'_{t})\) of \(G_{t}\) such that
  \begin{enumerate}
  \item \(X_{t} \cap V(G'_{t})=X\).
  \item \(G'_{t}\) has exactly \(p\) connected components \(C_{1},C_{2},\ldots,
    C_{p}\) and for each \(i \in \{1,2,\ldots,p\}\), \(X^{i} \subseteq
    V(C_{i})\). That is, the vertex set of each connected component of
    \(G'_{t}\) has a non-empty intersection with set \(X\), and \(P\) is the
    partition of \(X\) defined by the subgraph \(G'_{t}\).
  \item Every terminal vertex from $K \cap V_{t}$ is in $V(G'_{t})$.
  \item The set of odd-degree vertices in $G'_{t}$ is exactly the set $O$.
  \end{enumerate}
  Such a subgraph \(G'_{t}\) of $G_{t}$ is a \emph{witness} for partition \(P\)
  being valid for the combination \((t,X,O)\) or, in short: \(G'_{t}\) \emph{is
    a witness for \(((t,X,O), P)\)}.
\end{definition}

Note that the fourth condition implies in particular that every vertex
$v \in V'_{t} \backslash X_{t}$ has an even degree in $G'_{t}$. The intuition
behind this definition is that (i) the subgraph $G'_{t}$ of $G_{t}$ is the
intersection of an (unknown) Eulerian Steiner subgraph \(G'\) of \(G\) with the
``uncovered'' subgraph \(G_{t}\), (ii) the set \(X \subseteq X_{t}\) is the
subset of vertices of \(G'_{t}\) which could potentially gain new neighbours as
we uncover the rest of the subgraph \(G'\), and (iii) the set \(O \subseteq X\)
is exactly the subset of vertices of \(G'_{t}\) which have odd degrees in the
uncovered part, and hence will \emph{definitely} gain new neighbours as we
uncover the rest of \(G'\). By the time we uncover all of \(G'\) (e.g., at the
root node of \TT) there will be (i) no vertices in the set \(O\) and (ii)
exactly one set in the partition \(P\).

\begin{definition}[Completion]\label{def:completion}
  For a bag $X_{t}$ and subsets $X \subseteq X_{t}$, $O \subseteq X$ let \(P\)
  be a partition of \(X\) which is valid for the combination \((t,X,O)\). Let
  \(H\) be a residual subgraph with respect to \(t\) such that \(V(H) \cap X_{t}
  = X\). We say that \(((t,X,O), P)\) \emph{completes} \(H\) if there exists a
  subgraph \(G'_{t}\) of \(G_{t}\) which is a witness for \(((t,X,O), P)\), such
  that the graph \(G'_{t} \cup H\) is an Eulerian Steiner subgraph of \(G\) for
  the terminal set \(K\). We say that \(G'_{t}\) is a \emph{certificate} for
  \(((t,X,O), P)\) completing \(H\).
\end{definition}

\begin{lemma}\label{lem:completion_odd_subset}
  Let \((G, K, \TT, tw)\) be an instance of \ESS. Let \(t\) be an arbitrary node
  of \TT, let $X \subseteq X_{t}$, $O \subseteq X$, let \(P\) be a partition of
  \(X\) which is valid for the combination \((t,X,O)\), and let \(H\) be a
  residual subgraph with respect to \(t\) with \(V(H) \cap X_{t} = X\). If
  \(((t,X,O), P)\) completes \(H\) then the set of odd-degree vertices in \(H\)
  is exactly the set \(O\).
\end{lemma}
\begin{proof}
  Let \(H_{odd} \subseteq V(H)\) be the set of odd-degree vertices in \(H\).
  Since \(((t,X,O), P)\) completes \(H\) we know that there is a subgraph
  \(G'_{t} = (V'_{t}, E'_{t})\) of \(G_{t}\) which is a witness for \(((t,X,O),
  P)\), such that the graph \(G^{\star} = G'_{t} \cup H\) is an Eulerian Steiner
  subgraph of \(G\) for the terminal set \(K\). Since \(G'_{t}\) is a witness
  for \(((t,X,O), P)\) we get that the set of odd-degree vertices in \(G'_{t}\)
  is exactly the set \(O\). Since \(H\) is a residual subgraph with respect to
  \(t\) we have that \(E'_{t} \cap E(H) = \emptyset\). Thus the degree of any
  vertex \(v\) in the graph \(G^{\star}\) is the sum of its degrees in the two
  subgraphs \(H\) and \(G'_{t}\): \(deg_{G^{\star}}(v) = deg_{H}(v) +
  deg_{G'_{t}}(v)\). And since \(G^{\star}\) is Eulerian we have that
  \(deg_{G^{\star}}(v)\) is even for every vertex \(v \in V(G^{\star})\).

  Now let \(v \in H_{odd} \subseteq V(H)\) be a vertex of odd degree in \(H\).
  Then \(v \in V(G^{\star})\) and we get that
  \(deg_{G'_{t}}(v) = deg_{G^{\star}}(v) - deg_{H}(v)\) is odd. Thus
  \(v \in O\), and so \(H_{odd} \subseteq O\). Conversely, let
  \(x \in O \subseteq V'_{t}\) be a vertex of odd degree in \(G'_{t}\). Then
  \(x \in V(G^{\star})\) and we get that
  \(deg_{H}(x) = deg_{G^{\star}}(x) - deg_{G'_{t}}(x)\) is odd. Thus
  \(x \in H_{odd}\), and so \(O \subseteq H_{odd}\). Thus the set of odd-degree
  vertices in \(H\) is exactly the set \(O\).
\end{proof}

The next lemma tells us that it is safe to apply the representative set
computation to collections of valid partitions.
\begin{lemma}\label{lem:repsets_preserve_completion}
  Let \((G, K, \TT, tw)\) be an instance of \ESS, and let \(t\) be an arbitrary
  node of \TT. Let $X \subseteq X_{t}$, $O \subseteq X$, and let \calA be a
  collection of partitions of \(X\), each of which is valid for the combination
  \((t,X,O)\). Let \calB be a representative subset of \calA, and let \(H\) be
  an arbitrary residual subgraph of \(G\) with respect to \(t\) such that
  \(V(H) \cap X_{t} = X\) holds. If there is a partition \(P \in \calA\) such
  that \(((t,X,O), P)\) completes \(H\) then there is a partition
  \(Q \in \calB\) such that \(((t,X,O), Q)\) completes \(H\).
\end{lemma}
\begin{proof}
  Suppose there is a partition \(P \in \calA\) such that \(((t,X,O), P)\)
  completes the residual subgraph \(H\). Then there exists a subgraph \(G'_{t} =
  (V'_{t},E'_{t})\) of $G_{t}$--- \(G'_{t}\) being a witness for \(((t,X,O),
  P)\)---such that (i) \(X_{t} \cap V(G'_{t})=X\), (ii) \(P\) is the partition
  of \(X\) defined by \(G'_{t}\), (iii) every terminal vertex from \(K \cap
  V_{t}\) is in \(V(G'_{t})\), (iv) the set of odd-degree vertices in \(G'_{t}\)
  is exactly the set \(O\), and (v) the graph \(G'_{t} \cup H\) is an Eulerian
  Steiner subgraph of \(G\) for the terminal set \(K\). Observe that every
  terminal vertex in the set \(K \setminus V_{t}\) is in the set \(V(H)\). Let
  \(R\) be the partition of the set \(X\) defined by the residual subgraph
  \(H\). Since the union of \(G'_{t}\) and \(H\) is connected we
  get---\autoref{lem:partition_join_connectivity}---that \(P \sqcup R = \{X\}\)
  holds. Since \calB is a representative subset of \calA we get that there
  exists a partition \(Q \in \calB\) such that \(Q \sqcup R = \{X\}\) holds.
  Since \(\calB \subseteq \calA\) we have that the partition \(Q\) of \(X\) is
  valid for the combination \((t,X,O)\). So there exists a subgraph
  \(G^{\star}_{t} = (V^{\star}_{t},E^{\star}_{t})\) of
  $G_{t}$---\(G^{\star}_{t}\) being a witness for \(((t,X,O), Q)\)---such that
  (i) \(X_{t} \cap V(G^{\star}_{t})=X\), (ii) \(Q\) is the partition of \(X\)
  defined by \(G^{\star}_{t}\), (iii) every terminal vertex from \(K \cap
  V_{t}\) is in \(V(G^{\star}_{t})\), and (iv) the set of odd-degree vertices in
  \(G^{\star}_{t}\) is exactly the set \(O\).
  Now the graph \(G^{\star}_{t} \cup H\):
  \begin{enumerate}
  \item Contains all the terminal vertices \(K\), because every terminal
    vertex from \(K \cap V_{t}\) is in \(V(G^{\star}_{t})\), and every terminal
    vertex in the set \(K \setminus V_{t}\) is in the set \(V(H)\).
  \item Has all degrees even, because (i) the edge sets \(E(G^{\star}_{t})\) and
    \(E(H)\) are disjoint, and (ii) the sets of odd-degree vertices in the two
    graphs \(G^{\star}_{t}\) and \(H\) are identical---namely, the set \(O\).
  \item Is connected---by \autoref{lem:partition_join_connectivity}---because
    \(Q \sqcup R = \{X\}\) holds.
  \end{enumerate}
  Thus the subgraph \(G^{\star}_{t}\) of \(G_{t}\) is a witness for
  \(((t,X,O), Q)\) such that the graph \(G^{\star}_{t} \cup H\) is an Eulerian
  Steiner subgraph of \(G\) for the terminal set \(K\). Hence \(((t,X,O), Q)\)
  completes the residual subgraph \(H\).
\end{proof}

\begin{lemma}\label{lem:validity_condition_at_root}
  Let \((G, K, \TT, tw)\) be an instance of \ESS, let \(r\) be the root node of
  \TT, and let \(\vstar\) be the terminal vertex which is present in every bag
  of \TT. Then \((G, K, \TT, tw)\) is a \yes instance of \ESS if and only if the
  partition \(P = \{\{\vstar\}\}\) is valid for the combination
  \((r,X = \{\vstar\}, O = \emptyset)\).
\end{lemma}
\begin{proof}
  Let \((G, K, \TT, tw)\) be a \yes instance of \ESS and let \(G'\) be an
  Eulerian Steiner subgraph of \(G\) for the terminal set \(K\). Then the
  terminal vertex \(\vstar\) is in \(V(G')\). Since \(r\) is the root node of
  \TT we have that \(X_{r} = \{v^{\star}\}\), \(V_{r} = V(G)\) and
  \(G_{r} = G\). We set \(G'_{r} = G'\). Then (i)
  $X_{r} \cap V(G'_{r}) = \{v^{\star}\} = X$, (ii) $G'_{r} = G'$ has exactly one
  connected component \(C_{1} = V(G')\) and the partition
  \(P = \{\{v^{\star}\}\}\) of \(X = \{v^{\star}\}\) is the partition of \(X\)
  defined by $G'_{r}$, (iii) every terminal vertex from $K \cap V_{r} = K$ is in
  \(V(G'_{r}) = V(G')\), and (iv) the set of odd-degree vertices in $G'_{r}$ is
  exactly the empty set $O$.
  Thus the partition \(P = \{\{v^{\star}\}\}\) is valid for the combination
  \((r,X = \{v^{\star}\}, O = \emptyset)\). This completes the forward
  direction.

  For the reverse direction, suppose the partition \(P = \{\{v^{\star}\}\}\) is
  valid for the combination \((r,X = \{v^{\star}\}, O = \emptyset)\). Then by
  definition there exists a subgraph \(G'_{r} = (V'_{r},E'_{r})\) of $G_{r} = G$
  such that (i) $X_{r} \cap V(G'_{r}) = X = \{v^{\star}\}$, (ii) $G'_{r}$ has
  exactly one connected component \(C_{1} = V(G'_{r})\), (iii) every terminal
  vertex from $K \cap V_{r} = K$ is in $V(G'_{r})$, and (iv) the set of
  odd-degree vertices in $G'_{r}$ is exactly the empty set $O$.
  Thus \(G'_{r}\) is a connected subgraph of \(G\) which contains every terminal
  vertex, and whose degrees are all even. But \(G'_{r}\) is then an Eulerian
  Steiner subgraph of \(G\), and so \((G, K, \TT, tw)\) is a \yes instance of
  \ESS.
\end{proof}

\begin{lemma}\label{lem:completion_at_root}
  Let \((G, K, \TT, tw)\) be an instance of \ESS, let \(r\) be the \emph{root}
  node of \TT, and let \(v^{\star}\) be the terminal vertex which is present in
  every bag of \TT. Let \(H = (\{v^{\star}\}, \emptyset)\),
  \(X = \{v^{\star}\}\), \(O = \emptyset\), and \(P = \{\{v^{\star}\}\}\).
  Then \((G, K, \TT, tw)\) is a \yes instance if and only if \(((r,X,O), P)\)
  completes \(H\).
\end{lemma}
\begin{proof}
  Note that \(G_{r} = G\). It is easy to verify by inspection that \(H\) is a
  residual subgraph with respect to \(r\).

  Let \((G, K, \TT, tw)\) be a \yes instance of \ESS and let \(G'\) be an
  Eulerian Steiner subgraph of \(G\) for the terminal set \(K\). Then the
  terminal vertex \(v^{\star}\) is in \(V(G')\). From
  \autoref{lem:validity_condition_at_root} we get that partition \(P\) is valid
  for the combination \((r,X,O)\), and from the proof of
  \autoref{lem:validity_condition_at_root} we get that the Eulerian Steiner
  subgraph \(G'\) is itself a witness for \(((r,X,O), P)\). Now
  \(((V(G') \cup V(H)), (E(G') \cup E(H))) = (V(G'), E(G')) = G'\), and so
  \(G' \cup H\) is an Eulerian Steiner subgraph of
  \(G\) for the terminal set \(K\). Thus \(((r,X,O), P)\) completes \(H\).

  The reverse direction is trivial: if \(((r,X,O), P)\) completes \(H\) then
  by definition there exists an Eulerian Steiner subgraph of \(G\) for the
  terminal set \(K\), and so \((G, K, \TT, tw)\) is a \yes instance.
\end{proof}

A na\"ive implementation of our algorithm would consist of computing, for each
node \(t\) of the tree decomposition \TT---starting at the leaves and working up
towards the root---and subsets \(O \subseteq X \subseteq X_{t}\), the set of all
partitions \(P\) which are valid for the combination \((t, X, O)\). At the root
node \(r\) the algorithm would apply \autoref{lem:validity_condition_at_root} to
decide the instance \((G, K, \TT, tw)\). Since a bag \(X_{t}\) can have up to
\(tw + 2\) elements (including the special terminal \(v^{\star}\)) the running
time of this algorithm could have a factor of \(tw^{tw}\) in it, since \(X_{t}\)
can have these many partitions. To avoid this we turn to the completion-based
alternate characterization of \yes
instances---\autoref{lem:completion_at_root}---and the
fact---\autoref{lem:repsets_preserve_completion}---that representative subset
computations do not ``forget'' completion properties. After computing a set
\calA of valid partitions for each combination \((t, X, O)\) we compute a
representative subset \(\calB \subseteq \calA\) and throw away the remaining
partitions \(\calA \setminus \calB\). Thus the number of partitions which we
need to remember for any combination \((t, X, O)\) never exceeds \(2^{tw}\). We
now describe the steps of the DP algorithm for each type of node in \TT. We use
\(VP[t, X, O]\) to denote the set of \textbf{v}alid \textbf{p}artitions for the
combination \((t, X, O)\) which we store in the DP table for node \(t\).

\begin{description}
\item[Leaf node \(t\):] In this case \(X_{t} = \{v^{\star}\}\). Set
  \(VP[t, \{v^{\star}\}, \{v^{\star}\}] = \emptyset\),
  \(VP[t, \{v^{\star}\}, \emptyset] = \{\{\{v^{\star}\}\}\}\), and
  \(VP[t, \emptyset, \emptyset] = \{\emptyset\}\).
\item[Introduce vertex node \(t\):] Let \(t'\) be the child node of \(t\), and
  let \(v\) be the vertex introduced at \(t\). Then \(v \notin X_{t'}\) and
  \(X_{t} = X_{t'} \cup \{v\}\). For each \(X \subseteq X_{t}\) and
  \(O \subseteq X\),
  \begin{enumerate}
  \item If \(v\) is a terminal vertex, then 
    \begin{itemize}
    \item if \(v \notin X\) or if \(v \in O\) then set
      \(VP[t, X, O] = \emptyset\)
    \item if \(v \in (X \setminus O)\) then for each partition \(P'\) in
      \(VP[t', X \setminus \{v\}, O]\), add the partition
      \(P = (P' \cup \{\{v\}\})\) to the set \(VP[t, X, O]\)
    \end{itemize}
  \item If \(v\) is \emph{not} a terminal vertex, then
    \begin{itemize}
    \item if \(v \in O\) then set \(VP[t, X, O] = \emptyset\)
    \item if \(v \in (X \setminus O)\) then for each partition \(P'\) in
      \(VP[t', X \setminus \{v\}, O]\), add the partition
      \(P = P' \cup \{\{v\}\}\) to the set \(VP[t, X, O]\)
    \item if \(v \notin X\) then set \(VP[t, X, O] = VP[t', X, O]\)
    \end{itemize}
  \item Set \(\calA = VP[t, X, O]\). Compute a representative subset
    \(\calB \subseteq \calA\) and set \(VP[t, X, O] = \calB\).
  \end{enumerate}
\item[Introduce edge node \(t\):] Let \(t'\) be the child node of \(t\), and let
  \(uv\) be the edge introduced at \(t\). Then \(X_{t} = X_{t'}\) and
  \(uv \in (E(G_{t}) \setminus E(G_{t'}))\). For each \(X \subseteq X_{t}\) and
  \(O \subseteq X\),
  \begin{enumerate}
  \item Set \(VP[t, X, O] = VP[t', X, O]\).
  \item If \(\{u,v\} \subseteq X\) then:
    \begin{enumerate}
    \item Construct a set of \emph{candidate partitions} \PP as follows.
      Initialize \(\PP = \emptyset\).
      \begin{itemize}
      \item if \(\{u,v\} \subseteq O\) then add all the partitions in
        \(VP[t', X, O \setminus \{u,v\}]\) to \PP.
      \item if \(\{u,v\} \cap O = \{u\}\) then add all the partitions in
        \(VP[t', X, (O \setminus \{u\}) \cup \{v\}]\) to \PP.
      \item if \(\{u,v\} \cap O = \{v\}\) then add all the partitions in
        \(VP[t', X, (O \setminus \{v\}) \cup \{u\}]\) to \PP.
      \item if \(\{u,v\} \cap O = \emptyset\) then add all the partitions in
        \(VP[t', X, O \cup \{u, v\}]\) to \PP.
      \end{itemize}
    \item For each candidate partition \(P' \in \PP\), if vertices \(u, v\) are
      in different blocks of \(P'\)---say
      \(u \in P'_{u}, v \in P'_{v}\;;\;P'_{u} \neq P'_{v}\)---then merge these
      two blocks of \(P'\) to obtain \(P\). That is, set
      \(P = (P' \setminus \{P'_{u}, P'_{v}\}) \cup (P'_{u} \cup P'_{v})\). Now
      set \(\PP = (\PP \setminus \{P'\}) \cup P\).
    \item Add all of \(\PP\) to the list \(VP[t, X, O]\).
    \end{enumerate}
  \item Set \(\calA = VP[t, X, O]\). Compute a representative subset
    \(\calB \subseteq \calA\) and set \(VP[t, X, O] = \calB\).
  \end{enumerate}
\item[Forget node \(t\):] Let \(t'\) be the child node of \(t\), and let \(v\)
  be the vertex forgotten at \(t\). Then \(v \in X_{t'}\) and
  \(X_{t} = X_{t'} \setminus \{v\}\). Recall that \(P(v)\) is the block of
  partition \(P\) which contains element \(v\), and that \(P - v\) is the
  partition obtained by eliding \(v\) from \(P\). For each \(X \subseteq X_{t}\)
  and \(O \subseteq X\),
  \begin{enumerate}
  \item Set
    \(VP[t, X, O] = \{P' - v \;;\; P' \in VP[t', X \cup \{v\}, O],\,|P'(v)| >
    1\}\).
  \item If \(v\) is \emph{not} a terminal vertex then set
    \(VP[t, X, O] = VP[t, X, O] \cup VP[t', X, O]\).
  \item Set \(\calA = VP[t, X, O]\). Compute a representative subset
    \(\calB \subseteq \calA\) and set \(VP[t, X, O] = \calB\).
  \end{enumerate}
\item[Join node \(t\):] Let \(t_{1}, t_{2}\) be the children of \(t\). Then
  \(X_{t} = X_{t_{1}} = X_{t_{2}}\). For each
  \(X \subseteq X_{t}, O \subseteq X\):
  \begin{enumerate}
  \item Set \(VP[t, X, O] = \emptyset\)
  \item For each \(O_{1} \subseteq O\) and
    \(\hat{O} \subseteq (X \setminus O)\): 
    \begin{enumerate}
    \item Let \(O_{2} = O \setminus O_{1}\).
    \item For each pair of partitions
      \(P_{1} \in VP[t_{1}, X, O_{1} \cup \hat{O}], P_{2} \in VP[t_{2}, X, O_{2}
      \cup \hat{O}]\), add their join \(P_{1} \sqcup P_{2}\) to the set
      \(VP[t, X, O]\).
    \end{enumerate}
  \item Set \(\calA = VP[t, X, O]\). Compute a representative subset
    \(\calB \subseteq \calA\) and set \(VP[t, X, O] = \calB\).
  \end{enumerate}
\end{description}

We now show that this DP correctly computes a solution in the stated time bound.
We assume that the tree decomposition in the input instance is modified as
described earlier. We prove the correctness of the algorithm by induction on the
structure of this tree decomposition \TT. The key insight in the proof is that
the processing at every node in \TT preserves the following
  \paragraph*{Correctness Criteria}\label{correctness_criteria}
  Let \(t\) be a node of \TT, let \(X \subseteq X_{t}, O \subseteq X\), and let
  \(VP[t, X, O]\) be the set of partitions computed by the DP for the
  combination \((t, X, O)\).
  \begin{enumerate}
  \item{\textbf{Soundness:}} Every partition \(P \in VP[t, X, O]\) is valid for
    the combination \((t, X, O)\).
  \item{\textbf{Completeness:}} For any residual subgraph \(H\) with respect to
    \(t\) with \(V(H) \cap X_{t} = X\), if there exists a partition \(P\) of
    \(X\) such that \(((t, X, O), P)\) completes \(H\) then the set
    \(VP[t, X, O]\) contains a partition \(Q\) of \(X\) such that
    \(((t, X, O), Q)\) completes \(H\). Note that
    \begin{itemize}
    \item the two partitions \(P,Q\) must both be valid for the combination
      \((t, X, O)\); and,
    \item \(Q\) can potentially be the same partition as \(P\).
    \end{itemize}
  \end{enumerate}

  The processing at each of the non-leaf nodes computes a representative subset
  as a final step. This step does not negate the correctness criteria.
  \begin{observation}\label{obs:repset_computation_preserves_correctness}
    Let \(t\) be a node of \TT, let \(X \subseteq X_{t}, O \subseteq X\), and
    let \calA be a set of partitions which satisfies the correctness criteria
    for the combination \((t, X, O)\). Let \calB be a representative subset of
    \calA. Then \calB satisfies the correctness criteria for the combination
    \((t, X, O)\).
  \end{observation}
  \begin{proof}
    Since \(\calB \subseteq \calA\) holds we get that \calB satisfies the
    soundness criterion. From \autoref{lem:repsets_preserve_completion} we get
    that \calB satisfies the completeness criterion as well.
  \end{proof}

\begin{lemma}\label{lem:leaf_node_ok}
  Let \(t\) be a leaf node of the tree decomposition \TT and let
  \(X \subseteq X_{t}, O \subseteq X\) be arbitrary subsets of \(X_{t}, X\)
  respectively. The collection \calA of partitions computed by the DP for the
  combination \((t, X, O)\) satisfies the correctness criteria.
\end{lemma}
\begin{proof}
  Here \(X_{t} = \{v^{\star}\}\). Note that the graph \(G_{t}\) consists of (i)
  the one vertex \(v^{\star}\), and (ii) no edges. We verify the conditions for
  all the three possible cases:
  \begin{itemize}
  \item \(X = \{v^{\star}\}, O = \{v^{\star}\}\). The algorithm sets
    \(\calA = \emptyset\). The soundness criterion holds vacuously.

    Observe that there is no subgraph \(G_{t'}\) of \(G_{t}\) in which vertex
    \(v^{\star}\) has an odd degree. This means that there can exist no subgraph
    \(G_{t'}\) of \(G_{t}\) for which the fourth condition in the definition of
    a valid partition---~\autoref{def:valid_partitions_witnesses}---holds. Thus
    there is no partition which is valid for the combination \((t, X, O)\).
    Hence the completeness criterion holds vacuously as well.
  \item \(X = \{v^{\star}\}, O = \emptyset\). The algorithm sets
    \(\calA = \{\{\{v^{\star}\}\}\}\). It is easy to verify by inspection that
    the subgraph \(G_{t'} = G_{t}\) of \(G_{t}\) is a witness for the partition
    \(\{\{v^{\star}\}\}\) being valid for the combination \((t, X, O)\). Hence
    the soundness criterion holds.

    Since \(X\) is the set \(\{v^{\star}\}\), the \emph{only} valid partition
    for the combination \((t, X, O)\) is \(\{\{v^{\star}\}\}\). Hence the
    completeness criterion holds trivially.
  \item \(X = \emptyset, O = \emptyset\). The algorithm sets
    \(\calA = \emptyset\). The soundness criterion holds vacuously.

    Since \(v^{\star} \in V_{t}\) is a terminal vertex and \(X = \emptyset\)
    holds, there can exist no subgraph \(G_{t'}\) of \(G_{t}\) for which both
    the conditions (1) and (3) of the definition of a valid
    partition---~\autoref{def:valid_partitions_witnesses}---hold simultaneously.
    Thus there is no partition which is valid for the combination \((t, X, O)\).
    Hence the completeness criterion holds vacuously as well. \qedhere
  \end{itemize}
\end{proof}

\begin{lemma}\label{lem:introduce_vertex_node_ok}
  Let \(t\) be an introduce vertex node of the tree decomposition \TT and let
  \(X \subseteq X_{t}, O \subseteq X\) be arbitrary subsets of \(X_{t}, X\)
  respectively. The collection \calA of partitions computed by the DP for the
  combination \((t, X, O)\) satisfies the correctness criteria.
\end{lemma}
\begin{proof}
  Let \(t'\) be the child node of \(t\), and let \(v\) be the vertex introduced
  at \(t\). Then \(v \notin X_{t'}\) and \(X_{t} = X_{t'} \cup \{v\}\) hold.
  Note that no edges incident with \(v\) have been introduced so far; so we have
  that \(deg_{G_{t}}(v) = 0\) holds. We analyze each choice made by the
  algorithm:
  \begin{enumerate}
  \item If \(v \in O\) holds then the algorithm sets \(\calA = \emptyset\). The
    soundness criterion holds vacuously.

    Since \(deg_{G_{t}}(v) = 0\) holds, there can exist no subgraph \(G_{t'}\)
    of \(G_{t}\) for which the fourth condition of the definition of a valid
    partition---~\autoref{def:valid_partitions_witnesses}---holds. Thus there is
    no partition which is valid for the combination \((t, X, O)\). Hence the
    completeness condition holds vacuously as well.

  \item If \(v \in (X \setminus O)\) holds then the algorithm takes each
    partition \(P'\) in \(VP[t', X \setminus \{v\}, O]\) and adds the partition
    \(P = (P' \cup \{\{v\}\})\) to the set \(\calA\). By inductive assumption we
    have that the set \(VP[t', X \setminus \{v\}, O]\) of partitions is sound
    and complete for the combination \((t', X \setminus \{v\}, O)\).

    Let \(P = (P' \cup \{\{v\}\})\) be an arbitrary partition in the set
    \(\calA\), where \(P'\) is a partition from the set
    \(VP[t', X \setminus \{v\}, O]\). Then the partition \(P'\) is valid for the
    combination \((t', X \setminus \{v\}, O)\), and so there exists a subgraph
    \(H\) of the graph \(G_{t'}\) such that \(H\) is a witness for
    \(((t', X \setminus \{v\}, O), P')\). It is easy to verify by inspection
    that the graph \(G_{t}' = ( V(H) \cup \{v\}, E(H))\) is a subgraph of
    \(G_{t}\) which satisfies all the four conditions of
    \autoref{def:valid_partitions_witnesses} for being a witness for
    \(((t, X, O), P)\). Thus the soundness condition holds for the set
    \(\calA\).
    
    Now we prove completeness. So let \(H\) be a residual subgraph with respect
    to \(t\) with \(V(H) \cap X_{t} = X\), for which there exists a partition
    \(P=\{X^{1},X^{2},\ldots X^{p}\}\) of \(X\) such that \(((t, X, O), P)\)
    completes \(H\). We need to show that the set \calA computed by the
    algorithm contains some partition \(Q\) of \(X\) such that
    \(((t, X, O), Q)\) completes \(H\). Observe that there exists a subgraph
    \(G'_{t}\) of \(G_{t}\)---a witness for \(((t,X,O), P)\)---such that the
    following hold:
    \begin{enumerate}
    \item \(X_{t} \cap V(G'_{t})=X\).
    \item \(G'_{t}\) has exactly \(p\) connected components
      \(C_{1}, C_{2}, \dotsc, C_{p}\) and for each
      \(i \in \{1, 2, \dotsc, p\}$, $X^{i} \subseteq V(C_{i})\) holds.
    \item Every terminal vertex from \(K \cap V_{t}\) is in \(V(G'_{t})\).
    \item The set of odd-degree vertices in \(G'_{t}\) is exactly the set \(O\).
    \item The graph \(G'_{t} \cup H\) is an
      Eulerian Steiner subgraph of \(G\) for the terminal set \(K\).
    \end{enumerate}
    Since \(deg_{G_{t}}(v) = 0\) holds, we get that \(deg_{G'_{t}}(v) = 0\)
    holds as well. Thus vertex \(v\) forms a connected component by itself in
    graph \(G'_{t}\). Without loss of generality, let this component by
    \(C_{p}\). Then we get that \(X^{p} = V(C_{p}) = \{v\}\), and that
    \(P'=\{X^{1},X^{2},\ldots X^{(p - 1)}\}\) is a partition of the set
    \(X \setminus \{v\}\).

    Since \(v \in X\) and \(V(H) \cap X_{t} = X\) hold, and since the graph
    \(G'_{t} \cup H\) is Eulerian, we get that
    vertex \(v\) has a positive even degree in graph \(H\). Since \(H\) is a
    residual subgraph with respect to \(t\) we have that (i)
    \(V(H) \cap (V_{t} \setminus X_{t}) = \emptyset\) and (ii)
    \(E(H) \cap E_{t} = \emptyset\) hold. Since \(X_{t} = X_{t'} \cup \{v\}\)
    holds, we get that \(V_{t'} = V_{t} \setminus \{v\}\) and hence
    \(V_{t'} \setminus X_{t'} = V_{t} \setminus X_{t}\) holds. Hence
    \(V(H) \cap (V_{t'} \setminus X_{t'}) = \emptyset\) holds. Further, since
    \(E_{t'} \subseteq E_{t}\) holds we get that
    \(E(H) \cap E_{t'} = \emptyset\) holds as well. Thus \(H\) is a residual
    subgraph with respect to node \(t'\) which (i) contains vertex \(v\) and
    (ii) satisfies \(V(H) \cap X_{t'} = (X \setminus \{v\})\).

    Now let \(G'_{t'}\) be the graph obtained from \(G'_{t}\) by deleting vertex
    \(v\). Then \(G'_{t'}\) is a subgraph of \(G_{t'}\), and it is
    straightforward to verify that the following hold:
    \begin{enumerate}
    \item \(X_{t'} \cap V(G'_{t'}) = (X \setminus \{v\})\).
    \item \(G'_{t'}\) has exactly \(p - 1\) connected components
      \(C_{1}, C_{2}, \dotsc, C_{(p - 1)}\) and for each
      \(i \in \{1, 2, \dotsc, p - 1\}$, $X^{i} \subseteq V(C_{i})\) holds.
    \item Every terminal vertex from \(K \cap V_{t'}\) is in \(V(G'_{t'})\).
    \item The set of odd-degree vertices in \(G'_{t'}\) is exactly the set
      \(O\).
    \item The graph \(G'_{t'} \cup H\) is identical to the graph
      \(G'_{t} \cup H\), and
      hence is an Eulerian Steiner subgraph of \(G\) for the terminal set \(K\).
    \end{enumerate}
    Thus \(H\) is a residual subgraph with respect to \(t'\) with
    \(V(H) \cap X_{t'} = (X \setminus \{v\})\), and
    \(P'=\{X^{1},X^{2},\ldots X^{(p - 1)}\}\) is a partition of
    \(X \setminus \{v\}\) such that \(((t', X \setminus \{v\}, O), P')\)
    completes \(H\). From the inductive assumption we know that the set
    \(VP[t', X \setminus \{v\}, O]\) contains a partition
    \(Q'=\{Y^{1},Y^{2},\ldots Y^{q}\}\) of \(X \setminus \{v\}\) such that
    \(((t', X \setminus \{v\}, O), Q')\) completes \(H\). So there is a subgraph
    \(G''_{t'}\) of \(G_{t'}\)---a witness for
    \(((t', X \setminus \{v\}, O), Q')\)---such that the following hold:
    \begin{enumerate}
    \item \(X_{t'} \cap V(G''_{t'})=X \setminus \{v\}\).
    \item \(G''_{t'}\) has exactly \(q\) connected components
      \(D_{1}, D_{2}, \dotsc, D_{q}\) and for each
      \(i \in \{1, 2, \dotsc, q\}$, $Y^{i} \subseteq V(D_{i})\) holds.
    \item Every terminal vertex from \(K \cap V_{t'}\) is in \(V(G''_{t'})\).
    \item The set of odd-degree vertices in \(G''_{t'}\) is exactly the set
      \(O\).
    \item The graph \(G''_{t'} \cup H\) is an
      Eulerian Steiner subgraph of \(G\) for the terminal set \(K\).
    \end{enumerate}

    Now the algorithm adds the partition
    \(Q = Q' \cup \{\{v\}\} = =\{Y^{1},Y^{2},\ldots Y^{q}, \{v\}\}\) of set
    \(X\) to the set \calA. It is straightforward to verify that the graph
    \(\hat{G}_{t} = (V(G''_{t'}) \cup \{v\}, E(G''_{t'}))\) is a subgraph of
    graph \(G_{t}\) for which the following hold:
    \begin{enumerate}
    \item \(X_{t} \cap V(\hat{G}_{t})=X\).
    \item \(\hat{G}_{t}\) has exactly \(q+1\) connected components
      \(D_{1}, D_{2}, \dotsc, D_{q}, D_{q+1} = (\{v\}, \emptyset)\) and for each
      \(i \in \{1, 2, \dotsc, q+1\}$, $Y^{i} \subseteq V(D_{i})\) holds.
    \item Every terminal vertex from \(K \cap V_{t}\) is in \(V(\hat{G}_{t})\).
    \item The set of odd-degree vertices in \(\hat{G}_{t}\) is exactly the set
      \(O\).
    \item The graph \(\hat{G}_{t} \cup H\) is
      an Eulerian Steiner subgraph of \(G\) for the terminal set \(K\).
    \end{enumerate}
    Thus \calA contains a partition \(Q\) of \(X\) such that \(((t, X, O), Q)\)
    completes \(H\), as was required to be shown for completeness.
  \item If \(v\) is a terminal vertex and \(v \notin X\) holds then the
    algorithm sets \(\calA = \emptyset\). The soundness criterion holds
    vacuously.

    Since \(v \in V_{t}\) is a terminal vertex and \(v \notin X\) holds, there
    can exist no subgraph \(G_{t'}\) of \(G_{t}\) for which both the conditions
    (1) and (3) of the definition of a valid
    partition---\autoref{def:valid_partitions_witnesses}---hold simultaneously.
    Thus there is no partition which is valid for the combination \((t, X, O)\).
    Hence the completeness condition holds vacuously as well.
  \item If \(v\) is \emph{not} a terminal vertex and \(v \notin X\) holds then
    the algorithm sets \(\calA = VP[t', X, O]\). It is straightforward to verify
    using Definitions~\ref{def:residual_subgraph},
    \ref{def:valid_partitions_witnesses}, and~\ref{def:completion} that:
    \begin{itemize}
    \item a partition \(P\) of set \(X\) is valid for the combination
      \((t, X, O)\) if and only if it is valid for the combination
      \((t', X, O)\);
    \item a subgraph of \(G_{t}\) is a witness for \(((t, X, O), P)\) if and
      only if it is (i) a subgraph of \(G_{t'}\) and (ii) a witness for
      \(((t', X, O), P)\);
    \item a graph \(H\) is a residual subgraph with respect to \(t\) with
      \(V(H) \cap X_{t} = X\) if and only if \(H\) is a residual subgraph with
      respect to \(t'\) with \(V(H) \cap X_{t'} = X\); and,
    \item for any residual subgraph \(H\) with respect to \(t\) with
      \(V(H) \cap X_{t} = X\) and any partition \(P\) of \(X\),
      \(((t, X, O), P)\) completes \(H\) if and only if \(((t', X, O), P)\)
      completes \(H\).
    \end{itemize}
    By the inductive assumption we have that the set \(VP[t', X, O]\) of
    partitions is sound and complete for the combination \((t', X, O)\). It
    follows from the above equivalences that the set \(\calA = VP[t', X, O]\) is
    sound and complete for the combination \((t, X, O)\). \qedhere
  \end{enumerate}
\end{proof}

\begin{lemma}\label{lem:introduce_edge_node_ok}
  Let \(t\) be an introduce edge node of the tree decomposition \TT and let
  \(X \subseteq X_{t}, O \subseteq X\) be arbitrary subsets of \(X_{t}, X\)
  respectively. The collection \calA of partitions computed by the DP for the
  combination \((t, X, O)\) satisfies the correctness criteria.
\end{lemma}
\begin{proof}
  Let \(t'\) be the child node of \(t\), and let \(uv\) be the edge introduced
  at \(t\). Then \(X_{t} = X_{t'}\), \(V_{t} = V_{t'}\) and
  \(uv \in (E(G_{t}) \setminus E(G_{t'}))\). The algorithm initializes
  \(\calA = VP[t', X, O]\). By the inductive assumption we have that every
  partition \(P' \in \calA = VP[t', X, O]\) is valid for the combination
  \((t', X, O)\). Note that while edge \(uv\) is \emph{available} for use in
  constructing a witness for \(((t, X, O), P)\), it is not \emph{mandatory} to
  use this edge in any such witness. Applying this observation, it is
  straightforward to verify that if a subgraph \(G_{t'}'\) of \(G_{t'}\) is a
  witness for \(((t', X, O), P')\) then it is also (i) a subgraph of \(G_{t}\),
  and (ii) a witness for \(((t, X, O), P')\). Thus all partitions in
  \(VP[t', X, O]\) are valid for the combination \((t, X, O)\).

  The algorithm adds zero or more partitions to \(\calA\) depending on how the
  set \(\{u, v\}\) intersects the sets \(X\) and \(O\). We analyze each choice
  made by the algorithm:
  \begin{enumerate}
  \item If \(u \notin X\) or \(v \notin X\) holds then the algorithm does not
    make further changes to \(\calA\): it sets \(\calA = VP[t', X, O]\). Since
    (i) the criteria for
    validity---\autoref{def:valid_partitions_witnesses}---are based only on
    graphs whose intersection with \(X_{t}\) is exactly the set \(X\), and (ii)
    the new edge \(uv\) does not have both end points in this set, it is
    intuitively clear that the relevant set of valid partitions should not
    change in this case. Formally, it is straightforward to verify using
    Definitions~\ref{def:residual_subgraph},
    \ref{def:valid_partitions_witnesses}, and~\ref{def:completion} that:
    \begin{itemize}
    \item a partition \(P\) of set \(X\) is valid for the combination
      \((t, X, O)\) if and only if it is valid for the combination
      \((t', X, O)\);
    \item a subgraph of \(G_{t}\) is a witness for \(((t, X, O), P)\) if and
      only if it is (i) a subgraph of \(G_{t'}\) and (ii) a witness for
      \(((t', X, O), P)\);
    \item a graph \(H\) is a residual subgraph with respect to \(t\) with
      \(V(H) \cap X_{t} = X\) if and only if \(H\) is a residual subgraph with
      respect to \(t'\) with \(V(H) \cap X_{t'} = X\); and,
    \item for any residual subgraph \(H\) with respect to \(t\) with
      \(V(H) \cap X_{t} = X\) and any partition \(P\) of \(X\),
      \(((t, X, O), P)\) completes \(H\) if and only if \(((t', X, O), P)\)
      completes \(H\).
    \end{itemize}
    By the inductive assumption we have that the set \(VP[t', X, O]\) of
    partitions is sound and complete for the combination \((t', X, O)\). It
    follows from the above equivalences that the set \(\calA = VP[t', X, O]\) is
    sound and complete for the combination \((t, X, O)\).
  \item If \(\{u,v\} \subseteq O\) then for each partition
    \(P' \in VP[t', X, O \setminus \{u,v\}]\),
    \begin{itemize}
    \item If vertices \(u,v\) are in the same block of \(P'\) then the algorithm
      adds \(P = P'\) to the set \calA.
    \item If vertices \(u,v\) are in different blocks of \(P'\) then the
      algorithm merges these two blocks of \(P'\) and adds the resulting
      partition \(P\)---with one fewer block than \(P'\)---to the set \calA.
    \end{itemize}
    In either case, by the inductive assumption we have that partition \(P'\) is
    valid for the combination \((t', X, O \setminus \{u,v\})\). Let \(G_{t'}''\)
    be (i) a subgraph of \(G_{t'}\) and (ii) a witness for
    \(((t', X, O \setminus \{u,v\}), P')\), and let
    \(G_{t}' = (V(G_{t'}''), E(G_{t'}'') \cup \{uv\})\) be the graph obtained
    from \(G_{t'}''\) by adding the edge \(uv\). Then \(G_{t}'\) is a subgraph
    of \(G_{t}\). Vertices \(u,v\) have even degrees in \(G_{t'}''\), and hence
    they have odd degrees in \(G_{t}'\). It is straightforward to verify that
    \(G_{t}'\) is a witness for \(((t, X, O), P)\). Thus the addition of
    partition \(P\) to \calA preserves the soundness of \calA.

    Now we prove completeness. So let \(H\) be a residual subgraph with respect
    to \(t\) with \(V(H) \cap X_{t} = X\), for which there exists a partition
    \(P=\{X^{1},X^{2},\ldots X^{p}\}\) of \(X\) such that \(((t, X, O), P)\)
    completes \(H\). We need to show that the set \calA computed by the
    algorithm contains some partition \(Q\) of \(X\) such that
    \(((t, X, O), Q)\) completes \(H\). Observe that there exists a subgraph
    \(G'_{t}\) of \(G_{t}\)---a witness for \(((t,X,O), P)\)---such that the
    following hold:
    \begin{enumerate}
    \item \(X_{t} \cap V(G'_{t})=X\).
    \item \(G'_{t}\) has exactly \(p\) connected components
      \(C_{1}, C_{2}, \dotsc, C_{p}\) and for each
      \(i \in \{1, 2, \dotsc, p\}$, $X^{i} \subseteq V(C_{i})\) holds.
    \item Every terminal vertex from \(K \cap V_{t}\) is in \(V(G'_{t})\).
    \item The set of odd-degree vertices in \(G'_{t}\) is exactly the set
      \(O\).
    \item The graph \(G'_{t} \cup H\) is an
      Eulerian Steiner subgraph of \(G\) for the terminal set \(K\).
    \end{enumerate}
    Note that by the definition of a residual subgraph, graph \(H\) (i) does
    \emph{not} contain edge \(uv\), and (ii) is a residual subgraph with respect
    to node \(t'\) as well. We consider two cases.
    \begin{itemize}
    \item Suppose edge \(uv\) is not present in graph \(G'_{t}\). Then it is
      straightforward to verify that \(G'_{t}\) is a witness for
      \(((\mathbf{t'},X,O), P)\) as well. By the inductive hypothesis there
      exists some partition \(Q\) of \(X\) in the set \(VP[t', X, O]\) such that
      \(((t', X, O), Q)\) completes \(H\). So there exists a subgraph
      \(G'_{t'}\) of \(G_{t'}\) which is a certificate for \(((t',X,O), Q)\)
      completing \(H\). It is straightforward to verify that \(G'_{t'}\) is a
      certificate for \(((t,X,O), Q)\) completing \(H\) as well. The algorithm
      adds partition \(Q\) to the set \calA during the initialization, so the
      completeness criterion is satisfied in this case.
    \item Suppose edge \(uv\) \emph{is} present in graph \(G'_{t}\). 
      Let \(H' = (V(H), (E(H) \cup \{uv\}))\) be the graph obtained by adding
      edge \(uv\) to graph \(H\), and let
      \(G'_{t'} = (V(G'_{t}), (E(G'_{t}) \setminus \{uv\}))\) be the graph
      obtained by deleting edge \(uv\) from graph \(G'_{t}\). Then it is
      straightforward to verify that (i) the set of odd-degree vertices in
      \(G'_{t'}\) is exactly the set \(O \setminus \{u,v\}\), (ii) \(H'\) is a
      residual subgraph for node \(t'\), and (iii) \(G'_{t'}\) is a subgraph of
      \(G_{t'}\) such that the graph \(G'_{t'} \cup H' = G'_{t} \cup H\) is an
      Eulerian Steiner subgraph of \(G\) for the
      terminal set \(K\). Let \(P'\) be the partition of \(X\) defined by graph
      \(G'_{t'}\). Then \(G'_{t'}\) is a witness for
      \(((t', X, O \setminus \{u,v\}), P')\) such that the union of \(G'_{t'}\)
      and the residual subgraph \(H'\) of \(t'\) is an Eulerian Steiner subgraph
      of \(G\) for the terminal set \(K\). That is,
      \(((t', X, O \setminus \{u,v\}), P')\) completes \(H'\). So by the
      inductive assumption there exists some partition \(Q'\) of \(X\) in the
      set \(VP[t', X, O \setminus \{u,v\}]\) such that
      \(((t, X, O \setminus \{u,v\}), Q')\) completes \(H'\). So there exists a
      subgraph \(\hat{G}'\) of \(G_{t'}\) such that (i) \(\hat{G}'\) is a
      witness for \(((t, X, O \setminus \{u,v\}), Q')\) and (ii)
      \(\hat{G}' \cup H'\) is an Eulerian
      Steiner subgraph of \(G\) for the terminal set \(K\).

      Note that \(Q'\) is the partition of set \(X\) defined by the graph
      \(\hat{G}'\). Suppose both \(u\) and \(v\) are in the same block of
      partition \(Q'\). Then adding the edge \(uv\) to \(\hat{G}'\) does not
      change the partition of \(X\) defined by \(\hat{G}'\). It follows that
      the graph \(\hat{G} = (V(\hat{G}'), E(\hat{G}') \cup \{uv\})\) is a
      subgraph of \(G_{t}\) such that (i) \(\hat{G}\) is a witness for
      \(((t, X, O , Q')\) and (ii)
      \(\hat{G} \cup H\) is an Eulerian Steiner
      subgraph of \(G\) for the terminal set \(K\). Thus \(((t, X, O, Q')\)
      completes the residual subgraph \(H\). Now notice that our algorithm
      adds the partition \(Q'\) to the set \calA. Thus the completeness
      criterion holds in this case.

      In the remaining case, vertices \(u\) and \(v\) are in distinct blocks
      of partition \(Q'\). Let \(Q\) be the partition obtained from \(Q'\) by
      merging together the two blocks to which vertices \(u\) and \(v\)
      belong, respectively, and leaving the other blocks as they are. Let
      \(\hat{G}\) be defined as in the previous paragraph. Then the partition
      of \(X\) defined by \(\hat{G}\) is \(Q\). It follows that \(\hat{G}\) is
      a subgraph of \(G_{t}\) such that (i) \(\hat{G}\) is a witness for
      \(((t, X, O , Q)\) and (ii)
      \(\hat{G} \cup H\) is an Eulerian Steiner
      subgraph of \(G\) for the terminal set \(K\). Thus \(((t, X, O, Q)\)
      completes the residual subgraph \(H\). Now notice that our algorithm
      adds the partition \(Q\) to the set \calA. Thus the completeness
      criterion holds in this case as well.
    \end{itemize}

  \item If \(\{u,v\} \cap O = \{u\}\) then for each partition
    \(P' \in VP[t', X, (O \setminus \{u\}) \cup \{v\}]\),
    \begin{itemize}
    \item If vertices \(u,v\) are in the same block of \(P'\) then the
      algorithm adds \(P = P'\) to the set \calA.
    \item If vertices \(u,v\) are in different blocks of \(P'\) then the
      algorithm merges these two blocks of \(P'\) and adds the resulting
      partition \(P\)---with one fewer block than \(P'\)---to the set \calA.
    \end{itemize}
    In either case, by the inductive assumption we have that partition \(P'\)
    is valid for the combination \((t', X, (O \setminus \{u\}) \cup \{v\})\).
    Let \(G_{t'}''\) be (i) a subgraph of \(G_{t'}\) and (ii) a witness for
    \(((t', X, (O \setminus \{u\}) \cup \{v\}), P')\), and let
    \(G_{t}' = (V(G_{t'}''), E(G_{t'}'') \cup \{uv\})\) be the graph obtained
    from \(G_{t'}''\) by adding the edge \(uv\). Then \(G_{t}'\) is a subgraph
    of \(G_{t}\). In \(G_{t'}''\) the degree of vertex \(u\) is even, and the
    degree of vertex \(v\) is odd. So in \(G_{t}'\) vertex \(u\) has an odd
    degree, and vertex \(v\) has an even degree. It is straightforward to
    verify that \(G_{t}'\) is a witness for \(((t, X, O), P)\). Thus the
    addition of partition \(P\) to \calA preserves the soundness of \calA.

    Now we prove completeness. So let \(H\) be a residual subgraph with
    respect to \(t\) with \(V(H) \cap X_{t} = X\), for which there exists a
    partition \(P=\{X^{1},X^{2},\ldots X^{p}\}\) of \(X\) such that
    \(((t, X, O), P)\) completes \(H\).  We
    need to show that the set \calA computed by the algorithm contains some
    partition \(Q\) of \(X\) such that \(((t, X, O), Q)\) completes \(H\).
    Observe that there exists a subgraph \(G'_{t}\) of \(G_{t}\)---a witness
    for \(((t,X,O), P)\)---such that the following hold:
    \begin{enumerate}
    \item \(X_{t} \cap V(G'_{t})=X\).
    \item \(G'_{t}\) has exactly \(p\) connected components
      \(C_{1}, C_{2}, \dotsc, C_{p}\) and for each
      \(i \in \{1, 2, \dotsc, p\}$, $X^{i} \subseteq V(C_{i})\) holds.
    \item Every terminal vertex from \(K \cap V_{t}\) is in \(V(G'_{t})\).
    \item The set of odd-degree vertices in \(G'_{t}\) is exactly the set
      \(O\).
    \item The graph \(G'_{t} \cup H\) is an
      Eulerian Steiner subgraph of \(G\) for the terminal set \(K\).
    \end{enumerate}
    Note that by the definition of a residual subgraph, graph \(H\) (i) does
    \emph{not} contain edge \(uv\), and (ii) is a residual subgraph with
    respect to node \(t'\) as well. We consider two cases.
    \begin{itemize}
    \item Suppose edge \(uv\) is not present in graph \(G'_{t}\). Then it is
      straightforward to verify that \(G'_{t}\) is a witness for
      \(((\mathbf{t'},X,O), P)\) as well. By the inductive hypothesis there
      exists some partition \(Q\) of \(X\) in the set \(VP[t', X, O]\) such
      that \(((t, X, O), Q)\) completes \(H\). This same partition \(Q\) is
      present in the set \calA as well.
    \item Suppose edge \(uv\) \emph{is} present in graph \(G'_{t}\). 
      Let \(H' = (V(H), (E(H) \cup \{uv\}))\) be the graph obtained by adding
      edge \(uv\) to graph \(H\), and let
      \(G'_{t'} = (V(G'_{t}), (E(G'_{t}) \setminus \{uv\}))\) be the graph
      obtained by deleting edge \(uv\) from graph \(G'_{t}\). Then it is
      straightforward to verify that (i) the set of odd-degree vertices in
      \(G'_{t'}\) is exactly the set \((O \setminus \{u\}) \cup \{v\}\), (ii)
      \(H'\) is a residual subgraph for node \(t'\), and (iii) \(G'_{t'}\) is a
      subgraph of \(G_{t'}\) such that the graph
      \(G'_{t'} \cup H' = G'_{t} \cup H\) is an Eulerian Steiner subgraph of
      \(G\) for the
      terminal set \(K\). Let \(P'\) be the partition of \(X\) defined by graph
      \(G'_{t'}\). Then \(G'_{t'}\) is a witness for
      \(((t', X, (O \setminus \{u\}) \cup \{v\}), P')\) such that the union of
      \(G'_{t'}\) and the residual subgraph \(H'\) of \(t'\) is an Eulerian
      Steiner subgraph of \(G\) for the terminal set \(K\). That is,
      \(((t', X, (O \setminus \{u\}) \cup \{v\}), P')\) completes \(H'\). So by
      the inductive assumption there exists some partition \(Q'\) of \(X\) in
      the set \(VP[t', X, (O \setminus \{u\}) \cup \{v\}]\) such that
      \(((t, X, (O \setminus \{u\}) \cup \{v\}), Q')\) completes \(H'\). So
      there exists a subgraph \(\hat{G}'\) of \(G_{t'}\) such that (i)
      \(\hat{G}'\) is a witness for
      \(((t, X, (O \setminus \{u\}) \cup \{v\}), Q')\) and (ii)
      \(\hat{G}' \cup H'\) is an Eulerian
      Steiner subgraph of \(G\) for the terminal set \(K\).

      Note that \(Q'\) is the partition of set \(X\) defined by the graph
      \(\hat{G}'\), and that the set of odd-degree vertices in \(\hat{G}'\) is
      exactly the set \((O \setminus \{u\}) \cup \{v\}\). Suppose both \(u\) and
      \(v\) are in the same block of partition \(Q'\). Then adding the edge
      \(uv\) to \(\hat{G}'\) (i) does not change the partition of \(X\) defined
      by \(\hat{G}'\), and (ii) \emph{does} change the set of odd-degree
      vertices to \(O\). It follows that the graph
      \(\hat{G} = (V(\hat{G}'), E(\hat{G}') \cup \{uv\})\) is a subgraph of
      \(G_{t}\) such that (i) \(\hat{G}\) is a witness for \(((t, X, O , Q')\)
      and (ii) \(\hat{G} \cup H\) is an Eulerian
      Steiner subgraph of \(G\) for the terminal set \(K\). Thus
      \(((t, X, O, Q')\) completes the residual subgraph \(H\). Now notice that
      our algorithm adds the partition \(Q'\) to the set \calA. Thus the
      completeness criterion holds in this case.

      In the remaining case, vertices \(u\) and \(v\) are in distinct blocks of
      partition \(Q'\). Let \(Q\) be the partition obtained from \(Q'\) by
      merging together the two blocks to which vertices \(u\) and \(v\) belong,
      respectively, and leaving the other blocks as they are. Let \(\hat{G}\) be
      defined as in the previous paragraph. Then the partition of \(X\) defined
      by \(\hat{G}\) is \(Q\). It follows that \(\hat{G}\) is a subgraph of
      \(G_{t}\) such that (i) \(\hat{G}\) is a witness for \(((t, X, O , Q)\)
      and (ii) \(\hat{G} \cup H\) is an Eulerian
      Steiner subgraph of \(G\) for the terminal set \(K\). Thus
      \(((t, X, O, Q)\) completes the residual subgraph \(H\). Now notice that
      our algorithm adds the partition \(Q\) to the set \calA. Thus the
      completeness criterion holds in this case as well.
    \end{itemize}
  \item The case when \(\{u,v\} \cap O = \{v\}\) is symmetrical to the previous
    case, so we leave out the arguments for this case.
  \item If \(\{u,v\} \cap O = \emptyset\) then for each partition
    \(P' \in VP[t', X, O \cup \{u, v\}]\),
    \begin{itemize}
    \item If vertices \(u,v\) are in the same block of \(P'\) then the algorithm
      adds \(P = P'\) to the set \calA.
    \item If vertices \(u,v\) are in different blocks of \(P'\) then the
      algorithm merges these two blocks of \(P'\) and adds the resulting
      partition \(P\)---with one fewer block than \(P'\)---to the set \calA.
    \end{itemize}
    In either case, by the inductive assumption we have that partition \(P'\) is
    valid for the combination \((t', X, O \cup \{u,v\})\). Let \(G_{t'}''\) be
    (i) a subgraph of \(G_{t'}\) and (ii) a witness for
    \(((t', X, O \cup \{u, v\}), P')\), and let
    \(G_{t}' = (V(G_{t'}''), E(G_{t'}'') \cup \{uv\})\) be the graph obtained
    from \(G_{t'}''\) by adding the edge \(uv\). Then \(G_{t}'\) is a subgraph
    of \(G_{t}\). Vertices \(u,v\) have odd degrees in \(G_{t'}''\), and hence
    they have even degrees in \(G_{t}'\). It is straightforward to verify that
    \(G_{t}'\) is a witness for \(((t, X, O), P)\). Thus the addition of
    partition \(P\) to \calA preserves the soundness of \calA.

    Now we prove completeness. So let \(H\) be a residual subgraph with
    respect to \(t\) with \(V(H) \cap X_{t} = X\), for which there exists a
    partition \(P=\{X^{1},X^{2},\ldots X^{p}\}\) of \(X\) such that
    \(((t, X, O), P)\) completes \(H\).  We
    need to show that the set \calA computed by the algorithm contains some
    partition \(Q\) of \(X\) such that \(((t, X, O), Q)\) completes \(H\).
    Observe that there exists a subgraph \(G'_{t}\) of \(G_{t}\)---a witness
    for \(((t,X,O), P)\)---such that the following hold:
    \begin{enumerate}
    \item \(X_{t} \cap V(G'_{t})=X\).
    \item \(G'_{t}\) has exactly \(p\) connected components
      \(C_{1}, C_{2}, \dotsc, C_{p}\) and for each
      \(i \in \{1, 2, \dotsc, p\}$, $X^{i} \subseteq V(C_{i})\) holds.
    \item Every terminal vertex from \(K \cap V_{t}\) is in \(V(G'_{t})\).
    \item The set of odd-degree vertices in \(G'_{t}\) is exactly the set
      \(O\).
    \item The graph \(G'_{t} \cup H\) is an
      Eulerian Steiner subgraph of \(G\) for the terminal set \(K\).
    \end{enumerate}
    Note that by the definition of a residual subgraph, graph \(H\) (i) does
    \emph{not} contain edge \(uv\), and (ii) is a residual subgraph with
    respect to node \(t'\) as well. We consider two cases.
    \begin{itemize}
    \item Suppose edge \(uv\) is not present in graph \(G'_{t}\). Then it is
      straightforward to verify that \(G'_{t}\) is a witness for
      \(((\mathbf{t'},X,O), P)\) as well. By the inductive hypothesis there
      exists some partition \(Q\) of \(X\) in the set \(VP[t', X, O]\) such
      that \(((t, X, O), Q)\) completes \(H\). This same partition \(Q\) is
      present in the set \calA as well.
    \item Suppose edge \(uv\) \emph{is} present in graph \(G'_{t}\). 
      Let \(H' = (V(H), (E(H) \cup \{uv\}))\) be the graph obtained by adding
      edge \(uv\) to graph \(H\), and let
      \(G'_{t'} = (V(G'_{t}), (E(G'_{t}) \setminus \{uv\}))\) be the graph
      obtained by deleting edge \(uv\) from graph \(G'_{t}\). Then it is
      straightforward to verify that (i) the set of odd-degree vertices in
      \(G'_{t'}\) is exactly the set \(O \setminus \{u, v\}\), (ii) \(H'\) is a
      residual subgraph for node \(t'\), and (iii) \(G'_{t'}\) is a subgraph of
      \(G_{t'}\) such that the graph \(G'_{t'} \cup H' = G'_{t} \cup H\) is an
      Eulerian Steiner subgraph of \(G\) for the
      terminal set \(K\). Let \(P'\) be the partition of \(X\) defined by
      graph \(G'_{t'}\). Then \(G'_{t'}\) is a witness for
      \(((t', X, O \setminus \{u, v\}), P')\) such that the union of
      \(G'_{t'}\) and the residual subgraph \(H'\) of \(t'\) is an Eulerian
      Steiner subgraph of \(G\) for the terminal set \(K\). That is,
      \(((t', X, O \setminus \{u, v\}), P')\) completes \(H'\). So by the
      inductive assumption there exists some partition \(Q'\) of \(X\) in the
      set \(VP[t', X, O \setminus \{u, v\}]\) such that
      \(((t, X, O \setminus \{u, v\}), Q')\) completes \(H'\). So there exists
      a subgraph \(\hat{G}'\) of \(G_{t'}\) such that (i) \(\hat{G}'\) is a
      witness for \(((t, X, O \setminus \{u, v\}), Q')\) and (ii)
      \(\hat{G}' \cup H'\) is an Eulerian
      Steiner subgraph of \(G\) for the terminal set \(K\).

      Note that \(Q'\) is the partition of set \(X\) defined by the graph
      \(\hat{G}'\), and that the set of odd-degree vertices in \(\hat{G}'\) is
      exactly the set \(O \setminus \{u, v\}\). Suppose both \(u\) and \(v\)
      are in the same block of partition \(Q'\). Then adding the edge \(uv\)
      to \(\hat{G}'\) (i) does not change the partition of \(X\) defined by
      \(\hat{G}'\), and (ii) \emph{does} change the set of odd-degree vertices
      to \(O\). It follows that the graph
      \(\hat{G} = (V(\hat{G}'), E(\hat{G}') \cup \{uv\})\) is a subgraph of
      \(G_{t}\) such that (i) \(\hat{G}\) is a witness for \(((t, X, O , Q')\)
      and (ii) \(\hat{G} \cup H\) is an Eulerian
      Steiner subgraph of \(G\) for the terminal set \(K\). Thus
      \(((t, X, O, Q')\) completes the residual subgraph \(H\). Now notice
      that our algorithm adds the partition \(Q'\) to the set \calA. Thus the
      completeness criterion holds in this case.

      In the remaining case, vertices \(u\) and \(v\) are in distinct blocks of
      partition \(Q'\). Let \(Q\) be the partition obtained from \(Q'\) by
      merging together the two blocks to which vertices \(u\) and \(v\) belong,
      respectively, and leaving the other blocks as they are. Let \(\hat{G}\) be
      defined as in the previous paragraph. Then the partition of \(X\) defined
      by \(\hat{G}\) is \(Q\). It follows that \(\hat{G}\) is a subgraph of
      \(G_{t}\) such that (i) \(\hat{G}\) is a witness for \(((t, X, O , Q)\)
      and (ii) \(\hat{G} \cup H\) is an Eulerian
      Steiner subgraph of \(G\) for the terminal set \(K\). Thus
      \(((t, X, O, Q)\) completes the residual subgraph \(H\). Now notice that
      our algorithm adds the partition \(Q\) to the set \calA. Thus the
      completeness criterion holds in this case as well. \qedhere
    \end{itemize}
  \end{enumerate}
\end{proof}

\begin{lemma}\label{lem:forget_node_ok}
  Let \(t\) be a forget node of the tree decomposition \TT and let
  \(X \subseteq X_{t}, O \subseteq X\) be arbitrary subsets of \(X_{t}, X\)
  respectively. The collection \calA of partitions computed by the DP for the
  combination \((t, X, O)\) satisfies the correctness criteria.
\end{lemma}
\begin{proof}
  Let \(t'\) be the child node of \(t\), and let \(v\) be the vertex forgotten
  at \(t\). Then \(v \in X_{t'}\) and \(X_{t} = X_{t'} \setminus \{v\}\), and
  \(v \notin O\) hold. Recall that \(P(v)\) is the block of partition \(P\)
  which contains element \(v\) and that \(P - v\) is the partition obtained by
  eliding \(v\) from \(P\). The algorithm initializes
  \(\calA = \{P' - v \;;\; P' \in VP[t', X \cup \{v\}, O],\,|P'(v)| > 1\}\). By
  the inductive assumption we have that every partition
  \(P' \in VP[t', X \cup \{v\}, O]\) is valid for the combination
  \((t', X \cup \{v\}, O)\). Note that (i) the graph \(G_{t'}\) is identical to
  the graph \(G_{t}\), and (ii) for any subgraph \(H\) of \(G_{t'} = G_{t}\),
  \((V(H) \cap X_{t'}) = X \cup \{v\}\) implies \((V(H) \cap X_{t}) = X\). It
  follows that if every connected component of a graph \(H\) contains at least
  two vertices from the set \(X \cup \{v\}\) then every connected component of
  \(H\) contains at least one vertex from set \(X\). Using these observations it
  is straightforward to verify that if a subgraph \(G_{t'}'\) of \(G_{t'}\) is a
  witness for \(((t', X \cup \{v\}, O), P')\) where \(v \notin O\) and
  \(|P'(v)| > 1\) hold, then it is also (i) a subgraph of \(G_{t}\), and (ii) a
  witness for \(((t, X, O), P' - v)\). Thus for each partition
  \(P' \in VP[t', X \cup \{v\}, O],\,|P'(v)| > 1\) the partition \(P' - v\) is
  valid for the combination \((t, X, O)\). Hence directly after the
  initialization, all partitions in the set \calA are valid for \((t, X, O)\).

  The algorithm adds zero or more partitions to \(\calA\) depending on whether
  vertex \(v\) is a terminal or not. We analyze each choice made by the
  algorithm:
  \begin{enumerate}
  \item If \(v\) is a terminal vertex then the algorithm does not make further
    changes to \(\calA\). We have shown above that this set \calA satisfies the
    validity criterion. We now argue that it satisfies the completeness
    criterion as well.

    So let \(H\) be a residual subgraph with respect to \(t\) with
    \(V(H) \cap X_{t} = X\), for which there exists a partition
    \(P=\{X^{1},X^{2},\ldots X^{p}\}\) of \(X\) such that \(((t, X, O), P)\)
    completes \(H\). We need to show that the set \calA computed by the
    algorithm contains some partition \(Q\) of \(X\) such that
    \(((t, X, O), Q)\) completes \(H\). Observe that there exists a subgraph
    \(G'_{t}\) of \(G_{t}\)---a witness for \(((t,X,O), P)\)---such that the
    following hold:
    \begin{enumerate}
    \item \(X_{t} \cap V(G'_{t})=X\).
    \item \(G'_{t}\) has exactly \(p\) connected components
      \(C_{1}, C_{2}, \dotsc, C_{p}\) and for each
      \(i \in \{1, 2, \dotsc, p\}$, $X^{i} \subseteq V(C_{i})\) holds.
    \item Every terminal vertex from \(K \cap V_{t}\) is in \(V(G'_{t})\).
    \item The set of odd-degree vertices in \(G'_{t}\) is exactly the set \(O\).
    \item The graph \(G'_{t} \cup H\) is an
      Eulerian Steiner subgraph of \(G\) for the terminal set \(K\).
    \end{enumerate}
    From the definition of a residual subgraph we know that \(v \notin V(H)\)
    holds, and since \(v\) is a terminal vertex, from condition (c) above we get
    that \(v \in V(G'_{t})\) holds. Without loss of generality, let it be the
    case that \(v \in C_{p}\) holds. Since \(X_{t'} = X_{t} \cup \{v\}\) we get
    that \(X_{t'} \cap V(G'_{t}) = X \cup \{v\}\) holds. Let
    \(H' = (V(H) \cup \{v\}, E(H))\) be the graph obtained by adding vertex
    \(v\) (and no extra edges) to graph \(H\). Then it is straightforward to
    verify that (i) \(H'\) is a residual subgraph with respect to \(t'\) with
    \(V(H') \cap X_{t'} = X \cup \{v\}\), (ii) the graph \(G'_{t}\) is a witness
    for the partition \(P' = \{X^{1},X^{2},\ldots (X^{p} \cup \{v\}) \}\) of
    \(X \cup \{v\}\) being valid for the combination \((t', X \cup \{v\}, O)\),
    and (iii) the graph \(G'_{t} \cup H'\) is an
    Eulerian Steiner subgraph of \(G\) for the terminal set \(K\). That is,
    \(((t', X \cup \{v\}, O), P')\) completes \(H'\).

    By the inductive assumption there exists some partition \(Q'\) of
    \(X \cup \{v\}\) in the set \(VP[t', X \cup \{v\}, O]\}\) such that
    \(((t', X \cup \{v\}, O), Q')\) completes \(H'\). So there exists a subgraph
    \(\hat{G}'\) of \(G_{t'}\) such that (i) \(\hat{G}'\) is a witness for
    \(((t', X \cup \{v\}, O), Q')\) and (ii)
    \(\hat{G}' \cup H'\) is an Eulerian Steiner
    subgraph of \(G\) for the terminal set \(K\). Note that
    \(X_{t} \cap V(\hat{G}') = X\) holds.

    Since \(v\) had degree zero in graph \(H'\) we get that \(v\) has a positive
    even degree in \(\hat{G}'\). From the definition of a witness for
    validity---\autoref{def:valid_partitions_witnesses}---we get that \(Q'\) is
    the partition of the set \(X \cup \{v\}\) defined by the graph \(\hat{G}'\).
    Let \(Q_{H'}\) be the partition of the set \(X \cup \{v\}\) defined by the
    graph \(H'\). Since \(deg_{H'}(v) = 0\) holds we get that vertex \(v\)
    appears in a block of size one---namely, \(\{v\}\)---in \(Q_{H'}\). If
    \(\{v\}\) is a block of \(Q'\) as well, then \(\{v\}\) will also be a block
    in their join \(Q_{H'} \sqcup Q'\). But the union of graphs \(H'\) and
    \(\hat{G}'\) is connected and so from
    \autoref{lem:partition_join_connectivity} we know that
    \(Q_{H'} \sqcup Q' = \{\{X \cup \{v\}\}\}\). Thus \(\{v\}\) is \emph{not} a
    block of \(Q_{H'} \sqcup Q'\), or of \(Q'\). So there exists a vertex
    \(v' \in X \) such that \(v,v'\) are in the same block of \(Q'\). In
    particular, this implies that the partition \(Q = Q' - v\), which is the
    partition of set \(X\) defined by graph \(\hat{G}'\), has exactly as many
    blocks as has the partition \(Q'\) of \(X \cup \{v\}\).

    Putting these together we get that the subgraph \(\hat{G}'\) of \(G_{t}\) is
    a witness for \(((t, X, O), Q = Q' - v)\). Now since graph \(H\) can be
    obtained from graph \(H'\) by deleting vertex \(v\), we get that the graphs
    \(\hat{G}' \cup H'\) and \(\hat{G}' \cup H\) are identical. In
    particular, the latter is an Eulerian Steiner subgraph of \(G\) for the
    terminal set \(K\). Thus \(((t, X, O), Q)\) completes the residual graph
    \(H\). Since the algorithm adds partition \(Q\) to the set \calA, we get
    that \calA satisfies the completeness criterion.
    
  \item If \(v\) is not a terminal vertex then the algorithm adds all the
    partitions from \(VP[t', X, O]\) to \(\calA\). By the inductive assumption
    we have that every partition \(P' \in VP[t', X, O]\) is valid for the
    combination \((t', X, O)\). It is once again straightforward to verify that
    if a subgraph \(G_{t'}'\) of \(G_{t'}\) is a witness for
    \(((t', X, O), P')\) then it is also (i) a subgraph of \(G_{t}\), and (ii) a
    witness for \(((t, X, O), P')\). Thus each partition \(P' \in VP[t', X, O]\)
    is valid for the combination \((t, X, O)\). Hence all partitions added to
    the set \calA in this step are valid for \((t, X, O)\).

    We now argue that the set \calA satisfies the completeness criterion. So let
    \(H\) be a residual subgraph with respect to \(t\) with
    \(V(H) \cap X_{t} = X\), for which there exists a partition
    \(P=\{X^{1},X^{2},\ldots X^{p}\}\) of \(X\) such that \(((t, X, O), P)\)
    completes \(H\). We need to show that the set \calA computed by the
    algorithm contains some partition \(Q\) of \(X\) such that
    \(((t, X, O), Q)\) completes \(H\). Observe that there exists a subgraph
    \(G'_{t}\) of \(G_{t}\)---a witness for \(((t,X,O), P)\)---such that the
    following hold:
    \begin{enumerate}
    \item \(X_{t} \cap V(G'_{t})=X\).
    \item \(G'_{t}\) has exactly \(p\) connected components
      \(C_{1}, C_{2}, \dotsc, C_{p}\) and for each
      \(i \in \{1, 2, \dotsc, p\}$, $X^{i} \subseteq V(C_{i})\) holds.
    \item Every terminal vertex from \(K \cap V_{t}\) is in \(V(G'_{t})\).
    \item The set of odd-degree vertices in \(G'_{t}\) is exactly the set \(O\).
    \item The graph \(G'_{t} \cup H\) is an
      Eulerian Steiner subgraph of \(G\) for the terminal set \(K\).
    \end{enumerate}

    Suppose graph \(G'_{t}\) does \emph{not} contain vertex \(v\). Then it is
    easy to verify that \(H\) is a residual subgraph with respect to \(t'\) with
    \(V(H) \cap X_{t'} = X\), and that graph \(G'_{t}\) is a witness for
    \(((t',X,O), P)\) such that the union of graphs \(H\) and \(G'_{t}\) is an
    Eulerian Steiner subgraph of \(G\) for the terminal set \(K\). That is,
    \(((t', X, O), P)\) completes \(H\). By inductive assumption there exists a
    partition \(Q \in VP[t', X, O]\) such that \(((t', X, O), Q)\) completes
    \(H\). Since the algorithm adds this partition \(Q\) to \calA we get that
    \calA satisfies the completeness criterion in this case.

    Now suppose graph \(G'_{t}\) contains vertex \(v\). The analysis from the
    case where \(v\) was a terminal and was thus forced to be in graph
    \(G'_{t}\), applies verbatim in this case. Note that the set \calA in the
    present case is a \emph{superset} of the set \calA computed in that case.
    Thus we get that the current set \calA satisfies the completeness
    criterion.\qedhere
  \end{enumerate}
\end{proof}

\begin{lemma}\label{lem:join_node_ok}
  Let \(t\) be a join node of the tree decomposition \TT and let
  \(X \subseteq X_{t}, O \subseteq X\) be arbitrary subsets of \(X_{t}, X\)
  respectively. The collection \calA of partitions computed by the DP for the
  combination \((t, X, O)\) satisfies the correctness criteria.
\end{lemma}
\begin{proof}
  Let \(t_{1}, t_{2}\) be the children of \(t\). Then
  \(X_{t} = X_{t_{1}} = X_{t_{2}}\). Note that
  \(V(G_{t}) = V(G_{t_{1}}) \cup V(G_{t_{2}}) \) and
  \(E(G_{t}) = E(G_{t_{1}}) \cup E(G_{t_{2}}) \) hold, and so graph \(G_{t}\) is
  the union of graphs \(G_{t_{1}}\) and \(G_{t_{2}}\). Further, since each edge
  in the graph is introduced at exactly one bag in \TT we get that
  \(E(G_{t_{1}}) \cap E(G_{t_{2}}) = \emptyset\) holds. Moreover,
  \(V(G_{t_{1}}) \cap V(G_{t_{2}}) = X_{t}\) holds as well. The algorithm
  initializes \calA to the empty set. For each way of dividing set \(O\) into
  two disjoint subsets \(O_{1},O_{2}\) (one of which could be empty) and for
  each subset \(\hat{O}\) (which could also be empty) of the set
  \(X \setminus O\), the algorithm picks a number of pairs \((P_{1}, P_{2})\) of
  partitions and adds their joins \(P_{1} \sqcup P_{2}\) to the set \calA. We
  first show that the partition \(P_{1} \sqcup P_{2}\) is valid for the
  combination \((t, X, O)\), for each choice of pairs \((P_{1}, P_{2})\) made by
  the algorithm.

  So let
  \(P_{1} \in VP[t_{1}, X, O_{1} \cup \hat{O}], P_{2} \in VP[t_{2}, X, O_{2}
  \cup \hat{O}]\). By the inductive hypothesis we get that \(P_{1}\) is valid
  for the combination \((t_{1}, X, O_{1} \cup \hat{O})\) and \(P_{2}\) is valid
  for the combination \((t_{2}, X, O_{2} \cup \hat{O})\). So there exist
  subgraphs \(G'_{t_{1}} = (V'_{t_{1}},E'_{t_{1}})\) of \(G_{t_{1}}\) and
  \(G'_{t_{2}} = (V'_{t_{2}},E'_{t_{2}})\) of \(G_{t_{2}}\) such that
    \begin{enumerate}
    \item \(X_{t} \cap V'_{t_{1}} = X = X_{t} \cap V'_{t_{2}}\);
    \item The vertex set of each connected component of \(G'_{t_{1}}\) and of
      \(G'_{t_{2}}\) has a non-empty intersection with set \(X\). Moreover,
      \(P_{1}\) is the partition of \(X\) defined by the subgraph \(G'_{t_{1}}\)
      and \(P_{2}\) is the partition of \(X\) defined by the subgraph
      \(G'_{t_{2}}\);
    \item Every terminal vertex from \(K \cap V(G_{t_{1}})\) is in
      \(V'_{t_{1}}\) and every terminal vertex from \(K \cap V(G_{t_{2}})\) is
      in \(V'_{t_{2}}\); and,
    \item The set of odd-degree vertices in \(G'_{t_{1}}\) is exactly the set
      \(O_{1} \cup \hat{O}\) and the set of odd-degree vertices in
      \(G'_{t_{2}}\) is exactly the set \(O_{2} \cup \hat{O}\).
    \end{enumerate}
    Let \(G'_{t} = G'_{t_{1}} \cup G'_{t_{2}}\). Then \(G'_{t}\) is a subgraph
    of \(G_{t}\), and
    \begin{enumerate}
    \item Since \(X_{t} \cap V'_{t_{1}} = X = X_{t} \cap V'_{t_{2}}\) holds we
      have that \(X_{t} \cap V(G'_{t})=X\) holds as well;
    \item The vertex set of each connected component of \(G'_{t}\) has a
      non-empty intersection with set \(X\). Moreover, from
      \autoref{lem:partition_join_connectivity} we get that
      \(P_{1} \sqcup P_{2}\) is the partition of \(X\) defined by the subgraph
      \(G'_{t}\);
    \item Every terminal vertex from the set \(K \cap V(G_{t})\) is in
      \(V(G'_{t})\); and,
    \item Since \(E(G_{t_{1}}) \cap E(G_{t_{2}}) = \emptyset\) holds we get that
      the degree of any vertex \(v\) in graph \({G'_{t}}\) is the sum of its
      degrees in the two graphs \(G'_{t_{1}}\) and \(G'_{t_{2}}\). Since (i) the
      set of odd-degree vertices in graph \(G'_{t_{1}}\) is exactly the set
      \(O_{1} \cup \hat{O}\), (ii) the set of odd-degree vertices in graph
      \(G'_{t_{2}}\) is exactly the set \(O_{2} \cup \hat{O}\), and (iii) \(O\)
      is the disjoint union of sets \(O_{1}\) and \(O_{2}\), we get that the set
      of odd-degree vertices in graph \(G'_{t}\) is exactly the set \(O\).
    \end{enumerate}
    Thus graph \(G'_{t}\) is a witness for partition \(P_{1} \sqcup P_{2}\)
    being valid for the combination \((t, X, O)\), and so partition
    \(P_{1} \sqcup P_{2} \in \calA\) is valid for the combination \((t, X, O)\).
    This proves that collection \calA satisfies the soundness criterion.

    We now argue that the set \calA satisfies the completeness criterion. So let
    \(H\) be a residual subgraph with respect to \(t\) with
    \(V(H) \cap X_{t} = X\), for which there exists a partition
    \(P=\{X^{1},X^{2},\ldots X^{p}\}\) of \(X\) such that \(((t, X, O), P)\)
    completes \(H\). We need to show that the set \calA computed by the
    algorithm contains some partition \(Q\) of \(X\) such that
    \(((t, X, O), Q)\) completes \(H\). Observe that there exists a subgraph
    \(G'_{t}\) of \(G_{t}\)---a witness for \(((t,X,O), P)\)---such that the
    following hold:
    \begin{enumerate}
    \item \(X_{t} \cap V(G'_{t})=X\).
    \item \(G'_{t}\) has exactly \(p\) connected components
      \(C_{1}, C_{2}, \dotsc, C_{p}\) and for each
      \(i \in \{1, 2, \dotsc, p\}$, $X^{i} \subseteq V(C_{i})\) holds.
    \item Every terminal vertex from \(K \cap V_{t}\) is in \(V(G'_{t})\).
    \item The set of odd-degree vertices in \(G'_{t}\) is exactly the set \(O\).
    \item The graph \(G'_{t} \cup H\) is an
      Eulerian Steiner subgraph of \(G\) for the terminal set \(K\).
    \end{enumerate}
    Let \(G_{1} = (V(G'_{t}) \cap V(G_{t_{1}}), E(G'_{t}) \cap E(G_{t_{1}}))\)
    and \(G_{2} = (V(G'_{t}) \cap V(G_{t_{2}}), E(G'_{t}) \cap E(G_{t_{2}}))\)
    be, respectively, the subgraphs of \(G'_{t}\) defined by the subtrees of \TT
    rooted at nodes \(t_{1}\) and \(t_{2}\), respectively. Then
    \(G'_{t} = G_{1} \cup G_{2}\),
    \(V(G_{1}) \cap X_{t_{1}} = V(G_{2}) \cap X_{t_{2}}= V(G_{1}) \cap V(G_{2})
    = X\), and \(E(G_{1}) \cap E(G_{2}) = \emptyset\) all hold. Let
    \(\tilde{O_{1}}, \tilde{O_{2}}\) be the sets of vertices of odd degree in
    graphs \(G_{1}, G_{2}\), respectively. Since graph
    \((H \cup G_{1}) \cup G_{2}\) is Eulerian and since
    \(V(H \cup G_{1}) \cap V(G_{2}) = X\) holds, we get that (i)
    \(\tilde{O_{2}} \subseteq X\) holds, and (ii) every connected component of
    graph \(G_{2}\) contains at least one vertex from set \(X\). By symmetric
    reasoning we get that (i) \(\tilde{O_{1}} \subseteq X\) holds, and (ii)
    every connected component of graph \(G_{1}\) contains at least one vertex
    from set \(X\). Let \(O_{2} = \tilde{O_{2}} \cap O\) and
    \(\hat{O} = \tilde{O_{2}} \setminus O\). Then
    \(\tilde{O_{2}} = O_{2} \cup \hat{O}\). Define
    \(O_{1} = O \setminus O_{2}\). Since (i) the set of odd-degree vertices in
    graph \(G'_{t}\) is exactly the set \(O\), and (ii)
    \(E(G_{1}) \cap E(G_{2}) = \emptyset\) holds, we get that the set of
    odd-degree vertices in graph \(G_{1}\) is
    \(\tilde{O_{1}} =(O \setminus O_{2}) \cup \hat{O} = O_{1} \cup \hat{O}\).

    Let \(Q_{2}\) be the partition of set \(X\) defined by graph \(G_{2}\), and
    let \(R_{1} = H \cup G_{1}\). It is straightforward to verify the following:
    (i) \(R_{1}\) is a residual subgraph with respect to node \(t_{2}\) with
    \(V(R_{1}) \cap X_{t_{2}} = X\); (ii) graph \(G_{2}\) is a witness for
    partition \(Q_{2}\) being valid for the combination
    \((t_{2}, X, \tilde{O_{2}})\), and (iii) \(G_{2}\) is a certificate for
    \(((t_{2}, X, \tilde{O_{2}}), Q_{2})\) completing the residual graph
    \(R_{1}\). By the inductive assumption there is a partition \(P_{2}\) of
    \(X\) in the set \(VP[t_{2}, X, O_{2} \cup \hat{O}]\) such that
    \(((t_{2}, X, O_{2} \cup \hat{O}), P_{2})\) completes the residual graph
    \(R_{1}\). Let \(H_{2}\) be a certificate for
    \(((t_{2}, X, O_{2} \cup \hat{O}), P_{2})\) completing \(R_{1}\). Note that
    \(H_{2}\) is a subgraph of \(G_{t_{2}}\), and that
    \(R_{1} \cup H_{2} = (H \cup G_{1}) \cup H_{2}\) is an Eulerian Steiner
    subgraph of \(G\) for the terminal set \(K\).

    Let \(Q_{1}\) be the partition of set \(X\) defined by graph \(G_{1}\), and
    let \(R_{2} = H \cup H_{2}\). From \autoref{lem:completion_odd_subset} we
    get that the set of odd-degree vertices of the residual subgraph \(H\) is
    exactly the set \(O\), and from
    Definitions~\ref{def:valid_partitions_witnesses} and~\ref{def:completion} we
    get that the set of odd-degree vertices of graph \(H_{2}\) is the set
    \(O_{2} \cup \hat{O}\). From the definition of a residual subgraph we get
    that \(E(H) \cap E(H_{2}) = \emptyset\) holds. It follows that the set of
    odd-degree vertices of graph \(R_{2}\) is
    \((O \setminus O_{2}) \cup \hat{O} = O_{1} \cup \hat{O}\), which is exactly
    the set of odd-degree vertices of graph \(G_{1}\).

    It is now straightforward to verify the following: (i) \(R_{2}\) is a
    residual subgraph with respect to node \(t_{1}\) with
    \(V(R_{2}) \cap X_{t_{1}} = X\); (ii) graph \(G_{1}\) is a witness for
    partition \(Q_{1}\) being valid for the combination
    \((t_{1}, X, O_{1} \cup \hat{O})\), and (iii) \(G_{1}\) is a certificate for
    \(((t_{1}, X, O_{1} \cup \hat{O}), Q_{1})\) completing the residual graph
    \(R_{2}\). By the inductive assumption there is a partition \(P_{1}\) of
    \(X\) in the set \(VP[t_{1}, X, O_{1} \cup \hat{O}]\) such that
    \(((t_{1}, X, O_{1} \cup \hat{O}), P_{1})\) completes the residual graph
    \(R_{2}\). Let \(H_{1}\) be a certificate for
    \(((t_{1}, X, O_{1} \cup \hat{O}), P_{1})\) completing \(R_{2}\). Note that
    \(H_{1}\) is a subgraph of \(G_{t_{1}}\), and that
    \(R_{2} \cup H_{1} = (H \cup H_{2}) \cup H_{1}\) is an Eulerian Steiner
    subgraph of \(G\) for the terminal set \(K\).

    Let \(\hat{H} = H_{1} \cup H_{2}\). Then \(\hat{H}\) is a subgraph of
    \(G_{t}\), and
    \begin{enumerate}
    \item Since \(X_{t} \cap V(H_{1}) = X = X_{t} \cap V(H_{2})\) holds we have
      that \(X_{t} \cap V(\hat{H})=X\) holds as well;
    \item The vertex set of each connected component of \(\hat{H}\) has a
      non-empty intersection with set \(X\). Moreover, from
      \autoref{lem:partition_join_connectivity} we get that
      \(P_{1} \sqcup P_{2}\) is the partition of \(X\) defined by the subgraph
      \(\hat{H}\);
    \item Every terminal vertex from the set \(K \cap V(G_{t})\) is in
      \(V(\hat{H})\); and,
    \item Since \(E(G_{t_{1}}) \cap E(G_{t_{2}}) = \emptyset\) holds we get that
      the degree of any vertex \(v\) in graph \(\hat{H}\) is the sum of its
      degrees in the two graphs \(H_{1}\) and \(H_{2}\). Since (i) the set of
      odd-degree vertices in graph \(H_{1}\) is exactly the set
      \(O_{1} \cup \hat{O}\), (ii) the set of odd-degree vertices in graph
      \(H_{2}\) is exactly the set \(O_{2} \cup \hat{O}\), and (iii) \(O\) is
      the disjoint union of sets \(O_{1}\) and \(O_{2}\), we get that the set of
      odd-degree vertices in graph \(\hat{H}\) is exactly the set \(O\).
    \end{enumerate}
    Graph \(\hat{H}\) is thus a witness for partition \(P_{1} \sqcup P_{2}\) of
    \(X\) being valid for the combination \((t, X, O)\), and \(H \cup \hat{H}\)
    is an Eulerian Steiner subgraph of \(G\) for the terminal set \(K\). Thus
    \(((t, X, O), P_{1} \sqcup P_{2})\) completes \(H\). Since the algorithm
    adds partition \(P_{1} \sqcup P_{2}\) to the set \(\calA\) we get that \calA
    satisfies the completeness criterion. \qedhere
\end{proof}

We can now prove
\begingroup
\def\thetheorem{\ref{thm:ESS_is_FPT}}
\begin{theorem}
  There is an algorithm which solves an instance \((G,K, \mathcal{T},tw)\) of
  \ESS in \(\OhStar{(1 + 2^{(\omega + 3)})^{tw}}\) time.
\end{theorem}
\addtocounter{theorem}{-1}
\endgroup
\begin{proof}
  We first modify \TT to make it a ``nearly-nice'' tree decomposition rooted at
  \(r\) as described at the start of this section. We then execute the dynamic
  programming steps described above on \TT. We return \yes if the element
  \(\{\{v^{\star}\}\}\) is present in the set
  \(VP[r,X = \{v^{\star}\}, O = \emptyset]\) computed by the DP, and \no
  otherwise.

  From \autoref{lem:completion_at_root} we know that \((G, K, \TT, tw)\) is a
  \yes instance of \ESS if and only if the combination
  \(((r,X = \{v^{\star}\}, O = \emptyset), P = \{\{v^{\star}\}\})\) completes
  the residual graph \(H = (\{v^{\star}\}, \emptyset)\). By induction on the
  structure of the tree decomposition \TT and using
  \autoref{obs:repset_computation_preserves_correctness} and
  Lemmas~\ref{lem:leaf_node_ok},~\ref{lem:introduce_vertex_node_ok},~\ref{lem:introduce_edge_node_ok},~\ref{lem:forget_node_ok},
  and~\ref{lem:join_node_ok} we get that the set
  \(VP[r,X = \{v^{\star}\}, O = \emptyset]\) computed by the algorithm satisfies
  the correctness criteria. And since \(\{\{v^{\star}\}\}\) is the unique
  partition of set \(\{v^{\star}\}\) we get that the set
  \(VP[r,X = \{v^{\star}\}, O = \emptyset]\) computed by the algorithm will
  contain the partition \(\{\{v^{\star}\}\}\) if and only if \((G, K, \TT, tw)\)
  is a \yes instance of \ESS.

  Note that we compute representative subsets as the last step in the
  computation at each bag. So we get, while performing computations at an
  intermediate node \(t\), that the number of partitions in any set
  \(VP[t', X', \cdot]\) for any \emph{child} node \(t'\) of \(t\) and subset
  \(X'\) of \(X_{t'}\) is at most \(2^{(|X'| - 1)}\) (See
  \autoref{thm:computing_representative_subsets}). We use
  \autoref{fac:partition_join_is_union_of_graphs} to perform various operations
  on one or two partitions---such as adding a block to a partition, merging two
  blocks of a partition, eliding an element from a partition, or computing the
  join of two partitions---in polynomial time.

  The computation at each \textbf{leaf node} of \TT can be done in constant
  time. For an \textbf{introduce vertex node} or an \textbf{introduce edge node}
  or a \textbf{forget node} \(t\) and a fixed pair of subsets
  \(X \subseteq X_{t}, O \subseteq X\), the computation of set \calA
  involves---in the worst case---spending polynomial time for each partition
  \(P'\) in some set \(VP[t', X' \subseteq X, \cdot]\). Since the number of
  partitions in this latter set is at most \(2^{(|X'| - 1)} \leq 2^{(|X| - 1)}\)
  we get that the set \calA can be computed in \(\OhStar{2^{(|X| - 1)}}\) time,
  and that the set \calB can be computed---see
  \autoref{thm:computing_representative_subsets}---in
  \(\OhStar{2^{(|X| - 1)} \cdot 2^{(\omega - 1)\cdot |X|}} = \OhStar{2^{\omega
      \cdot |X|}}\) time. Since the number of ways of choosing the subset
  \(O \subseteq X\) is \(2^{|X|}\) the entire computation at an introduce
  vertex, introduce edge, or forget node \(t\) can be done in time
  \begin{align*}
    \sum_{|X| = 0}^{|X_{t}|}\binom{|X_{t}|}{|X|} 2^{|X|} \OhStar{2^{\omega \cdot
    |X|}} &= \OhStar{\sum_{|X| = 0}^{tw + 1}\binom{tw + 1}{|X|} 2^{(\omega + 1)|X|}}\\
          &= \OhStar{(1 + 2^{(\omega + 1)})^{(tw + 1)}}\\
          &= \OhStar{(1 + 2\cdot 2^{\omega })^{tw}}. 
  \end{align*}
  For a \textbf{join node} \(t\) and a fixed subset \(X \subseteq X_{t}\) we
  guess three pairwise disjoint subsets \(\hat{O}, O_{1}, O_{2}\) of \(X\) in
  time \(4^{|X|}\). For each guess we go over all partitions
  \(P_{1} \in VP[t_{1}, X, O_{1} \cup \hat{O}], P_{2} \in VP[t_{2}, X, O_{2}
  \cup \hat{O}]\) and add their join \(P_{1} \sqcup P_{2}\) to the set \calA.
  Since the number of partitions in each of the two sets
  \(VP[t_{1}, X, O_{1} \cup \hat{O}], VP[t_{2}, X, O_{2} \cup \hat{O}]\) is at
  most \(2^{(|X| - 1)}\), the size of set \calA is at most \(2^{(2|X| - 2)}\).
  The entire computation at the join node can be done in time
  \begin{align*}
    \sum_{|X| = 0}^{|X_{t}|}\binom{|X_{t}|}{|X|} 4^{|X|} (2^{(2|X| - 2)} + \OhStar{2^{(2|X| - 2)} \cdot 2^{(\omega - 1)\cdot |X|}})
    &= \OhStar{\sum_{|X| = 0}^{tw + 1}\binom{tw + 1}{|X|} 2^{4|X| - 2 + \omega{}|X| - |X|}}\\
    &= \OhStar{\sum_{|X| = 0}^{tw + 1}\binom{tw + 1}{|X|} 2^{(\omega + 3)|X|}}\\
          &= \OhStar{(1 + 2^{(\omega + 3)})^{(tw + 1)}}\\
          &= \OhStar{(1 + 2^{(\omega + 3)})^{tw}}.
  \end{align*}
  The entire DP over \TT can thus be done in
  \(\OhStar{(1 + 2^{(\omega + 3)})^{tw}}\) time.
\end{proof}


%% file: hamIndex.tex
\section{Finding the Hamiltonian Index}\label{sec:hamIndex}
In this section we prove \autoref{thm:hamIndex_is_FPT}: we describe an algorithm
which takes an instance \((G, \mathcal{T}, tw, r)\) of \HI as input and outputs
in \(\OhStar{(1 + 2^{(\omega + 3)})^{tw}}\) time whether graph \(G\) has
Hamiltonian Index at most \(r\). If \(r \geq (|V(G)| - 3)\) holds then our
algorithm returns \yes. If \(r < (|V(G)| - 3)\) then it checks, for each
\(i = 0, 1, \dotsc, r\) in increasing order, whether \(h(G) = i\) holds. From
\autoref{fac:all_graphs_have_finite_hamiltonian_index} we know that this
procedure correctly solves \HI. We now describe how we check if \(h(G) = i\)
holds for increasing values of \(i\). For \(i = 0\) we apply an algorithm of
Bodlaender et al., and for \(i = 1\) we leverage a classical result of Harary
and Nash-Williams.
\begin{theorem}\textup{\cite{bodlaenderCyganKratschNederlof2015deterministic}}\label{fac:graph_hamiltonicity_tw_fpt}
  There is an algorithm which takes a graph \(G\) and a tree decomposition of
  \(G\) of width \(tw\) as input, runs in
  \(\OhStar{(5 + 2^{(\omega + 2)/2})^{tw}}\) time, and tells whether \(G\) is
  Hamiltonian.
\end{theorem}
\begin{theorem}\textup{\cite{harary1965eulerian}}\label{fac:line_graph_hamiltonicity_DES}
  Let \(G\) be a connected graph with at least three edges. Then \(L(G)\) is
  Hamiltonian if and only if \(G\) has a dominating Eulerian subgraph.
\end{theorem}

For checking if \(h(G) \in \{2, 3\}\) holds we make use of a structural result
of Hong et al.~\cite{hong2009hamiltonian}. For a connected subgraph \(H\) of
graph \(G\) the \emph{contraction} \(G/H\) is the graph obtained from \(G\) by
replacing all of \(V(H)\) with a single vertex \(v_{H}\) and adding edges
between \(v_{H}\) and \(V(G) \setminus V(H)\) such that the number of edges in
\(G/H\) between \(v_{H}\) and any vertex \(v \in V(G) \setminus V(H)\) is equal
to the number of edges in \(G\) with one end point at \(v\) and the other in
\(V(H)\). Note that the graph \(G/H\) is, in general, a multigraph with
multiedges incident on \(v_{H}\). Let \(V_{2}\) be the set of all vertices of
\(G\) of degree two, and let \(\hat{V} = V(G) \setminus V_{2}\). A \emph{lane}
of \(G\) is either (i) a path whose end-vertices are in \(\hat{V}\) and internal
vertices (if any) are in \(V_{2}\), or (ii) a cycle which contains exactly one
vertex from \(\hat{V}\). The \emph{length} of a lane is the number of edges in
the lane. An \emph{end-lane} is a lane which has a degree-one vertex of \(G\) as
an end-vertex.

For \(i \in \{2, 3\}\) let \(U_{i}\) be the union of lanes of length \emph{less
  than} \(i\). Let \(C_{1}^{i}, C_{2}^{i}, \dotsc, C_{p_{i}}^{i}\) be the
connected components of \(G[\hat{V}] \cup U_{i}\). Then each \(C_{j}^{i}\)
consists of components of \(G[\hat{V}]\) connected by lanes of length less than
\(i\). Let \(H^{(i)}\) be the graph obtained from \(G\) by contracting each of
the connected subgraphs \(C_{1}^{i}, C_{2}^{i}, \dotsc, C_{p_{i}}^{i}\) to a
distinct vertex. Let \(D_{j}^{i}\) denote the vertex of \(H^{(i)}\) obtained by
contracting subgraph \(C_{j}^{i}\) of \(G\). Let \(\tilde{H}^{(i)}\) be the
graph obtained from \(H^{(i)}\) by these steps:
\begin{enumerate}
\item Delete all lanes beginning and ending at the same vertex \(D_{j}^{i}\).
\item If there are two vertices \(D_{j}^{i}, D_{k}^{i}\) in \(H^{(i)}\) which
  are connected by \(\ell_{1}\) lanes of length at least \(i + 2\) and
  \(\ell_{2}\) lanes of length \(i\) or \(i + 1\) such that
  \(\ell_{1} + \ell_{2} \geq 3\) holds, then delete an arbitrary subset of these
  lanes such that there remain \(\ell_{3}\) lanes with length at least \(i + 2\)
  and \(\ell_{4}\) lanes of length \(i\) or \(i + 1\), where
  \[
    (\ell_{3}, \ell_{4}) =
    \begin{cases}
      (2, 0) & \text{if } \ell_{1} \text{ is even and } \ell_{2} = 0;\\
      (1, 0) & \text{if } \ell_{1} \text{ is odd and } \ell_{2} = 0;\\
      (1, 1) & \text{if } \ell_{2} = 1;\\
      (0, 2) & \text{if } \ell_{2} \geq 2.
    \end{cases}
  \]
\item Delete all end-lanes of length \(i\), and replace each lane of length
  \(i\) or \(i+1\) by a single edge.
\end{enumerate}
\begin{theorem}\textup{\cite[See Theorem~3]{hong2009hamiltonian}}\label{fac:hamIndex_from_SES}
  Let \(G\) be a connected graph with \(h(G) \geq 2\) and with at least one
  vertex of degree at least three, and let \(\tilde{H}^{(2)}, \tilde{H}^{(3)}\)
  be graphs constructed from \(G\) as described above. Then
  \begin{itemize}
  \item \(h(G) = 2\) if and only if \(\tilde{H}^{(2)}\) has a spanning
    Eulerian subgraph; and
  \item \(h(G) = 3\) if and only if \(h(G) \neq 2\) and \(\tilde{H}^{(3)}\) has
    a spanning Eulerian subgraph.
  \end{itemize}
\end{theorem}

For checking if \(h(G) = i\) holds for \(i \in \{4, 5, \dots\}\) we appeal to a
reduction due to Xiong and Liu~\cite{xiongLiu2002hamiltonian}. Let
\(\LL = \{L_{1}, L_{2}, \dotsc, L_{t}\}\) be a set of lanes (called
\emph{branches} in~\cite{xiongLiu2002hamiltonian}) in \(G\), each of length at
least \(2\). A \emph{contraction} of \(G\) by \LL, denoted \(G//\LL\), is a
graph obtained from \(G\) by contracting one edge of each lane in \LL. Note that
\(G//\LL\) is not, in general, unique. 
\begin{theorem}\textup{\cite[Theorem~20]{xiongLiu2002hamiltonian}}\label{fac:hamIndex_reduction}
  Let \(G\) be a connected graph with \(h(G) \geq 4\) and let \LL be the set of
  all lanes of length at least \(2\) in \(G\). Then \(h(G) = h(G//\LL) + 1.\)
\end{theorem}
We can now prove
\begingroup
\def\thetheorem{\ref{thm:hamIndex_is_FPT}}
\begin{theorem}
  There is an algorithm which solves an instance \((G,\mathcal{T},tw,r)\) of \HI
  in \(\OhStar{(1 + 2^{(\omega + 3)})^{tw}}\) time.
\end{theorem}
\addtocounter{theorem}{-1}
\endgroup
\begin{proof}
  We first apply \autoref{fac:graph_hamiltonicity_tw_fpt} to check if \(G\) is
  Hamiltonian. If \(G\) is Hamiltonian then we return \yes. If \(G\) is not
  Hamiltonian and \(r = 0\) holds then we return \no. Otherwise we apply
  \autoref{fac:EHP_EHC_DES_FPT} and \autoref{fac:line_graph_hamiltonicity_DES} to
  check if \(L(G)\) is Hamiltonian. If \(L(G)\) is Hamiltonian then we return
  \yes. If \(L(G)\) is not Hamiltonian and \(r = 1\) holds then we return \no.

  At this point we know---since \(G\) is connected, is not a path, and is not
  Hamiltonian---that \(G\) has at least one vertex of degree at least three, and
  that \(h(G) \geq 2\) holds. We construct the graph \(\tilde{H}^{(2)}\) of
  \autoref{fac:hamIndex_from_SES} and use \autoref{cor:SES_is_FPT} to check if
  \(\tilde{H}^{(2)}\) has a spanning Eulerian subgraph. If it does then we
  return \yes. If it does not and \(r = 2\) holds then we return \no. Otherwise
  we construct the graph \(\tilde{H}^{(3)}\) of \autoref{fac:hamIndex_from_SES}
  and use \autoref{cor:SES_is_FPT} to check if \(\tilde{H}^{(3)}\) has a
  spanning Eulerian subgraph. If it does then we return \yes. If it does not and
  \(r = 3\) holds then we return \no.

  At this point we know that \(h(G) \geq 4\) holds. We compute the set \LL of
  all lanes of \(G\) of length at least \(2\), and a contraction
  \(G' = G//\LL\). We construct a tree decomposition \(\TT'\) of \(G'\) from \TT
  as follows: For each edge \(xy\) of \(G\) which is contracted to get \(G'\),
  we introduce a new vertex \(v_{xy}\) to each bag of \TT which contains at
  least one of \(\{x, y\}\). We now delete vertices \(x\) and \(y\) from all
  bags. It is easy to verify that the resulting structure \(\TT'\) is a tree
  decomposition of \(G'\), of width \(tw' \leq tw\). We now recursively invoke
  the algorithm on the instance \((G', \TT', tw', (r - 1))\) and return its
  return value (\yes or \no).

  The correctness of this algorithm follows from
  \autoref{fac:graph_hamiltonicity_tw_fpt}, \autoref{fac:EHP_EHC_DES_FPT},
  \autoref{fac:line_graph_hamiltonicity_DES}, \autoref{fac:hamIndex_from_SES},
  \autoref{cor:SES_is_FPT}, and \autoref{fac:hamIndex_reduction}. As for the
  running time, checking Hamiltonicity takes
  \(\OhStar{(5 + 2^{(\omega + 2)/2})^{tw}}\) time
  (\autoref{fac:graph_hamiltonicity_tw_fpt}). Checking if \(L(G)\) is
  Hamiltonian takes \(\OhStar{(1 + 2^{(\omega + 3)})^{tw}}\) time
  (\autoref{fac:EHP_EHC_DES_FPT}, \autoref{fac:line_graph_hamiltonicity_DES}). The
  graphs \(\tilde{H}^{(2)}\) and \(\tilde{H}^{(3)}\) of
  \autoref{fac:hamIndex_from_SES} can each be constructed in polynomial time,
  and checking if each has a spanning Eulerian subgraph takes
  \(\OhStar{(1 + 2^{(\omega + 3)})^{tw}}\) time (\autoref{cor:SES_is_FPT}). The
  graph \(G'\) and its tree decomposition \(\TT'\) of width \(tw'\) can be
  constructed in polynomial time. Given that
  \(5 + 2^{(\omega + 2)/2} < 1 + 2^{(\omega + 3)}\) and \(tw' \leq tw\) hold, we
  get that the running time of the algorithm satisfies the recurrence
  \(T(r) = \OhStar{(1 + 2^{(\omega + 3)})^{tw}} + T(r - 1)\). Since we recurse
  only if \(r < |V(G)| - 3\) holds we get that the recurrence resolves to
  \(T(r) = \OhStar{(1 + 2^{(\omega + 3)})^{tw}}\).
\end{proof}


%% file: conclusion.tex
\section{Conclusion}\label{sec:conclusion}
The Hamiltonian Index \(h(G)\) of a graph \(G\) is a generalization of the
notion of Hamiltonicity. It was introduced by Chartrand in 1968, and has
received a lot of attention from graph theorists over the years. It is known to
be \NPH to check if \(h(G) = t\) holds for any fixed integer \(t \geq 0\), even
for subcubic graphs \(G\). We initiate the parameterized complexity analysis of
the problem of finding \(h(G)\) with the treewidth \(tw(G)\) of \(G\) as the
parameter. We show that this problem is \FPT and can be solved in
\(\OhStar{(1 + 2^{(\omega + 3)})^{tw(G)}}\) time. This running time matches that
of the current fastest algorithm, due to Misra et
al.~\cite{MisraPanolanSaurabh2019}, for checking if \(h(G) = 1\) holds. We also
derive an algorithm of our own, with the same running time, for checking if
\(h(G) = 1\) holds. A key ingredient of our solution for finding \(h(G)\) is an
algorithm which solves the \ESS problem in
\(\OhStar{(1 + 2^{(\omega + 3)})^{tw(G)}}\) time. This is---to the best of our
knowledge---the first \FPT algorithm for this problem, and it subsumes known
algorithms for the special case of \SES in series-parallel graphs and planar
graphs. We note in passing that it is not clear that the algorithm of Misra et
al. for solving \textsc{LELP} can be adapted to check for larger values of
\(h(G)\). We believe that our \FPT result on \ESS could turn out to be useful
for solving other problems as well.

Two different approaches to checking if \(h(G) = 1\) holds---Misra et al.'s
approach via \textsc{LELP} and our solution using \DES---both run in
\(\OhStar{(1 + 2^{(\omega + 3)})^{tw(G)}}\) time. Does this suggest the
existence of a matching lower bound, or can this be improved? More generally,
can \(h(G)\) be found in the same \FPT running time as it takes to check if
\(G\) is Hamiltonian (currently: \(\OhStar{(5 + 2^{(\omega + 2)/2})^{tw(G)}}\)
due to Bodlaender et al.)? Since \(tw(G) \leq (|V(G)| - 1)\) our algorithm
implies an \(\OhStar{(1 + 2^{(\omega + 3)})^{|V(G)|}}\)-time exact exponential
algorithm for finding \(h(G)\). We ask if this can be improved, as a first step,
to the classical \(\OhStar{2^{|V(G)|}}\) bound for Hamiltonicity.


%% file: des.tex
\section{An \FPT Algorithm for \DES}\label{sec:DES}
In this section we derive an alternate algorithm for

\defparproblem{\DES (\DESs)}%
{An undirected graph \(G=(V,E)\) and a tree decomposition
  \(\mathcal{T}=(T,\{X_{t}\}_{t \in V(T)})\) of \(G\), of width \(tw\).}%
{\(tw\)}%
{Does there exist an Eulerian subgraph \(G'\) of \(G\) such that \(V(G')\)
  contains a vertex cover of \(G\)? }

Following established terminology we call such a subgraph \(G'\) a
\emph{dominating Eulerian subgraph} of \(G\). \textbf{Note that the word
  ``dominating'' here denotes the existence of a \emph{vertex cover}, and
  \emph{not} of a \emph{dominating set}.}
\begin{theorem}\label{thm:DES_is_FPT}
  There is an algorithm which solves an instance \((G,\mathcal{T},tw)\) of \DES
  in \(\OhStar{(1 + 2^{(\omega + 3)})^{tw}}\) time.
\end{theorem}

We describe an algorithm which takes an instance \((G, \mathcal{T}, tw)\) of
\DES as input and tells in \(\OhStar{(1 + 2^{(\omega + 3)})^{tw}}\) time whether
graph \(G\) has a subgraph which is Eulerian, and whose vertex set is a vertex
cover of \(G\). The algorithm is a DP over a tree decomposition, very similar to
the one in \autoref{sec:ESS}. As before we assume that \TT is a nice tree
decomposition.
Our proofs are simplified if we assume that we know of a vertex \(\vstar\) which
is definitely part of the unknown dominating Eulerian subgraph which we are
trying to find. Note that for any edge \(xy\) of \(G\) at least one of the two
vertices \(\{x, y\}\) must be part of \emph{any} dominating Eulerian subgraph of
\(G\). So one of the two choices \(\vstar = x\) and \(\vstar = y\) will satisfy
our requirement on \(\vstar\). Hence we assume, without loss of generality, that
we have picked a correct choice for \(\vstar\). We add the vertex \(\vstar\) to
every bag of \(\mathcal{T}\); \emph{from now on we use \(\mathcal{T}\) to refer
  to the resulting ``nearly-nice'' tree decomposition in which the bags at all
  the leaves and the root are equal to \(\{\vstar\}\)}. 

If the graph \(G\) has a dominating Eulerian subgraph $G'=(V',E')$ then it
interacts with the structures defined by node $t$ in the following way: The part
of $G'$ contained in $G_{t}$ is a collection \(\CC = \{C_{1},\dotsc,C_{\ell}\}\)
of pairwise vertex-disjoint connected subgraphs of $G_{t}$ where \emph{each
  element} \(C_{i}\) of \(\CC\) has a non-empty intersection with $X_{t}$.
Further, the union of the vertex sets of the components in \CC forms a vertex
cover of graph \(G_{t}\).

\begin{definition}[Valid partitions, witness for validity for \DES]\label{def:valid_partitions_witnesses_des}
  For a bag \(X_{t}\) and subsets \(X \subseteq X_{t}\), \(O \subseteq X\), we
  say that a partition \(P=\{X^{1},X^{2},\ldots X^{p}\}\) of \(X\) is
  \emph{valid for the combination} \((t,X,O)\) if there exists a subgraph
  \(G'_{t} = (V'_{t},E'_{t})\) of \(G_{t}\) such that
  \begin{enumerate}
  \item \(X_{t} \cap V(G'_{t})=X\).
  \item \(G'_{t}\) has exactly \(p\) connected components
    \(C_{1},C_{2},\ldots, C_{p}\) and for each \(i \in \{1,2,\ldots,p\}\),
    \(X^{i} \subseteq V(C_{i})\). That is, the vertex set of each connected
    component of \(G'_{t}\) has a non-empty intersection with set \(X\), and
    \(P\) is the partition of \(X\) defined by the subgraph
    \(G'_{t}\). 
  \item \(v^{\star} \in V(G'_{t})\) holds, and \(V(G'_{t})\) is a vertex cover
    of graph \(G_{t}\).
  \item The set of odd-degree vertices in $G'_{t}$ is exactly the set $O$.
  \end{enumerate}
  Such a subgraph \(G'_{t}\) of $G_{t}$ is a \emph{witness} for partition \(P\)
  being valid for the combination \((t,X,O)\) or, in short: \(G'_{t}\) \emph{is
    a witness for \(((t,X,O), P)\)}.
\end{definition}

\begin{definition}[Completion for \DES]\label{def:completion_des}
  For a bag $X_{t}$ and subsets $X \subseteq X_{t}$, $O \subseteq X$ let \(P\)
  be a partition of \(X\) which is valid for the combination \((t,X,O)\). Let
  \(H\) be a residual subgraph with respect to \(t\) such that
  \(V(H) \cap X_{t} = X\). We say that \(((t,X,O), P)\) \emph{completes} \(H\)
  if there exists a subgraph \(G'_{t}\) of \(G_{t}\) which is a witness for
  \(((t,X,O), P)\), such that the graph \(G'_{t} \cup H\) is a dominating
  Eulerian subgraph of \(G\). We say that \(G'_{t}\) is a \emph{certificate} for
  \(((t,X,O), P)\) completing \(H\).
\end{definition}

\begin{observation}\label{obs:residual_completion_vertexcover_des}
  If \(((t,X,O), P)\) completes \(H\) then every edge in the set
  \(E(G) \setminus E(G_{t})\) has at least one end-point in the vertex set
  \(V(H)\).
\end{observation}

\begin{lemma}\label{lem:completion_odd_subset_des}
  Let \((G, \TT, tw)\) be an instance of \DES. Let \(t\) be an arbitrary node
  of \TT, let $X \subseteq X_{t}$, $O \subseteq X$, let \(P\) be a partition of
  \(X\) which is valid for the combination \((t,X,O)\), and let \(H\) be a
  residual subgraph with respect to \(t\) with \(V(H) \cap X_{t} = X\). If
  \(((t,X,O), P)\) completes \(H\) then the set of odd-degree vertices in \(H\)
  is exactly the set \(O\).
\end{lemma}
\begin{proof}
  Let \(H_{odd} \subseteq V(H)\) be the set of odd-degree vertices in \(H\).
  Since \(((t,X,O), P)\) completes \(H\) we know that there is a subgraph
  \(G'_{t} = (V'_{t}, E'_{t})\) of \(G_{t}\) which is a witness for
  \(((t,X,O), P)\), such that the graph \(G^{\star} = G'_{t} \cup H\) is a
  dominating Eulerian subgraph of \(G\). Since \(G'_{t}\) is a witness for
  \(((t,X,O), P)\) we get that the set of odd-degree vertices in \(G'_{t}\) is
  exactly the set \(O\). Since \(H\) is a residual subgraph with respect to
  \(t\) we have that \(E'_{t} \cap E(H) = \emptyset\). Thus the degree of any
  vertex \(v\) in the graph \(G^{\star}\) is the sum of its degrees in the two
  subgraphs \(H\) and \(G'_{t}\):
  \(deg_{G^{\star}}(v) = deg_{H}(v) + deg_{G'_{t}}(v)\). And since \(G^{\star}\)
  is Eulerian we have that \(deg_{G^{\star}}(v)\) is even for every vertex
  \(v \in V(G^{\star})\).

  Now let \(v \in H_{odd} \subseteq V(H)\) be a vertex of odd degree in \(H\).
  Then \(v \in V(G^{\star})\) and we get that
  \(deg_{G'_{t}}(v) = deg_{G^{\star}}(v) - deg_{H}(v)\) is odd. Thus
  \(v \in O\), and so \(H_{odd} \subseteq O\). Conversely, let
  \(x \in O \subseteq V'_{t}\) be a vertex of odd degree in \(G'_{t}\). Then
  \(x \in V(G^{\star})\) and we get that
  \(deg_{H}(x) = deg_{G^{\star}}(x) - deg_{G'_{t}}(x)\) is odd. Thus
  \(x \in H_{odd}\), and so \(O \subseteq H_{odd}\). Thus the set of odd-degree
  vertices in \(H\) is exactly the set \(O\).
\end{proof}

The next lemma tells us that it is safe to apply the representative set
computation to collections of valid partitions.
\begin{lemma}\label{lem:repsets_preserve_completion_des}
  Let \((G, \TT, tw)\) be an instance of \DES, and let \(t\) be an arbitrary
  node of \TT. Let $X \subseteq X_{t}$, $O \subseteq X$, and let \calA be a
  collection of partitions of \(X\), each of which is valid for the combination
  \((t,X,O)\). Let \calB be a representative subset of \calA, and let \(H\) be
  an arbitrary residual subgraph of \(G\) with respect to \(t\) such that
  \(V(H) \cap X_{t} = X\) holds. If there is a partition \(P \in \calA\) such
  that \(((t,X,O), P)\) completes \(H\) then there is a partition
  \(Q \in \calB\) such that \(((t,X,O), Q)\) completes \(H\).
\end{lemma}
\begin{proof}
  Suppose there is a partition \(P \in \calA\) such that \(((t,X,O), P)\)
  completes the residual subgraph \(H\). Then there exists a subgraph
  \(G'_{t} = (V'_{t},E'_{t})\) of $G_{t}$--- \(G'_{t}\) being a witness for
  \(((t,X,O), P)\)---such that (i) \(X_{t} \cap V(G'_{t})=X\), (ii) \(P\) is the
  partition of \(X\) defined by \(G'_{t}\), (iii) \(V(G'_{t})\) is a vertex
  cover of of $G_{t}$, (iv) the set of odd-degree vertices in \(G'_{t}\) is
  exactly the set \(O\), and (v) the graph \(G'_{t} \cup H\) is a dominating
  Eulerian subgraph of \(G\). Observe that every edge in
  \(E(G) \setminus E(G_{t})\) has at least one end-point in the set \(V(H)\).
  Let \(R\) be the partition of the set \(X\) defined by the residual subgraph
  \(H\). Since the union of \(G'_{t}\) and \(H\) is connected we
  get---\autoref{lem:partition_join_connectivity}---that \(P \sqcup R = \{X\}\)
  holds. Since \calB is a representative subset of \calA we get that there
  exists a partition \(Q \in \calB\) such that \(Q \sqcup R = \{X\}\) holds.
  Since \(\calB \subseteq \calA\) we have that the partition \(Q\) of \(X\) is
  valid for the combination \((t,X,O)\). So there exists a subgraph
  \(G^{\star}_{t} = (V^{\star}_{t},E^{\star}_{t})\) of
  $G_{t}$---\(G^{\star}_{t}\) being a witness for \(((t,X,O), Q)\)---such that
  (i) \(X_{t} \cap V(G^{\star}_{t})=X\), (ii) \(Q\) is the partition of \(X\)
  defined by \(G^{\star}_{t}\), (iii) \(V(G^{\star}_{t})\) is a vertex cover of
  \(G_{t}\) with \(v^{\star} \in V(G^{\star}_{t})\) , and (iv) the set of
  odd-degree vertices in \(G^{\star}_{t}\) is exactly the set \(O\). Now:
  \begin{enumerate}
  \item The vertex set of the graph \(G^{\star}_{t} \cup H\) is a vertex cover
    of graph \(G\). This follows from
    \autoref{obs:residual_completion_vertexcover_des}, since
    \(V(G^{\star}_{t})\) is a vertex cover of \(G_{t}\).
  \item The graph \(G^{\star}_{t} \cup H\) has all degrees even, because (i) the
    edge sets \(E(G^{\star}_{t})\) and \(E(H)\) are disjoint, and (ii) the sets
    of odd-degree vertices in the two graphs \(G^{\star}_{t}\) and \(H\) are
    identical---namely, the set \(O\).
  \item \(G^{\star}_{t} \cup H\) is connected---by
    \autoref{lem:partition_join_connectivity}---because \(Q \sqcup R = \{X\}\)
    holds.
  \end{enumerate}
  Thus the subgraph \(G^{\star}_{t}\) of \(G_{t}\) is a witness for
  \(((t,X,O), Q)\) such that the graph \(G^{\star}_{t} \cup H\) is a dominating
  Eulerian subgraph of \(G\). Hence \(((t,X,O), Q)\) completes \(H\).
\end{proof}

\begin{lemma}\label{lem:validity_condition_at_root_des}
  Let \((G, \TT, tw)\) be an instance of \DES, let \(r\) be the root node of
  \TT, and let \(v^{\star}\) be the vertex which is present in every bag of \TT.
  Then \((G, \TT, tw)\) is a \yes instance of \DES if and only if the partition
  \(P = \{\{v^{\star}\}\}\) is valid for the combination
  \((r,X = \{v^{\star}\}, O = \emptyset)\).
\end{lemma}
\begin{proof}
  Let \((G, \TT, tw)\) be a \yes instance of \DES and let \(G'\) be a dominating
  Eulerian subgraph of \(G\). By assumption vertex \(v^{\star}\) is in
  \(V(G')\). Since \(r\) is the root node of \TT we have that
  \(X_{r} = \{v^{\star}\}\), \(V_{r} = V(G)\) and \(G_{r} = G\). We set
  \(G'_{r} = G'\). Then (i) $X_{r} \cap V(G'_{r}) = \{v^{\star}\} = X$, (ii)
  $G'_{r} = G'$ has exactly one connected component \(C_{1} = V(G')\) and the
  partition \(P = \{\{v^{\star}\}\}\) of \(X = \{v^{\star}\}\) is the partition
  of \(X\) defined by $G'_{r}$, (iii) \(V(G'_{r}) = V(G')\) is a vertex cover of
  graph \(G\) with \(v^{\star} \in V(G'_{r})\), and (iv) the set of odd-degree
  vertices in $G'_{r}$ is exactly the empty set. Thus the partition
  \(P = \{\{v^{\star}\}\}\) is valid for the combination
  \((r,X = \{v^{\star}\}, O = \emptyset)\). This completes the forward
  direction.

  For the reverse direction, suppose the partition \(P = \{\{v^{\star}\}\}\) is
  valid for the combination \((r,X = \{v^{\star}\}, O = \emptyset)\). Then by
  definition there exists a subgraph \(G'_{r} = (V'_{r},E'_{r})\) of $G_{r} = G$
  such that (i) $X_{r} \cap V(G'_{r}) = X = \{v^{\star}\}$, (ii) \(G'_{r}\) has
  exactly one connected component \(C_{1} = V(G'_{r})\), (iii) \(V(G'_{r})\) is
  a vertex cover of graph \(G\) with \(v^{\star} \in V(G'_{r})\) , and (iv) the
  set of odd-degree vertices in \(G'_{r}\) is exactly the empty set \(O\). Thus
  \(G'_{r}\) is a connected subgraph of \(G\) whose vertex set is a vertex cover
  of \(G\), and whose degrees are all even. But \(G'_{r}\) is then a dominating
  Eulerian subgraph of \(G\), and so \((G, \TT, tw)\) is a \yes instance of
  \DES.
\end{proof}

\begin{lemma}\label{lem:completion_at_root_des}
  Let \((G, \TT, tw)\) be an instance of \DES, let \(r\) be the \emph{root} node
  of \TT, and let \(v^{\star}\) be the vertex which is present in every bag of
  \TT. Let \(H = (\{v^{\star}\}, \emptyset)\), \(X = \{v^{\star}\}\),
  \(O = \emptyset\), and \(P = \{\{v^{\star}\}\}\). Then \((G, \TT, tw)\) is a
  \yes instance if and only if \(((r,X,O), P)\) completes \(H\).
\end{lemma}
\begin{proof}
  Note that \(G_{r} = G\). It is easy to verify by inspection that \(H\) is a
  residual subgraph with respect to \(r\).

  Let \((G, \TT, tw)\) be a \yes instance of \DES and let \(G'\) be a dominating
  Eulerian subgraph of \(G\). By assumption vertex \(v^{\star}\) is in
  \(V(G')\). From \autoref{lem:validity_condition_at_root_des} we get that
  partition \(P\) is valid for the combination \((r,X,O)\), and from the proof
  of \autoref{lem:validity_condition_at_root_des} we get that the dominating
  Eulerian subgraph \(G'\) is itself a witness for \(((r,X,O), P)\). Now
  \(((V(G') \cup V(H)), (E(G') \cup E(H))) = (V(G'), E(G')) = G'\), and so
  \(G' \cup H\) is a dominating Eulerian subgraph of \(G\). Thus
  \(((r,X,O), P)\) completes \(H\).

  The reverse direction is trivial: if \(((r,X,O), P)\) completes \(H\) then by
  definition there exists a dominating Eulerian subgraph of \(G\), and so
  \((G, \TT, tw)\) is a \yes instance.
\end{proof}

As in \autoref{sec:ESS} we use the completion-based alternate characterization
of \yes instances---\autoref{lem:completion_at_root_des}---and representative
subset computations---\autoref{lem:repsets_preserve_completion_des}---to speed
up our DP. We now describe the steps of the DP algorithm for each type of node
in \TT. We use \(VP[t, X, O]\) to denote the set of \textbf{v}alid
\textbf{p}artitions for the combination \((t, X, O)\) which we store in the DP
table for node \(t\).
\begin{description}
\item[Leaf node \(t\):] In this case \(X_{t} = \{v^{\star}\}\). Set
  \(VP[t, \{v^{\star}\}, \{v^{\star}\}] = \emptyset\),
  \(VP[t, \{v^{\star}\}, \emptyset] = \{\{\{v^{\star}\}\}\}\), and
  \(VP[t, \emptyset, \emptyset] = \{\emptyset\}\).
\item[Introduce vertex node \(t\):] Let \(t'\) be the child node of \(t\), and
  let \(v\) be the vertex introduced at \(t\). Then \(v \notin X_{t'}\) and
  \(X_{t} = X_{t'} \cup \{v\}\). For each \(X \subseteq X_{t}\) and
  \(O \subseteq X\),
  \begin{enumerate}
  \item if \(v \in O\) then set \(VP[t, X, O] = \emptyset\)
  \item if \(v \in (X \setminus O)\) then for each partition \(P'\) in
    \(VP[t', X \setminus \{v\}, O]\), add the partition
    \(P = P' \cup \{\{v\}\}\) to the set \(VP[t, X, O]\)
  \item if \(v \notin X\) then set \(VP[t, X, O] = VP[t', X, O]\)
  \item Set \(\calA = VP[t, X, O]\). Compute a representative subset
    \(\calB \subseteq \calA\) and set \(VP[t, X, O] = \calB\).
  \end{enumerate}
\item[Introduce edge node \(t\):] Let \(t'\) be the child node of \(t\), and let
  \(uv\) be the edge introduced at \(t\). Then \(X_{t} = X_{t'}\) and
  \(uv \in (E(G_{t}) \setminus E(G_{t'}))\). For each \(X \subseteq X_{t}\) and
  \(O \subseteq X\), 
  \begin{enumerate}
  \item If \(\{u, v\} \cap X = \emptyset\) then set \(VP[t, X, O] = \emptyset\);
    else set \(VP[t, X, O] = VP[t', X, O]\).
  \item If \(\{u,v\} \subseteq X\) then:
    \begin{enumerate}
    \item Construct a set of \emph{candidate partitions} \PP as follows.
      Initialize \(\PP = \emptyset\).
      \begin{itemize}
      \item if \(\{u,v\} \subseteq O\) then add all the partitions in
        \(VP[t', X, O \setminus \{u,v\}]\) to \PP.
      \item if \(\{u,v\} \cap O = \{u\}\) then add all the partitions in
        \(VP[t', X, (O \setminus \{u\}) \cup \{v\}]\) to \PP.
      \item if \(\{u,v\} \cap O = \{v\}\) then add all the partitions in
        \(VP[t', X, (O \setminus \{v\}) \cup \{u\}]\) to \PP.
      \item if \(\{u,v\} \cap O = \emptyset\) then add all the partitions in
        \(VP[t', X, O \cup \{u, v\}]\) to \PP.
      \end{itemize}
    \item For each candidate partition \(P' \in \PP\), if vertices \(u, v\) are
      in different blocks of \(P'\)---say
      \(u \in P'_{u}, v \in P'_{v}\;;\;P'_{u} \neq P'_{v}\)---then merge these
      two blocks of \(P'\) to obtain \(P\). That is, set
      \(P = (P' \setminus \{P'_{u}, P'_{v}\}) \cup (P'_{u} \cup P'_{v})\). Now
      set \(\PP = (\PP \setminus \{P'\}) \cup P\).
    \item Add all of \(\PP\) to the list \(VP[t, X, O]\).
    \end{enumerate}
  \item Set \(\calA = VP[t, X, O]\). Compute a representative subset
    \(\calB \subseteq \calA\) and set \(VP[t, X, O] = \calB\).
  \end{enumerate}
\item[Forget node \(t\):] Let \(t'\) be the child node of \(t\), and let \(v\)
  be the vertex forgotten at \(t\). Then \(v \in X_{t'}\) and
  \(X_{t} = X_{t'} \setminus \{v\}\). Recall that \(P(v)\) is the block of
  partition \(P\) which contains element \(v\), and that \(P - v\) is the
  partition obtained by eliding \(v\) from \(P\). For each \(X \subseteq X_{t}\)
  and \(O \subseteq X\),
  \begin{enumerate}
  \item Set \(VP[t, X, O] = \{P' - v \;;\; P' \in VP[t', X \cup \{v\},
    O],\,|P'(v)| > 1\} \cup VP[t', X, O]\).
  \item Set \(\calA = VP[t, X, O]\). Compute a representative subset
    \(\calB \subseteq \calA\) and set \(VP[t, X, O] = \calB\).
  \end{enumerate}
\item[Join node \(t\):] Let \(t_{1}, t_{2}\) be the children of \(t\). Then
  \(X_{t} = X_{t_{1}} = X_{t_{2}}\). For each
  \(X \subseteq X_{t}, O \subseteq X\):
  \begin{enumerate}
  \item Set \(VP[t, X, O] = \emptyset\)
  \item For each \(O_{1} \subseteq O\) and
    \(\hat{O} \subseteq (X \setminus O)\): 
    \begin{enumerate}
    \item Let \(O_{2} = O \setminus O_{1}\).
    \item For each pair of partitions
      \(P_{1} \in VP[t_{1}, X, O_{1} \cup \hat{O}], P_{2} \in VP[t_{2}, X, O_{2}
      \cup \hat{O}]\), add their join \(P_{1} \sqcup P_{2}\) to the set
      \(VP[t, X, O]\).
    \end{enumerate}
  \item Set \(\calA = VP[t, X, O]\). Compute a representative subset
    \(\calB \subseteq \calA\) and set \(VP[t, X, O] = \calB\).
  \end{enumerate}
\end{description}

As before, we prove by induction on the structure of \TT that every node in \TT
preserves the Correctness Criteria (see \autopageref{correctness_criteria}).
The processing at each of the non-leaf nodes computes a representative subset as
a final step. This step does not negate the correctness criteria.
  \begin{observation}\label{obs:repset_computation_preserves_correctness_des}
    Let \(t\) be a node of \TT, let \(X \subseteq X_{t}, O \subseteq X\), and
    let \calA be a set of partitions which satisfies the correctness criteria
    for the combination \((t, X, O)\). Let \calB be a representative subset of
    \calA. Then \calB satisfies the correctness criteria for the combination
    \((t, X, O)\).
  \end{observation}
  \begin{proof}
    Since \(\calB \subseteq \calA\) holds we get that \calB satisfies the
    soundness criterion. From \autoref{lem:repsets_preserve_completion_des} we
    get that \calB satisfies the completeness criterion as well.
  \end{proof}

\begin{lemma}\label{lem:leaf_node_ok_des}
  Let \(t\) be a leaf node of the tree decomposition \TT and let
  \(X \subseteq X_{t}, O \subseteq X\) be arbitrary subsets of \(X_{t}, X\)
  respectively. The collection \calA of partitions computed by the DP for the
  combination \((t, X, O)\) satisfies the correctness criteria.
\end{lemma}
\begin{proof}
  Here \(X_{t} = \{v^{\star}\}\). Note that the graph \(G_{t}\) consists of (i)
  the one vertex \(v^{\star}\), and (ii) no edges. We verify the conditions for
  all the three possible cases:
  \begin{itemize}
  \item \(X = \{v^{\star}\}, O = \{v^{\star}\}\). The algorithm sets
    \(\calA = \emptyset\). The soundness criterion holds vacuously.

    Observe that there is no subgraph \(G_{t'}\) of \(G_{t}\) in which vertex
    \(v^{\star}\) has an odd degree. This means that there can exist no subgraph
    \(G_{t'}\) of \(G_{t}\) for which the fourth condition in the definition of
    a valid partition---~\autoref{def:valid_partitions_witnesses_des}---holds.
    Thus there is no partition which is valid for the combination \((t, X, O)\).
    Hence the completeness criterion holds vacuously as well.
  \item \(X = \{v^{\star}\}, O = \emptyset\). The algorithm sets
    \(\calA = \{\{\{v^{\star}\}\}\}\). It is easy to verify by inspection that
    the subgraph \(G_{t'} = G_{t}\) of \(G_{t}\) is a witness for the partition
    \(\{\{v^{\star}\}\}\) being valid for the combination \((t, X, O)\). Hence
    the soundness criterion holds.

    Since \(X\) is the set \(\{v^{\star}\}\), the \emph{only} valid partition
    for the combination \((t, X, O)\) is \(\{\{v^{\star}\}\}\). Hence the
    completeness criterion holds trivially.
  \item \(X = \emptyset, O = \emptyset\). The algorithm sets
    \(\calA = \emptyset\). The soundness criterion holds vacuously.

    Since \(X = \emptyset\) holds there can exist no subgraph \(G_{t'}\) of
    \(G_{t}\) for which both the conditions (1) and (3) of the definition of a
    valid partition---~\autoref{def:valid_partitions_witnesses_des}---hold
    simultaneously. Thus there is no partition which is valid for the
    combination \((t, X, O)\). Hence the completeness criterion holds vacuously
    as well. \qedhere
  \end{itemize}
\end{proof}

\begin{lemma}\label{lem:introduce_vertex_node_ok_des}
  Let \(t\) be an introduce vertex node of the tree decomposition \TT and let
  \(X \subseteq X_{t}, O \subseteq X\) be arbitrary subsets of \(X_{t}, X\)
  respectively. The collection \calA of partitions computed by the DP for the
  combination \((t, X, O)\) satisfies the correctness criteria.
\end{lemma}
\begin{proof}
  Let \(t'\) be the child node of \(t\), and let \(v\) be the vertex introduced
  at \(t\). Then \(v \notin X_{t'}\) and \(X_{t} = X_{t'} \cup \{v\}\) hold.
  Note that no edges incident with \(v\) have been introduced so far; so we have
  that \(deg_{G_{t}}(v) = 0\) holds. We analyze each choice made by the
  algorithm:
  \begin{enumerate}
  \item If \(v \in O\) holds then the algorithm sets \(\calA = \emptyset\). The
    soundness criterion holds vacuously.

    Since \(deg_{G_{t}}(v) = 0\) holds, there can exist no subgraph \(G_{t'}\)
    of \(G_{t}\) for which the fourth condition of the definition of a valid
    partition---~\autoref{def:valid_partitions_witnesses_des}---holds. Thus
    there is no partition which is valid for the combination \((t, X, O)\).
    Hence the completeness condition holds vacuously as well.

  \item If \(v \in (X \setminus O)\) holds then the algorithm takes each
    partition \(P'\) in \(VP[t', X \setminus \{v\}, O]\) and adds the partition
    \(P = (P' \cup \{\{v\}\})\) to the set \(\calA\). By inductive assumption we
    have that the set \(VP[t', X \setminus \{v\}, O]\) of partitions is sound
    and complete for the combination \((t', X \setminus \{v\}, O)\).

    Let \(P = (P' \cup \{\{v\}\})\) be an arbitrary partition in the set
    \(\calA\), where \(P'\) is a partition from the set
    \(VP[t', X \setminus \{v\}, O]\). Then the partition \(P'\) is valid for the
    combination \((t', X \setminus \{v\}, O)\), and so there exists a subgraph
    \(H\) of the graph \(G_{t'}\) such that \(H\) is a witness for
    \(((t', X \setminus \{v\}, O), P')\). It is easy to verify by inspection
    that the graph \(G_{t}' = ( V(H) \cup \{v\}, E(H))\) is a subgraph of
    \(G_{t}\) which satisfies all the four conditions of
    \autoref{def:valid_partitions_witnesses_des} for being a witness for
    \(((t, X, O), P)\). Thus the soundness condition holds for the set
    \(\calA\).
    
    Now we prove completeness. So let \(H\) be a residual subgraph with respect
    to \(t\) with \(V(H) \cap X_{t} = X\), for which there exists a partition
    \(P=\{X^{1},X^{2},\ldots X^{p}\}\) of \(X\) such that \(((t, X, O), P)\)
    completes \(H\). We need to show that the set \calA computed by the
    algorithm contains some partition \(Q\) of \(X\) such that
    \(((t, X, O), Q)\) completes \(H\). Observe that there exists a subgraph
    \(G'_{t}\) of \(G_{t}\)---a witness for \(((t,X,O), P)\)---such that the
    following hold:
    \begin{enumerate}
    \item \(X_{t} \cap V(G'_{t})=X\).
    \item \(G'_{t}\) has exactly \(p\) connected components
      \(C_{1}, C_{2}, \dotsc, C_{p}\) and for each
      \(i \in \{1, 2, \dotsc, p\}$, $X^{i} \subseteq V(C_{i})\) holds.
    \item \(v^{\star} \in V(G'_{t})\) holds, and \(V(G'_{t})\) is a vertex cover
    of graph \(G_{t}\).
    \item The set of odd-degree vertices in \(G'_{t}\) is exactly the set \(O\).
    \item The graph \(G'_{t} \cup H\) is a dominating Eulerian subgraph of
      \(G\).
    \end{enumerate}
    Since \(deg_{G_{t}}(v) = 0\) holds, we get that \(deg_{G'_{t}}(v) = 0\)
    holds as well. Thus vertex \(v\) forms a connected component by itself in
    graph \(G'_{t}\). Without loss of generality, let this component by
    \(C_{p}\). Then we get that \(X^{p} = V(C_{p}) = \{v\}\), and that
    \(P'=\{X^{1},X^{2},\ldots X^{(p - 1)}\}\) is a partition of the set
    \(X \setminus \{v\}\).

    Since \(v \in X\) and \(V(H) \cap X_{t} = X\) hold, and since the graph
    \(G'_{t} \cup H\) is Eulerian, we get that vertex \(v\) has a positive even
    degree in graph \(H\). Since \(H\) is a residual subgraph with respect to
    \(t\) we have that (i) \(V(H) \cap (V_{t} \setminus X_{t}) = \emptyset\) and
    (ii) \(E(H) \cap E_{t} = \emptyset\) hold. Since
    \(X_{t} = X_{t'} \cup \{v\}\) holds, we get that
    \(V_{t'} = V_{t} \setminus \{v\}\) and hence
    \(V_{t'} \setminus X_{t'} = V_{t} \setminus X_{t}\) holds. Hence
    \(V(H) \cap (V_{t'} \setminus X_{t'}) = \emptyset\) holds. Further, since
    \(E_{t'} \subseteq E_{t}\) holds we get that
    \(E(H) \cap E_{t'} = \emptyset\) holds as well. Thus \(H\) is a residual
    subgraph with respect to node \(t'\) which (i) contains vertex \(v\) and
    (ii) satisfies \(V(H) \cap X_{t'} = (X \setminus \{v\})\).

    Now let \(G'_{t'}\) be the graph obtained from \(G'_{t}\) by deleting vertex
    \(v\). Then \(G'_{t'}\) is a subgraph of \(G_{t'}\), and it is
    straightforward to verify that the following hold:
    \begin{enumerate}
    \item \(X_{t'} \cap V(G'_{t'}) = (X \setminus \{v\})\).
    \item \(G'_{t'}\) has exactly \(p - 1\) connected components
      \(C_{1}, C_{2}, \dotsc, C_{(p - 1)}\) and for each
      \(i \in \{1, 2, \dotsc, p - 1\}$, $X^{i} \subseteq V(C_{i})\) holds.
    \item \(v^{\star} \in V(G'_{t'})\) holds, and \(V(G'_{t'})\) is a vertex
      cover of graph \(G_{t'}\).
    \item The set of odd-degree vertices in \(G'_{t'}\) is exactly the set
      \(O\).
    \item The graph \(G'_{t'} \cup H\) is identical to the graph
      \(G'_{t} \cup H\), and hence is a dominating Eulerian subgraph of \(G\).
    \end{enumerate}
    Thus \(H\) is a residual subgraph with respect to \(t'\) with
    \(V(H) \cap X_{t'} = (X \setminus \{v\})\), and
    \(P'=\{X^{1},X^{2},\ldots X^{(p - 1)}\}\) is a partition of
    \(X \setminus \{v\}\) such that \(((t', X \setminus \{v\}, O), P')\)
    completes \(H\). From the inductive assumption we know that the set
    \(VP[t', X \setminus \{v\}, O]\) contains a partition
    \(Q'=\{Y^{1},Y^{2},\ldots Y^{q}\}\) of \(X \setminus \{v\}\) such that
    \(((t', X \setminus \{v\}, O), Q')\) completes \(H\). So there is a subgraph
    \(G''_{t'}\) of \(G_{t'}\)---a witness for
    \(((t', X \setminus \{v\}, O), Q')\)---such that the following hold:
    \begin{enumerate}
    \item \(X_{t'} \cap V(G''_{t'})=X \setminus \{v\}\).
    \item \(G''_{t'}\) has exactly \(q\) connected components
      \(D_{1}, D_{2}, \dotsc, D_{q}\) and for each
      \(i \in \{1, 2, \dotsc, q\}$, $Y^{i} \subseteq V(D_{i})\) holds.
    \item \(v^{\star} \in V(G''_{t'})\) holds, and \(V(G''_{t'})\) is a vertex
      cover of graph \(G_{t'}\).
    \item The set of odd-degree vertices in \(G''_{t'}\) is exactly the set
      \(O\).
    \item The graph \(G''_{t'} \cup H\) is a dominating Eulerian subgraph of
      \(G\).
    \end{enumerate}

    Now the algorithm adds the partition
    \(Q = Q' \cup \{\{v\}\} = =\{Y^{1},Y^{2},\ldots Y^{q}, \{v\}\}\) of set
    \(X\) to the set \calA. It is straightforward to verify that the graph
    \(\hat{G}_{t} = (V(G''_{t'}) \cup \{v\}, E(G''_{t'}))\) is a subgraph of
    graph \(G_{t}\) for which the following hold:
    \begin{enumerate}
    \item \(X_{t} \cap V(\hat{G}_{t})=X\).
    \item \(\hat{G}_{t}\) has exactly \(q+1\) connected components
      \(D_{1}, D_{2}, \dotsc, D_{q}, D_{q+1} = (\{v\}, \emptyset)\) and for each
      \(i \in \{1, 2, \dotsc, q+1\}$, $Y^{i} \subseteq V(D_{i})\) holds.
    \item \(v^{\star} \in V(\hat{G}_{t})\) holds, and \(V(\hat{G}_{t})\) is a
      vertex cover of graph \(G_{t}\).
    \item The set of odd-degree vertices in \(\hat{G}_{t}\) is exactly the set
      \(O\).
    \item The graph \(\hat{G}_{t} \cup H\) is a dominating Eulerian subgraph of
      \(G\).
    \end{enumerate}
    Thus \calA contains a partition \(Q\) of \(X\) such that \(((t, X, O), Q)\)
    completes \(H\), as was required to be shown for completeness.
  \item If \(v \notin X\) holds then the algorithm sets
    \(\calA = VP[t', X, O]\). It is straightforward to verify using
    Definitions~\ref{def:residual_subgraph},
    \ref{def:valid_partitions_witnesses_des}, and~\ref{def:completion_des} that:
    \begin{itemize}
    \item a partition \(P\) of set \(X\) is valid for the combination
      \((t, X, O)\) if and only if it is valid for the combination
      \((t', X, O)\);
    \item a subgraph of \(G_{t}\) is a witness for \(((t, X, O), P)\) if and
      only if it is (i) a subgraph of \(G_{t'}\) and (ii) a witness for
      \(((t', X, O), P)\);
    \item a graph \(H\) is a residual subgraph with respect to \(t\) with
      \(V(H) \cap X_{t} = X\) if and only if \(H\) is a residual subgraph with
      respect to \(t'\) with \(V(H) \cap X_{t'} = X\); and,
    \item for any residual subgraph \(H\) with respect to \(t\) with
      \(V(H) \cap X_{t} = X\) and any partition \(P\) of \(X\),
      \(((t, X, O), P)\) completes \(H\) if and only if \(((t', X, O), P)\)
      completes \(H\).
    \end{itemize}
    By the inductive assumption we have that the set \(VP[t', X, O]\) of
    partitions is sound and complete for the combination \((t', X, O)\). It
    follows from the above equivalences that the set \(\calA = VP[t', X, O]\) is
    sound and complete for the combination \((t, X, O)\). \qedhere
  \end{enumerate}
\end{proof}

\begin{lemma}\label{lem:introduce_edge_node_ok_des}
  Let \(t\) be an introduce edge node of the tree decomposition \TT and let
  \(X \subseteq X_{t}, O \subseteq X\) be arbitrary subsets of \(X_{t}, X\)
  respectively. The collection \calA of partitions computed by the DP for the
  combination \((t, X, O)\) satisfies the correctness criteria.
\end{lemma}
\begin{proof}
  Let \(t'\) be the child node of \(t\), and let \(uv\) be the edge introduced
  at \(t\). Then \(X_{t} = X_{t'}\), \(V_{t} = V_{t'}\) and
  \(uv \in (E(G_{t}) \setminus E(G_{t'}))\). If \(\{u, v\} \cap X = \emptyset\)
  holds then the algorithm sets \(\calA = \emptyset\). The soundness criterion
  holds vacuously. Since edge \(uv\) is present in graph \(G_{t}\) there can
  exist no subgraph \(G_{t'}\) of \(G_{t}\) for which the first and third
  conditions of the definition of a valid
  partition---\autoref{def:valid_partitions_witnesses_des}---hold
  simultaneously. Thus there is no partition which is valid for the combination
  \((t, X, O)\). Hence the completeness condition holds vacuously as well.

  If \(\{u, v\} \cap X \neq \emptyset\) holds then the algorithm initializes
  \(\calA = VP[t', X, O]\). By the inductive assumption we have that every
  partition \(P' \in \calA = VP[t', X, O]\) is valid for the combination
  \((t', X, O)\). Note that while edge \(uv\) is \emph{available} for use in
  constructing a witness for \(((t, X, O), P)\), it is not \emph{mandatory} to
  use this edge in any such witness. Applying this observation, it is
  straightforward to verify that if a subgraph \(G_{t'}'\) of \(G_{t'}\) is a
  witness for \(((t', X, O), P')\) then it is also (i) a subgraph of \(G_{t}\),
  and (ii) a witness for \(((t, X, O), P')\). Thus all partitions in
  \(VP[t', X, O]\) are valid for the combination \((t, X, O)\).

  The algorithm adds zero or more partitions to \(\calA\) depending on how the
  set \(\{u, v\}\) intersects the sets \(X\) and \(O\). We analyze each choice
  made by the algorithm:
  \begin{enumerate}
  \item If \(u \notin X\) or \(v \notin X\) holds then the algorithm does not
    make further changes to \(\calA\): it sets \(\calA = VP[t', X, O]\). Since
    (i) the criteria for
    validity---\autoref{def:valid_partitions_witnesses_des}---are based only on
    graphs whose intersection with \(X_{t}\) is exactly the set \(X\), and (ii)
    the new edge \(uv\) does not have both end points in this set, it is
    intuitively clear that the relevant set of valid partitions should not
    change in this case. Formally, it is straightforward to verify using
    Definitions~\ref{def:residual_subgraph},
    \ref{def:valid_partitions_witnesses_des}, and~\ref{def:completion_des} that:
    \begin{itemize}
    \item a partition \(P\) of set \(X\) is valid for the combination
      \((t, X, O)\) if and only if it is valid for the combination
      \((t', X, O)\);
    \item a subgraph of \(G_{t}\) is a witness for \(((t, X, O), P)\) if and
      only if it is (i) a subgraph of \(G_{t'}\) and (ii) a witness for
      \(((t', X, O), P)\);
    \item a graph \(H\) is a residual subgraph with respect to \(t\) with
      \(V(H) \cap X_{t} = X\) if and only if \(H\) is a residual subgraph with
      respect to \(t'\) with \(V(H) \cap X_{t'} = X\); and,
    \item for any residual subgraph \(H\) with respect to \(t\) with
      \(V(H) \cap X_{t} = X\) and any partition \(P\) of \(X\),
      \(((t, X, O), P)\) completes \(H\) if and only if \(((t', X, O), P)\)
      completes \(H\).
    \end{itemize}
    By the inductive assumption we have that the set \(VP[t', X, O]\) of
    partitions is sound and complete for the combination \((t', X, O)\). It
    follows from the above equivalences that the set \(\calA = VP[t', X, O]\) is
    sound and complete for the combination \((t, X, O)\).
  \item If \(\{u,v\} \subseteq O\) then for each partition
    \(P' \in VP[t', X, O \setminus \{u,v\}]\),
    \begin{itemize}
    \item If vertices \(u,v\) are in the same block of \(P'\) then the algorithm
      adds \(P = P'\) to the set \calA.
    \item If vertices \(u,v\) are in different blocks of \(P'\) then the
      algorithm merges these two blocks of \(P'\) and adds the resulting
      partition \(P\)---with one fewer block than \(P'\)---to the set \calA.
    \end{itemize}
    In either case, by the inductive assumption we have that partition \(P'\) is
    valid for the combination \((t', X, O \setminus \{u,v\})\). Let \(G_{t'}''\)
    be (i) a subgraph of \(G_{t'}\) and (ii) a witness for
    \(((t', X, O \setminus \{u,v\}), P')\), and let
    \(G_{t}' = (V(G_{t'}''), E(G_{t'}'') \cup \{uv\})\) be the graph obtained
    from \(G_{t'}''\) by adding the edge \(uv\). Then \(G_{t}'\) is a subgraph
    of \(G_{t}\). Vertices \(u,v\) have even degrees in \(G_{t'}''\), and hence
    they have odd degrees in \(G_{t}'\). It is straightforward to verify that
    \(G_{t}'\) is a witness for \(((t, X, O), P)\). Thus the addition of
    partition \(P\) to \calA preserves the soundness of \calA.

    Now we prove completeness. So let \(H\) be a residual subgraph with respect
    to \(t\) with \(V(H) \cap X_{t} = X\), for which there exists a partition
    \(P=\{X^{1},X^{2},\ldots X^{p}\}\) of \(X\) such that \(((t, X, O), P)\)
    completes \(H\). We need to show that the set \calA computed by the
    algorithm contains some partition \(Q\) of \(X\) such that
    \(((t, X, O), Q)\) completes \(H\). Observe that there exists a subgraph
    \(G'_{t}\) of \(G_{t}\)---a witness for \(((t,X,O), P)\)---such that the
    following hold:
    \begin{enumerate}
    \item \(X_{t} \cap V(G'_{t})=X\).
    \item \(G'_{t}\) has exactly \(p\) connected components
      \(C_{1}, C_{2}, \dotsc, C_{p}\) and for each
      \(i \in \{1, 2, \dotsc, p\}$, $X^{i} \subseteq V(C_{i})\) holds.
    \item \(v^{\star} \in V(G'_{t})\) holds, and \(V(G'_{t})\) is a vertex cover
    of graph \(G_{t}\).
    \item The set of odd-degree vertices in \(G'_{t}\) is exactly the set
      \(O\).
    \item The graph \(G'_{t} \cup H\) is a dominating Eulerian subgraph of
      \(G\).
    \end{enumerate}
    Note that by the definition of a residual subgraph, graph \(H\) (i) does
    \emph{not} contain edge \(uv\), and (ii) is a residual subgraph with respect
    to node \(t'\) as well. We consider two cases.
    \begin{itemize}
    \item Suppose edge \(uv\) is not present in graph \(G'_{t}\). Then it is
      straightforward to verify that \(G'_{t}\) is a witness for
      \(((\mathbf{t'},X,O), P)\) as well. By the inductive hypothesis there
      exists some partition \(Q\) of \(X\) in the set \(VP[t', X, O]\) such that
      \(((t', X, O), Q)\) completes \(H\). So there exists a subgraph
      \(G'_{t'}\) of \(G_{t'}\) which is a certificate for \(((t',X,O), Q)\)
      completing \(H\). It is straightforward to verify that \(G'_{t'}\) is a
      certificate for \(((t,X,O), Q)\) completing \(H\) as well. The algorithm
      adds partition \(Q\) to the set \calA during the initialization, so the
      completeness criterion is satisfied in this case.
    \item Suppose edge \(uv\) \emph{is} present in graph \(G'_{t}\). Let
      \(H' = (V(H), (E(H) \cup \{uv\}))\) be the graph obtained by adding edge
      \(uv\) to graph \(H\), and let
      \(G'_{t'} = (V(G'_{t}), (E(G'_{t}) \setminus \{uv\}))\) be the graph
      obtained by deleting edge \(uv\) from graph \(G'_{t}\). Then it is
      straightforward to verify that (i) the set of odd-degree vertices in
      \(G'_{t'}\) is exactly the set \(O \setminus \{u,v\}\), (ii) \(H'\) is a
      residual subgraph for node \(t'\), and (iii) \(G'_{t'}\) is a subgraph of
      \(G_{t'}\) such that the graph \(G'_{t'} \cup H' = G'_{t} \cup H\) is a
      dominating Eulerian subgraph of \(G\). Let \(P'\) be the partition of
      \(X\) defined by graph \(G'_{t'}\). Then \(G'_{t'}\) is a witness for
      \(((t', X, O \setminus \{u,v\}), P')\) such that the union of \(G'_{t'}\)
      and the residual subgraph \(H'\) of \(t'\) is a dominating Eulerian
      subgraph of \(G\). That is, \(((t', X, O \setminus \{u,v\}), P')\)
      completes \(H'\). So by the inductive assumption there exists some
      partition \(Q'\) of \(X\) in the set \(VP[t', X, O \setminus \{u,v\}]\)
      such that \(((t, X, O \setminus \{u,v\}), Q')\) completes \(H'\). So there
      exists a subgraph \(\hat{G}'\) of \(G_{t'}\) such that (i) \(\hat{G}'\) is
      a witness for \(((t, X, O \setminus \{u,v\}), Q')\) and (ii)
      \(\hat{G}' \cup H'\) is a dominating Eulerian subgraph of \(G\).

      Note that \(Q'\) is the partition of set \(X\) defined by the graph
      \(\hat{G}'\). Suppose both \(u\) and \(v\) are in the same block of
      partition \(Q'\). Then adding the edge \(uv\) to \(\hat{G}'\) does not
      change the partition of \(X\) defined by \(\hat{G}'\). It follows that the
      graph \(\hat{G} = (V(\hat{G}'), E(\hat{G}') \cup \{uv\})\) is a subgraph
      of \(G_{t}\) such that (i) \(\hat{G}\) is a witness for
      \(((t, X, O , Q')\) and (ii) \(\hat{G} \cup H\) is a dominating Eulerian
      subgraph of \(G\). Thus \(((t, X, O, Q')\) completes the residual subgraph
      \(H\). Now notice that our algorithm adds the partition \(Q'\) to the set
      \calA. Thus the completeness criterion holds in this case.

      In the remaining case, vertices \(u\) and \(v\) are in distinct blocks of
      partition \(Q'\). Let \(Q\) be the partition obtained from \(Q'\) by
      merging together the two blocks to which vertices \(u\) and \(v\) belong,
      respectively, and leaving the other blocks as they are. Let \(\hat{G}\) be
      defined as in the previous paragraph. Then the partition of \(X\) defined
      by \(\hat{G}\) is \(Q\). It follows that \(\hat{G}\) is a subgraph of
      \(G_{t}\) such that (i) \(\hat{G}\) is a witness for \(((t, X, O , Q)\)
      and (ii) \(\hat{G} \cup H\) is a dominating Eulerian subgraph of \(G\).
      Thus \(((t, X, O, Q)\) completes the residual subgraph \(H\). Now notice
      that our algorithm adds the partition \(Q\) to the set \calA. Thus the
      completeness criterion holds in this case as well.
    \end{itemize}

  \item If \(\{u,v\} \cap O = \{u\}\) then for each partition
    \(P' \in VP[t', X, (O \setminus \{u\}) \cup \{v\}]\),
    \begin{itemize}
    \item If vertices \(u,v\) are in the same block of \(P'\) then the
      algorithm adds \(P = P'\) to the set \calA.
    \item If vertices \(u,v\) are in different blocks of \(P'\) then the
      algorithm merges these two blocks of \(P'\) and adds the resulting
      partition \(P\)---with one fewer block than \(P'\)---to the set \calA.
    \end{itemize}
    In either case, by the inductive assumption we have that partition \(P'\)
    is valid for the combination \((t', X, (O \setminus \{u\}) \cup \{v\})\).
    Let \(G_{t'}''\) be (i) a subgraph of \(G_{t'}\) and (ii) a witness for
    \(((t', X, (O \setminus \{u\}) \cup \{v\}), P')\), and let
    \(G_{t}' = (V(G_{t'}''), E(G_{t'}'') \cup \{uv\})\) be the graph obtained
    from \(G_{t'}''\) by adding the edge \(uv\). Then \(G_{t}'\) is a subgraph
    of \(G_{t}\). In \(G_{t'}''\) the degree of vertex \(u\) is even, and the
    degree of vertex \(v\) is odd. So in \(G_{t}'\) vertex \(u\) has an odd
    degree, and vertex \(v\) has an even degree. It is straightforward to
    verify that \(G_{t}'\) is a witness for \(((t, X, O), P)\). Thus the
    addition of partition \(P\) to \calA preserves the soundness of \calA.

    Now we prove completeness. So let \(H\) be a residual subgraph with
    respect to \(t\) with \(V(H) \cap X_{t} = X\), for which there exists a
    partition \(P=\{X^{1},X^{2},\ldots X^{p}\}\) of \(X\) such that
    \(((t, X, O), P)\) completes \(H\).  We
    need to show that the set \calA computed by the algorithm contains some
    partition \(Q\) of \(X\) such that \(((t, X, O), Q)\) completes \(H\).
    Observe that there exists a subgraph \(G'_{t}\) of \(G_{t}\)---a witness
    for \(((t,X,O), P)\)---such that the following hold:
    \begin{enumerate}
    \item \(X_{t} \cap V(G'_{t})=X\).
    \item \(G'_{t}\) has exactly \(p\) connected components
      \(C_{1}, C_{2}, \dotsc, C_{p}\) and for each
      \(i \in \{1, 2, \dotsc, p\}$, $X^{i} \subseteq V(C_{i})\) holds.
    \item \(v^{\star} \in V(G'_{t})\) holds, and \(V(G'_{t})\) is a vertex cover
    of graph \(G_{t}\).
    \item The set of odd-degree vertices in \(G'_{t}\) is exactly the set
      \(O\).
    \item The graph \(G'_{t} \cup H\) is a dominating Eulerian subgraph of
      \(G\).
    \end{enumerate}
    Note that by the definition of a residual subgraph, graph \(H\) (i) does
    \emph{not} contain edge \(uv\), and (ii) is a residual subgraph with
    respect to node \(t'\) as well. We consider two cases.
    \begin{itemize}
    \item Suppose edge \(uv\) is not present in graph \(G'_{t}\). Then it is
      straightforward to verify that \(G'_{t}\) is a witness for
      \(((\mathbf{t'},X,O), P)\) as well. By the inductive hypothesis there
      exists some partition \(Q\) of \(X\) in the set \(VP[t', X, O]\) such
      that \(((t, X, O), Q)\) completes \(H\). This same partition \(Q\) is
      present in the set \calA as well.
    \item Suppose edge \(uv\) \emph{is} present in graph \(G'_{t}\). Let
      \(H' = (V(H), (E(H) \cup \{uv\}))\) be the graph obtained by adding edge
      \(uv\) to graph \(H\), and let
      \(G'_{t'} = (V(G'_{t}), (E(G'_{t}) \setminus \{uv\}))\) be the graph
      obtained by deleting edge \(uv\) from graph \(G'_{t}\). Then it is
      straightforward to verify that (i) the set of odd-degree vertices in
      \(G'_{t'}\) is exactly the set \((O \setminus \{u\}) \cup \{v\}\), (ii)
      \(H'\) is a residual subgraph for node \(t'\), and (iii) \(G'_{t'}\) is a
      subgraph of \(G_{t'}\) such that the graph
      \(G'_{t'} \cup H' = G'_{t} \cup H\) is a dominating Eulerian subgraph of
      \(G\). Let \(P'\) be the partition of \(X\) defined by graph \(G'_{t'}\).
      Then \(G'_{t'}\) is a witness for
      \(((t', X, (O \setminus \{u\}) \cup \{v\}), P')\) such that the union of
      \(G'_{t'}\) and the residual subgraph \(H'\) of \(t'\) is a dominating
      Eulerian subgraph of \(G\). That is,
      \(((t', X, (O \setminus \{u\}) \cup \{v\}), P')\) completes \(H'\). So by
      the inductive assumption there exists some partition \(Q'\) of \(X\) in
      the set \(VP[t', X, (O \setminus \{u\}) \cup \{v\}]\) such that
      \(((t, X, (O \setminus \{u\}) \cup \{v\}), Q')\) completes \(H'\). So
      there exists a subgraph \(\hat{G}'\) of \(G_{t'}\) such that (i)
      \(\hat{G}'\) is a witness for
      \(((t, X, (O \setminus \{u\}) \cup \{v\}), Q')\) and (ii)
      \(\hat{G}' \cup H'\) is a dominating Eulerian subgraph of \(G\).

      Note that \(Q'\) is the partition of set \(X\) defined by the graph
      \(\hat{G}'\), and that the set of odd-degree vertices in \(\hat{G}'\) is
      exactly the set \((O \setminus \{u\}) \cup \{v\}\). Suppose both \(u\) and
      \(v\) are in the same block of partition \(Q'\). Then adding the edge
      \(uv\) to \(\hat{G}'\) (i) does not change the partition of \(X\) defined
      by \(\hat{G}'\), and (ii) \emph{does} change the set of odd-degree
      vertices to \(O\). It follows that the graph
      \(\hat{G} = (V(\hat{G}'), E(\hat{G}') \cup \{uv\})\) is a subgraph of
      \(G_{t}\) such that (i) \(\hat{G}\) is a witness for \(((t, X, O , Q')\)
      and (ii) \(\hat{G} \cup H\) is a dominating Eulerian subgraph of \(G\).
      Thus \(((t, X, O, Q')\) completes the residual subgraph \(H\). Now notice
      that our algorithm adds the partition \(Q'\) to the set \calA. Thus the
      completeness criterion holds in this case.

      In the remaining case, vertices \(u\) and \(v\) are in distinct blocks of
      partition \(Q'\). Let \(Q\) be the partition obtained from \(Q'\) by
      merging together the two blocks to which vertices \(u\) and \(v\) belong,
      respectively, and leaving the other blocks as they are. Let \(\hat{G}\) be
      defined as in the previous paragraph. Then the partition of \(X\) defined
      by \(\hat{G}\) is \(Q\). It follows that \(\hat{G}\) is a subgraph of
      \(G_{t}\) such that (i) \(\hat{G}\) is a witness for \(((t, X, O , Q)\)
      and (ii) \(\hat{G} \cup H\) is a dominating Eulerian subgraph of \(G\).
      Thus \(((t, X, O, Q)\) completes the residual subgraph \(H\). Now notice
      that our algorithm adds the partition \(Q\) to the set \calA. Thus the
      completeness criterion holds in this case as well.
    \end{itemize}
  \item The case when \(\{u,v\} \cap O = \{v\}\) is symmetrical to the previous
    case, so we leave out the arguments for this case.
  \item If \(\{u,v\} \cap O = \emptyset\) then for each partition
    \(P' \in VP[t', X, O \cup \{u, v\}]\),
    \begin{itemize}
    \item If vertices \(u,v\) are in the same block of \(P'\) then the algorithm
      adds \(P = P'\) to the set \calA.
    \item If vertices \(u,v\) are in different blocks of \(P'\) then the
      algorithm merges these two blocks of \(P'\) and adds the resulting
      partition \(P\)---with one fewer block than \(P'\)---to the set \calA.
    \end{itemize}
    In either case, by the inductive assumption we have that partition \(P'\) is
    valid for the combination \((t', X, O \cup \{u,v\})\). Let \(G_{t'}''\) be
    (i) a subgraph of \(G_{t'}\) and (ii) a witness for
    \(((t', X, O \cup \{u, v\}), P')\), and let
    \(G_{t}' = (V(G_{t'}''), E(G_{t'}'') \cup \{uv\})\) be the graph obtained
    from \(G_{t'}''\) by adding the edge \(uv\). Then \(G_{t}'\) is a subgraph
    of \(G_{t}\). Vertices \(u,v\) have odd degrees in \(G_{t'}''\), and hence
    they have even degrees in \(G_{t}'\). It is straightforward to verify that
    \(G_{t}'\) is a witness for \(((t, X, O), P)\). Thus the addition of
    partition \(P\) to \calA preserves the soundness of \calA.

    Now we prove completeness. So let \(H\) be a residual subgraph with respect
    to \(t\) with \(V(H) \cap X_{t} = X\), for which there exists a partition
    \(P=\{X^{1},X^{2},\ldots X^{p}\}\) of \(X\) such that \(((t, X, O), P)\)
    completes \(H\). We need to show that the set \calA computed by the
    algorithm contains some partition \(Q\) of \(X\) such that
    \(((t, X, O), Q)\) completes \(H\). Observe that there exists a subgraph
    \(G'_{t}\) of \(G_{t}\)---a witness for \(((t,X,O), P)\)---such that the
    following hold:
    \begin{enumerate}
    \item \(X_{t} \cap V(G'_{t})=X\).
    \item \(G'_{t}\) has exactly \(p\) connected components
      \(C_{1}, C_{2}, \dotsc, C_{p}\) and for each
      \(i \in \{1, 2, \dotsc, p\}$, $X^{i} \subseteq V(C_{i})\) holds.
    \item \(v^{\star} \in V(G'_{t})\) holds, and \(V(G'_{t})\) is a vertex cover
    of graph \(G_{t}\).
    \item The set of odd-degree vertices in \(G'_{t}\) is exactly the set
      \(O\).
    \item The graph \(G'_{t} \cup H\) is a dominating Eulerian subgraph of
      \(G\).
    \end{enumerate}
    Note that by the definition of a residual subgraph, graph \(H\) (i) does
    \emph{not} contain edge \(uv\), and (ii) is a residual subgraph with
    respect to node \(t'\) as well. We consider two cases.
    \begin{itemize}
    \item Suppose edge \(uv\) is not present in graph \(G'_{t}\). Then it is
      straightforward to verify that \(G'_{t}\) is a witness for
      \(((\mathbf{t'},X,O), P)\) as well. By the inductive hypothesis there
      exists some partition \(Q\) of \(X\) in the set \(VP[t', X, O]\) such
      that \(((t, X, O), Q)\) completes \(H\). This same partition \(Q\) is
      present in the set \calA as well.
    \item Suppose edge \(uv\) \emph{is} present in graph \(G'_{t}\). Let
      \(H' = (V(H), (E(H) \cup \{uv\}))\) be the graph obtained by adding edge
      \(uv\) to graph \(H\), and let
      \(G'_{t'} = (V(G'_{t}), (E(G'_{t}) \setminus \{uv\}))\) be the graph
      obtained by deleting edge \(uv\) from graph \(G'_{t}\). Then it is
      straightforward to verify that (i) the set of odd-degree vertices in
      \(G'_{t'}\) is exactly the set \(O \setminus \{u, v\}\), (ii) \(H'\) is a
      residual subgraph for node \(t'\), and (iii) \(G'_{t'}\) is a subgraph of
      \(G_{t'}\) such that the graph \(G'_{t'} \cup H' = G'_{t} \cup H\) is a
      dominating Eulerian subgraph of \(G\). Let \(P'\) be the partition of
      \(X\) defined by graph \(G'_{t'}\). Then \(G'_{t'}\) is a witness for
      \(((t', X, O \setminus \{u, v\}), P')\) such that the union of \(G'_{t'}\)
      and the residual subgraph \(H'\) of \(t'\) is a dominating Eulerian
      subgraph of \(G\). That is, \(((t', X, O \setminus \{u, v\}), P')\)
      completes \(H'\). So by the inductive assumption there exists some
      partition \(Q'\) of \(X\) in the set \(VP[t', X, O \setminus \{u, v\}]\)
      such that \(((t, X, O \setminus \{u, v\}), Q')\) completes \(H'\). So
      there exists a subgraph \(\hat{G}'\) of \(G_{t'}\) such that (i)
      \(\hat{G}'\) is a witness for \(((t, X, O \setminus \{u, v\}), Q')\) and
      (ii) \(\hat{G}' \cup H'\) is a dominating Eulerian subgraph of \(G\).

      Note that \(Q'\) is the partition of set \(X\) defined by the graph
      \(\hat{G}'\), and that the set of odd-degree vertices in \(\hat{G}'\) is
      exactly the set \(O \setminus \{u, v\}\). Suppose both \(u\) and \(v\) are
      in the same block of partition \(Q'\). Then adding the edge \(uv\) to
      \(\hat{G}'\) (i) does not change the partition of \(X\) defined by
      \(\hat{G}'\), and (ii) \emph{does} change the set of odd-degree vertices
      to \(O\). It follows that the graph
      \(\hat{G} = (V(\hat{G}'), E(\hat{G}') \cup \{uv\})\) is a subgraph of
      \(G_{t}\) such that (i) \(\hat{G}\) is a witness for \(((t, X, O , Q')\)
      and (ii) \(\hat{G} \cup H\) is a dominating Eulerian subgraph of \(G\).
      Thus \(((t, X, O, Q')\) completes the residual subgraph \(H\). Now notice
      that our algorithm adds the partition \(Q'\) to the set \calA. Thus the
      completeness criterion holds in this case.

      In the remaining case, vertices \(u\) and \(v\) are in distinct blocks of
      partition \(Q'\). Let \(Q\) be the partition obtained from \(Q'\) by
      merging together the two blocks to which vertices \(u\) and \(v\) belong,
      respectively, and leaving the other blocks as they are. Let \(\hat{G}\) be
      defined as in the previous paragraph. Then the partition of \(X\) defined
      by \(\hat{G}\) is \(Q\). It follows that \(\hat{G}\) is a subgraph of
      \(G_{t}\) such that (i) \(\hat{G}\) is a witness for \(((t, X, O , Q)\)
      and (ii) \(\hat{G} \cup H\) is a dominating Eulerian subgraph of \(G\).
      Thus \(((t, X, O, Q)\) completes the residual subgraph \(H\). Now notice
      that our algorithm adds the partition \(Q\) to the set \calA. Thus the
      completeness criterion holds in this case as well. \qedhere
    \end{itemize}
  \end{enumerate}
\end{proof}

\begin{lemma}\label{lem:forget_node_ok_des}
  Let \(t\) be a forget node of the tree decomposition \TT and let
  \(X \subseteq X_{t}, O \subseteq X\) be arbitrary subsets of \(X_{t}, X\)
  respectively. The collection \calA of partitions computed by the DP for the
  combination \((t, X, O)\) satisfies the correctness criteria.
\end{lemma}
\begin{proof}
  Let \(t'\) be the child node of \(t\), and let \(v\) be the vertex forgotten
  at \(t\). Then \(v \in X_{t'}\) and \(X_{t} = X_{t'} \setminus \{v\}\), and
  \(v \notin O\) hold. Recall that \(P(v)\) is the block of partition \(P\)
  which contains element \(v\) and that \(P - v\) is the partition obtained by
  eliding \(v\) from \(P\). The algorithm adds all partitions in the set
  \(\{P' - v \;;\; P' \in VP[t', X \cup \{v\}, O],\,|P'(v)| > 1\}\) to \calA. By
  the inductive assumption we have that every partition
  \(P' \in VP[t', X \cup \{v\}, O]\) is valid for the combination
  \((t', X \cup \{v\}, O)\). Note that (i) the graph \(G_{t'}\) is identical to
  the graph \(G_{t}\), and (ii) for any subgraph \(H\) of \(G_{t'} = G_{t}\),
  \((V(H) \cap X_{t'}) = X \cup \{v\}\) implies \((V(H) \cap X_{t}) = X\). It
  follows that if every connected component of a graph \(H\) contains at least
  two vertices from the set \(X \cup \{v\}\) then every connected component of
  \(H\) contains at least one vertex from set \(X\). Using these observations it
  is straightforward to verify that if a subgraph \(G_{t'}'\) of \(G_{t'}\) is a
  witness for \(((t', X \cup \{v\}, O), P')\) where \(v \notin O\) and
  \(|P'(v)| > 1\) hold, then it is also (i) a subgraph of \(G_{t}\), and (ii) a
  witness for \(((t, X, O), P' - v)\). Thus for each partition
  \(P' \in VP[t', X \cup \{v\}, O],\,|P'(v)| > 1\) the partition \(P' - v\) is
  valid for the combination \((t, X, O)\). Thus all partitions in the set
  \(\{P' - v \;;\; P' \in VP[t', X \cup \{v\}, O]\}\) are valid for
  \((t, X, O)\).

  The algorithm also adds all the partitions from \(VP[t', X, O]\) to \(\calA\).
  By the inductive assumption we have that every partition
  \(P' \in VP[t', X, O]\) is valid for the combination \((t', X, O)\). It is
  once again straightforward to verify that if a subgraph \(G_{t'}'\) of
  \(G_{t'}\) is a witness for \(((t', X, O), P')\) then it is also (i) a
  subgraph of \(G_{t}\), and (ii) a witness for \(((t, X, O), P')\). Thus each
  partition \(P' \in VP[t', X, O]\) is valid for the combination \((t, X, O)\).
  Hence all partitions added to the set \calA by the algorithm are valid for
  \((t, X, O)\).

  We now argue that the set \calA satisfies the completeness criterion. So let
  \(H\) be a residual subgraph with respect to \(t\) with
  \(V(H) \cap X_{t} = X\), for which there exists a partition
  \(P=\{X^{1},X^{2},\ldots X^{p}\}\) of \(X\) such that \(((t, X, O), P)\)
  completes \(H\). We need to show that the set \calA computed by the algorithm
  contains some partition \(Q\) of \(X\) such that \(((t, X, O), Q)\) completes
  \(H\). Observe that there exists a subgraph \(G'_{t}\) of \(G_{t}\)---a
  witness for \(((t,X,O), P)\)---such that the following hold:
    \begin{enumerate}
    \item \(X_{t} \cap V(G'_{t})=X\).
    \item \(G'_{t}\) has exactly \(p\) connected components
      \(C_{1}, C_{2}, \dotsc, C_{p}\) and for each
      \(i \in \{1, 2, \dotsc, p\}$, $X^{i} \subseteq V(C_{i})\) holds.
    \item \(v^{\star} \in V(G'_{t})\) holds, and \(V(G'_{t})\) is a vertex cover
    of graph \(G_{t}\).
    \item The set of odd-degree vertices in \(G'_{t}\) is exactly the set \(O\).
    \item The graph \(G'_{t} \cup H\) is a dominating Eulerian subgraph of
      \(G\).
    \end{enumerate}

    Suppose graph \(G'_{t}\) does \emph{not} contain vertex \(v\). Then it is
    easy to verify that \(H\) is a residual subgraph with respect to \(t'\) with
    \(V(H) \cap X_{t'} = X\), and that graph \(G'_{t}\) is a witness for
    \(((t',X,O), P)\) such that the union of graphs \(H\) and \(G'_{t}\) is a
    dominating Eulerian subgraph of \(G\). That is, \(((t', X, O), P)\)
    completes \(H\). By inductive assumption there exists a partition
    \(Q \in VP[t', X, O]\) such that \(((t', X, O), Q)\) completes \(H\). Since
    the algorithm adds this partition \(Q\) to \calA we get that \calA satisfies
    the completeness criterion in this case.

    Now suppose graph \(G'_{t}\) contains vertex \(v\). From the definition of a
    residual subgraph we know that \(v \notin V(H)\) holds. Without loss of
    generality, let it be the case that \(v \in C_{p}\) holds. Since
    \(X_{t'} = X_{t} \cup \{v\}\) we get that
    \(X_{t'} \cap V(G'_{t}) = X \cup \{v\}\) holds. Let
    \(H' = (V(H) \cup \{v\}, E(H))\) be the graph obtained by adding vertex
    \(v\) (and no extra edges) to graph \(H\). Then it is straightforward to
    verify that (i) \(H'\) is a residual subgraph with respect to \(t'\) with
    \(V(H') \cap X_{t'} = X \cup \{v\}\), (ii) the graph \(G'_{t}\) is a witness
    for the partition \(P' = \{X^{1},X^{2},\ldots (X^{p} \cup \{v\}) \}\) of
    \(X \cup \{v\}\) being valid for the combination \((t', X \cup \{v\}, O)\),
    and (iii) the graph \(G'_{t} \cup H'\) is a dominating Eulerian subgraph of
    \(G\). That is, \(((t', X \cup \{v\}, O), P')\) completes \(H'\).

    By the inductive assumption there exists some partition \(Q'\) of
    \(X \cup \{v\}\) in the set \(VP[t', X \cup \{v\}, O]\}\) such that
    \(((t', X \cup \{v\}, O), Q')\) completes \(H'\). So there exists a subgraph
    \(\hat{G}'\) of \(G_{t'}\) such that (i) \(\hat{G}'\) is a witness for
    \(((t', X \cup \{v\}, O), Q')\) and (ii) \(\hat{G}' \cup H'\) is a
    dominating Eulerian subgraph of \(G\). Note that
    \(X_{t} \cap V(\hat{G}') = X\) holds.

    Since \(v\) had degree zero in graph \(H'\) we get that \(v\) has a positive
    even degree in \(\hat{G}'\). From the definition of a witness for
    validity---\autoref{def:valid_partitions_witnesses_des}---we get that \(Q'\)
    is the partition of the set \(X \cup \{v\}\) defined by the graph
    \(\hat{G}'\). Let \(Q_{H'}\) be the partition of the set \(X \cup \{v\}\)
    defined by the graph \(H'\). Since \(deg_{H'}(v) = 0\) holds we get that
    vertex \(v\) appears in a block of size one---namely, \(\{v\}\)---in
    \(Q_{H'}\). If \(\{v\}\) is a block of \(Q'\) as well, then \(\{v\}\) will
    also be a block in their join \(Q_{H'} \sqcup Q'\). But the union of graphs
    \(H'\) and \(\hat{G}'\) is connected and so from
    \autoref{lem:partition_join_connectivity} we know that
    \(Q_{H'} \sqcup Q' = \{\{X \cup \{v\}\}\}\). Thus \(\{v\}\) is \emph{not} a
    block of \(Q_{H'} \sqcup Q'\), or of \(Q'\). So there exists a vertex
    \(v' \in X \) such that \(v,v'\) are in the same block of \(Q'\). In
    particular, this implies that the partition \(Q = Q' - v\), which is the
    partition of set \(X\) defined by graph \(\hat{G}'\), has exactly as many
    blocks as has the partition \(Q'\) of \(X \cup \{v\}\).

    Putting these together we get that the subgraph \(\hat{G}'\) of \(G_{t}\) is
    a witness for \(((t, X, O), Q = Q' - v)\). Now since graph \(H\) can be
    obtained from graph \(H'\) by deleting vertex \(v\), we get that the graphs
    \(\hat{G}' \cup H'\) and \(\hat{G}' \cup H\) are identical. In particular,
    the latter is a dominating Eulerian subgraph of \(G\). Thus
    \(((t, X, O), Q)\) completes the residual graph \(H\). Since the algorithm
    adds partition \(Q\) to the set \calA, we get that \calA satisfies the
    completeness criterion.\qedhere
\end{proof}

\begin{lemma}\label{lem:join_node_ok_des}
  Let \(t\) be a join node of the tree decomposition \TT and let
  \(X \subseteq X_{t}, O \subseteq X\) be arbitrary subsets of \(X_{t}, X\)
  respectively. The collection \calA of partitions computed by the DP for the
  combination \((t, X, O)\) satisfies the correctness criteria.
\end{lemma}
\begin{proof}
  Let \(t_{1}, t_{2}\) be the children of \(t\). Then
  \(X_{t} = X_{t_{1}} = X_{t_{2}}\). Note that
  \(V(G_{t}) = V(G_{t_{1}}) \cup V(G_{t_{2}}) \) and
  \(E(G_{t}) = E(G_{t_{1}}) \cup E(G_{t_{2}}) \) hold, and so graph \(G_{t}\) is
  the union of graphs \(G_{t_{1}}\) and \(G_{t_{2}}\). Further, since each edge
  in the graph is introduced at exactly one bag in \TT we get that
  \(E(G_{t_{1}}) \cap E(G_{t_{2}}) = \emptyset\) holds. Moreover,
  \(V(G_{t_{1}}) \cap V(G_{t_{2}}) = X_{t}\) holds as well. The algorithm
  initializes \calA to the empty set. For each way of dividing set \(O\) into
  two disjoint subsets \(O_{1},O_{2}\) (one of which could be empty) and for
  each subset \(\hat{O}\) (which could also be empty) of the set
  \(X \setminus O\), the algorithm picks a number of pairs \((P_{1}, P_{2})\) of
  partitions and adds their joins \(P_{1} \sqcup P_{2}\) to the set \calA. We
  first show that the partition \(P_{1} \sqcup P_{2}\) is valid for the
  combination \((t, X, O)\), for each choice of pairs \((P_{1}, P_{2})\) made by
  the algorithm.

  So let
  \(P_{1} \in VP[t_{1}, X, O_{1} \cup \hat{O}], P_{2} \in VP[t_{2}, X, O_{2}
  \cup \hat{O}]\). By the inductive hypothesis we get that \(P_{1}\) is valid
  for the combination \((t_{1}, X, O_{1} \cup \hat{O})\) and \(P_{2}\) is valid
  for the combination \((t_{2}, X, O_{2} \cup \hat{O})\). So there exist
  subgraphs \(G'_{t_{1}} = (V'_{t_{1}},E'_{t_{1}})\) of \(G_{t_{1}}\) and
  \(G'_{t_{2}} = (V'_{t_{2}},E'_{t_{2}})\) of \(G_{t_{2}}\) such that
    \begin{enumerate}
    \item \(X_{t} \cap V'_{t_{1}} = X = X_{t} \cap V'_{t_{2}}\).
    \item The vertex set of each connected component of \(G'_{t_{1}}\) and of
      \(G'_{t_{2}}\) has a non-empty intersection with set \(X\). Moreover,
      \(P_{1}\) is the partition of \(X\) defined by the subgraph \(G'_{t_{1}}\)
      and \(P_{2}\) is the partition of \(X\) defined by the subgraph
      \(G'_{t_{2}}\).
    \item Both \(v^{\star} \in V'_{t_{1}}\) and \(v^{\star} \in V'_{t_{2}}\)
      hold. Further, \(V'_{t_{1}}\) is a vertex cover of graph \(G_{t_{1}}\) and
      \(V'_{t_{2}}\) is a vertex cover of graph \(G_{t_{2}}\).
    \item The set of odd-degree vertices in \(G'_{t_{1}}\) is exactly the set
      \(O_{1} \cup \hat{O}\) and the set of odd-degree vertices in
      \(G'_{t_{2}}\) is exactly the set \(O_{2} \cup \hat{O}\).
    \end{enumerate}
    Let \(G'_{t} = G'_{t_{1}} \cup G'_{t_{2}}\). Then \(G'_{t}\) is a subgraph
    of \(G_{t}\), and
    \begin{enumerate}
    \item Since \(X_{t} \cap V'_{t_{1}} = X = X_{t} \cap V'_{t_{2}}\) holds we
      have that \(X_{t} \cap V(G'_{t})=X\) holds as well.
    \item The vertex set of each connected component of \(G'_{t}\) has a
      non-empty intersection with set \(X\). Moreover, from
      \autoref{lem:partition_join_connectivity} we get that
      \(P_{1} \sqcup P_{2}\) is the partition of \(X\) defined by the subgraph
      \(G'_{t}\).
    \item \(v^{\star} \in V(G'_{t})\) holds, and \(V(G'_{t})\) is a vertex cover
    of graph \(G_{t}\).
    \item Since \(E(G_{t_{1}}) \cap E(G_{t_{2}}) = \emptyset\) holds we get that
      the degree of any vertex \(v\) in graph \({G'_{t}}\) is the sum of its
      degrees in the two graphs \(G'_{t_{1}}\) and \(G'_{t_{2}}\). Since (i) the
      set of odd-degree vertices in graph \(G'_{t_{1}}\) is exactly the set
      \(O_{1} \cup \hat{O}\), (ii) the set of odd-degree vertices in graph
      \(G'_{t_{2}}\) is exactly the set \(O_{2} \cup \hat{O}\), and (iii) \(O\)
      is the disjoint union of sets \(O_{1}\) and \(O_{2}\), we get that the set
      of odd-degree vertices in graph \(G'_{t}\) is exactly the set \(O\).
    \end{enumerate}
    Thus graph \(G'_{t}\) is a witness for partition \(P_{1} \sqcup P_{2}\)
    being valid for the combination \((t, X, O)\), and so partition
    \(P_{1} \sqcup P_{2} \in \calA\) is valid for the combination \((t, X, O)\).
    This proves that collection \calA satisfies the soundness criterion.

    We now argue that the set \calA satisfies the completeness criterion. So let
    \(H\) be a residual subgraph with respect to \(t\) with
    \(V(H) \cap X_{t} = X\), for which there exists a partition
    \(P=\{X^{1},X^{2},\ldots X^{p}\}\) of \(X\) such that \(((t, X, O), P)\)
    completes \(H\). We need to show that the set \calA computed by the
    algorithm contains some partition \(Q\) of \(X\) such that
    \(((t, X, O), Q)\) completes \(H\). Observe that there exists a subgraph
    \(G'_{t}\) of \(G_{t}\)---a witness for \(((t,X,O), P)\)---such that the
    following hold:
    \begin{enumerate}
    \item \(X_{t} \cap V(G'_{t})=X\).
    \item \(G'_{t}\) has exactly \(p\) connected components
      \(C_{1}, C_{2}, \dotsc, C_{p}\) and for each
      \(i \in \{1, 2, \dotsc, p\}$, $X^{i} \subseteq V(C_{i})\) holds.
    \item \(v^{\star} \in V(G'_{t})\) holds, and \(V(G'_{t})\) is a vertex cover
    of graph \(G_{t}\).
    \item The set of odd-degree vertices in \(G'_{t}\) is exactly the set \(O\).
    \item The graph \(G'_{t} \cup H\) is a dominating Eulerian subgraph of
      \(G\).
    \end{enumerate}
    Let \(G_{1} = (V(G'_{t}) \cap V(G_{t_{1}}), E(G'_{t}) \cap E(G_{t_{1}}))\)
    and \(G_{2} = (V(G'_{t}) \cap V(G_{t_{2}}), E(G'_{t}) \cap E(G_{t_{2}}))\)
    be, respectively, the subgraphs of \(G'_{t}\) defined by the subtrees of \TT
    rooted at nodes \(t_{1}\) and \(t_{2}\), respectively. Then
    \(G'_{t} = G_{1} \cup G_{2}\),
    \(V(G_{1}) \cap X_{t_{1}} = V(G_{2}) \cap X_{t_{2}}= V(G_{1}) \cap V(G_{2})
    = X\), and \(E(G_{1}) \cap E(G_{2}) = \emptyset\) all hold. Let
    \(\tilde{O_{1}}, \tilde{O_{2}}\) be the sets of vertices of odd degree in
    graphs \(G_{1}, G_{2}\), respectively. Since graph
    \((H \cup G_{1}) \cup G_{2}\) is Eulerian and since
    \(V(H \cup G_{1}) \cap V(G_{2}) = X\) holds, we get that (i)
    \(\tilde{O_{2}} \subseteq X\) holds, and (ii) every connected component of
    graph \(G_{2}\) contains at least one vertex from set \(X\). By symmetric
    reasoning we get that (i) \(\tilde{O_{1}} \subseteq X\) holds, and (ii)
    every connected component of graph \(G_{1}\) contains at least one vertex
    from set \(X\). Let \(O_{2} = \tilde{O_{2}} \cap O\) and
    \(\hat{O} = \tilde{O_{2}} \setminus O\). Then
    \(\tilde{O_{2}} = O_{2} \cup \hat{O}\). Define
    \(O_{1} = O \setminus O_{2}\). Since (i) the set of odd-degree vertices in
    graph \(G'_{t}\) is exactly the set \(O\), and (ii)
    \(E(G_{1}) \cap E(G_{2}) = \emptyset\) holds, we get that the set of
    odd-degree vertices in graph \(G_{1}\) is
    \(\tilde{O_{1}} =(O \setminus O_{2}) \cup \hat{O} = O_{1} \cup \hat{O}\).

    Let \(Q_{2}\) be the partition of set \(X\) defined by graph \(G_{2}\), and
    let \(R_{1} = H \cup G_{1}\). It is straightforward to verify the following:
    (i) \(R_{1}\) is a residual subgraph with respect to node \(t_{2}\) with
    \(V(R_{1}) \cap X_{t_{2}} = X\); (ii) graph \(G_{2}\) is a witness for
    partition \(Q_{2}\) being valid for the combination
    \((t_{2}, X, \tilde{O_{2}})\), and (iii) \(G_{2}\) is a certificate for
    \(((t_{2}, X, \tilde{O_{2}}), Q_{2})\) completing the residual graph
    \(R_{1}\). By the inductive assumption there is a partition \(P_{2}\) of
    \(X\) in the set \(VP[t_{2}, X, O_{2} \cup \hat{O}]\) such that
    \(((t_{2}, X, O_{2} \cup \hat{O}), P_{2})\) completes the residual graph
    \(R_{1}\). Let \(H_{2}\) be a certificate for
    \(((t_{2}, X, O_{2} \cup \hat{O}), P_{2})\) completing \(R_{1}\). Note that
    \(H_{2}\) is a subgraph of \(G_{t_{2}}\), and that
    \(R_{1} \cup H_{2} = (H \cup G_{1}) \cup H_{2}\) is a dominating Eulerian
    subgraph of \(G\).

    Let \(Q_{1}\) be the partition of set \(X\) defined by graph \(G_{1}\), and
    let \(R_{2} = H \cup H_{2}\). From \autoref{lem:completion_odd_subset_des}
    we get that the set of odd-degree vertices of the residual subgraph \(H\) is
    exactly the set \(O\), and from
    Definitions~\ref{def:valid_partitions_witnesses_des}
    and~\ref{def:completion_des} we get that the set of odd-degree vertices of
    graph \(H_{2}\) is the set \(O_{2} \cup \hat{O}\). From the definition of a
    residual subgraph we get that \(E(H) \cap E(H_{2}) = \emptyset\) holds. It
    follows that the set of odd-degree vertices of graph \(R_{2}\) is
    \((O \setminus O_{2}) \cup \hat{O} = O_{1} \cup \hat{O}\), which is exactly
    the set of odd-degree vertices of graph \(G_{1}\).

    It is now straightforward to verify the following: (i) \(R_{2}\) is a
    residual subgraph with respect to node \(t_{1}\) with
    \(V(R_{2}) \cap X_{t_{1}} = X\); (ii) graph \(G_{1}\) is a witness for
    partition \(Q_{1}\) being valid for the combination
    \((t_{1}, X, O_{1} \cup \hat{O})\), and (iii) \(G_{1}\) is a certificate for
    \(((t_{1}, X, O_{1} \cup \hat{O}), Q_{1})\) completing the residual graph
    \(R_{2}\). By the inductive assumption there is a partition \(P_{1}\) of
    \(X\) in the set \(VP[t_{1}, X, O_{1} \cup \hat{O}]\) such that
    \(((t_{1}, X, O_{1} \cup \hat{O}), P_{1})\) completes the residual graph
    \(R_{2}\). Let \(H_{1}\) be a certificate for
    \(((t_{1}, X, O_{1} \cup \hat{O}), P_{1})\) completing \(R_{2}\). Note that
    \(H_{1}\) is a subgraph of \(G_{t_{1}}\), and that
    \(R_{2} \cup H_{1} = (H \cup H_{2}) \cup H_{1}\) is a dominating Eulerian
    subgraph of \(G\).

    Let \(\hat{H} = H_{1} \cup H_{2}\). Then \(\hat{H}\) is a subgraph of
    \(G_{t}\), and
    \begin{enumerate}
    \item Since \(X_{t} \cap V(H_{1}) = X = X_{t} \cap V(H_{2})\) holds we have
      that \(X_{t} \cap V(\hat{H})=X\) holds as well.
    \item The vertex set of each connected component of \(\hat{H}\) has a
      non-empty intersection with set \(X\). Moreover, from
      \autoref{lem:partition_join_connectivity} we get that
      \(P_{1} \sqcup P_{2}\) is the partition of \(X\) defined by the subgraph
      \(\hat{H}\).
    \item \(v^{\star} \in V(\hat{H})\) holds, and \(V(\hat{H})\) is a vertex
      cover of graph \(G_{t}\).
    \item Since \(E(G_{t_{1}}) \cap E(G_{t_{2}}) = \emptyset\) holds we get that
      the degree of any vertex \(v\) in graph \(\hat{H}\) is the sum of its
      degrees in the two graphs \(H_{1}\) and \(H_{2}\). Since (i) the set of
      odd-degree vertices in graph \(H_{1}\) is exactly the set
      \(O_{1} \cup \hat{O}\), (ii) the set of odd-degree vertices in graph
      \(H_{2}\) is exactly the set \(O_{2} \cup \hat{O}\), and (iii) \(O\) is
      the disjoint union of sets \(O_{1}\) and \(O_{2}\), we get that the set of
      odd-degree vertices in graph \(\hat{H}\) is exactly the set \(O\).
    \end{enumerate}
    Graph \(\hat{H}\) is thus a witness for partition \(P_{1} \sqcup P_{2}\) of
    \(X\) being valid for the combination \((t, X, O)\), and \(H \cup \hat{H}\)
    is a dominating Eulerian subgraph of \(G\). Thus
    \(((t, X, O), P_{1} \sqcup P_{2})\) completes \(H\). Since the algorithm
    adds partition \(P_{1} \sqcup P_{2}\) to the set \(\calA\) we get that \calA
    satisfies the completeness criterion. \qedhere
\end{proof}

We can now prove
\begingroup
\def\thetheorem{\ref{thm:DES_is_FPT}}
\begin{theorem}
  There is an algorithm which solves an instance \((G,\mathcal{T},tw)\) of
  \DES in \(\OhStar{(1 + 2^{(\omega + 3)})^{tw}}\) time.
\end{theorem}
\addtocounter{theorem}{-1}
\endgroup
\begin{proof}
  We first modify \TT to make it a ``nearly-nice'' tree decomposition rooted at
  \(r\) as described at the start of this section. We then execute the dynamic
  programming steps described above on \TT. We return \yes if the element
  \(\{\{v^{\star}\}\}\) is present in the set
  \(VP[r,X = \{v^{\star}\}, O = \emptyset]\) computed by the DP, and \no
  otherwise.

  From \autoref{lem:completion_at_root_des} we know that \((G, \TT, tw)\) is
  a \yes instance of \DES if and only if the combination
  \(((r,X = \{v^{\star}\}, O = \emptyset), P = \{\{v^{\star}\}\})\) completes
  the residual graph \(H = (\{v^{\star}\}, \emptyset)\). By induction on the
  structure of the tree decomposition \TT and using
  \autoref{obs:repset_computation_preserves_correctness_des} and
  Lemmas~\ref{lem:leaf_node_ok_des},~\ref{lem:introduce_vertex_node_ok_des},~\ref{lem:introduce_edge_node_ok_des},~\ref{lem:forget_node_ok_des},
  and~\ref{lem:join_node_ok_des} we get that the set
  \(VP[r,X = \{v^{\star}\}, O = \emptyset]\) computed by the algorithm satisfies
  the correctness criteria. And since \(\{\{v^{\star}\}\}\) is the unique
  partition of set \(\{v^{\star}\}\) we get that the set
  \(VP[r,X = \{v^{\star}\}, O = \emptyset]\) computed by the algorithm will
  contain the partition \(\{\{v^{\star}\}\}\) if and only if \((G, \TT, tw)\)
  is a \yes instance of \DES.

  Note that we compute representative subsets as the last step in the
  computation at each bag. So we get, while performing computations at an
  intermediate node \(t\), that the number of partitions in any set
  \(VP[t', X', \cdot]\) for any \emph{child} node \(t'\) of \(t\) and subset
  \(X'\) of \(X_{t'}\) is at most \(2^{(|X'| - 1)}\) (See
  \autoref{thm:computing_representative_subsets}). We use
  \autoref{fac:partition_join_is_union_of_graphs} to perform various operations
  on one or two partitions---such as adding a block to a partition, merging two
  blocks of a partition, eliding an element from a partition, or computing the
  join of two partitions---in polynomial time.

  The computation at each \textbf{leaf node} of \TT can be done in constant
  time.
  
  For an \textbf{introduce vertex node} or an \textbf{introduce edge node} or a
  \textbf{forget node} \(t\) and a fixed pair of subsets
  \(X \subseteq X_{t}, O \subseteq X\), the computation of set \calA
  involves---in the worst case---spending polynomial time for each partition
  \(P'\) in some set \(VP[t', X' \subseteq X, \cdot]\). Since the number of
  partitions in this latter set is at most \(2^{(|X'| - 1)} \leq 2^{(|X| - 1)}\)
  we get that the set \calA can be computed in \(\OhStar{2^{(|X| - 1)}}\) time,
  and that the set \calB can be computed---see
  \autoref{thm:computing_representative_subsets}---in
  \(\OhStar{2^{(|X| - 1)} \cdot 2^{(\omega - 1)\cdot |X|}} = \OhStar{2^{\omega
      \cdot |X|}}\) time. Since the number of ways of choosing the subset
  \(O \subseteq X\) is \(2^{|X|}\) the entire computation at an introduce
  vertex, introduce edge, or forget node \(t\) can be done in time
  \begin{align*}
    \sum_{|X| = 0}^{|X_{t}|}\binom{|X_{t}|}{|X|} 2^{|X|} \OhStar{2^{\omega \cdot
    |X|}} &= \OhStar{\sum_{|X| = 0}^{tw + 1}\binom{tw + 1}{|X|} 2^{(\omega + 1)|X|}}\\
          &= \OhStar{(1 + 2^{(\omega + 1)})^{(tw + 1)}}\\
          &= \OhStar{(1 + 2\cdot 2^{\omega })^{tw}}. 
  \end{align*}

  For a \textbf{join node} \(t\) and a fixed subset \(X \subseteq X_{t}\) we
  guess three pairwise disjoint subsets \(\hat{O}, O_{1}, O_{2}\) of \(X\) in
  time \(4^{|X|}\). For each guess we go over all partitions
  \(P_{1} \in VP[t_{1}, X, O_{1} \cup \hat{O}], P_{2} \in VP[t_{2}, X, O_{2}
  \cup \hat{O}]\) and add their join \(P_{1} \sqcup P_{2}\) to the set \calA.
  Since the number of partitions in each of the two sets
  \(VP[t_{1}, X, O_{1} \cup \hat{O}], VP[t_{2}, X, O_{2} \cup \hat{O}]\) is at
  most \(2^{(|X| - 1)}\), the size of set \calA is at most \(2^{(2|X| - 2)}\).
  The entire computation at the join node can be done in time

  \begin{align*}
    \sum_{|X| = 0}^{|X_{t}|}\binom{|X_{t}|}{|X|} 4^{|X|} (2^{(2|X| - 2)} + \OhStar{2^{(2|X| - 2)} \cdot 2^{(\omega - 1)\cdot |X|}})
    &= \OhStar{\sum_{|X| = 0}^{tw + 1}\binom{tw + 1}{|X|} 2^{4|X| - 2 + \omega{}|X| - |X|}}\\
    &= \OhStar{\sum_{|X| = 0}^{tw + 1}\binom{tw + 1}{|X|} 2^{(\omega + 3)|X|}}\\
          &= \OhStar{(1 + 2^{(\omega + 3)})^{(tw + 1)}}\\
          &= \OhStar{(1 + 2^{(\omega + 3)})^{tw}}.
  \end{align*}

  The entire DP over \TT can thus be done in \(\OhStar{(1 + 2^{(\omega +
      3)})^{tw}}\) time.
\end{proof}
